\documentclass[letterpaper,twocolumn,10pt]{article}
\usepackage{usenix2019_v3}

\newcommand{\citet}{\cite}
\usepackage{amsfonts}
\usepackage{amsmath}
\usepackage{color}
\usepackage{tikz}
\usetikzlibrary{shapes.geometric} 
\usetikzlibrary{shapes, arrows, positioning}
\usepackage{amsthm}
\usepackage{thmtools}
\usepackage{mathtools}
\usepackage{stmaryrd}
\usepackage{amssymb}
\usepackage{cite}

\usepackage{url}
\usepackage{subcaption}

\usepackage{listings}
\usepackage{caption}
\usepackage{cleveref}
\usepackage{xspace}
\usepackage{tabularx}
\usepackage{makecell}
\usepackage{multirow}
\usepackage{inconsolata}

\usepackage{mathrsfs}
\usepackage{algorithm}
\usepackage{algorithmic}
\usepackage{svg}
\usepackage{enumitem}

\usepackage[inference]{semantic}

\theoremstyle{plain}
\newtheorem{theorem}{Theorem}
\newtheorem{lemma}{Lemma}
\newtheorem{definition}{Definition}

\newcommand{\openssl}{OpenSSL\xspace}

\newcommand{\LightSLH}{LightSLH\xspace}

\newcommand{\Nat}{\mathbb{N}}
\newcommand{\Int}{\mathbb{Z}}

\newcommand{\Reg}{\mathit{Regs}}
\newcommand{\Val}{\mathit{Vals}}

\newcommand{\Obs}{\mathtt{Obs}}
\newcommand{\Dir}{\mathtt{Dir}}

\newcommand{\kywd}[1]{\mathbf{#1}}

\newcommand{\storeKywd}{\kywd{store}}
\newcommand{\loadKywd}{\kywd{load}}
\newcommand{\jmpKywd}{\kywd{jmp}}
\newcommand{\jzKywd}{\kywd{beqz}}
\newcommand{\pcKywd}{\kywd{pc}}
\newcommand{\fenceKywd}{\kywd{fence}}
\newcommand{\allocKywd}{\kywd{alloc}}

\newcommand{\Obsload}[3]{{\mathtt{load }\ \ensuremath{#1_{[#2,#3]}}}}
\newcommand{\Obsstore}[3]{{\mathtt{store }\ \ensuremath{#1_{[#2,#3]}}}}
\newcommand{\Obsbranch}[1]{{\mathtt{branch }\ \ensuremath{#1}}}
\newcommand{\Obsnone}{{\ensuremath{\epsilon}}}

\newcommand{\DirStep}{\mathtt{step}}
\newcommand{\DirForce}{\mathtt{force}}

\newcommand{\pc}{\pcKywd}
\newcommand{\mem}{\kywd{mem}}

\newcommand{\op}{\ensuremath{\odot}\xspace}
\newcommand{\uop}{\ensuremath{\ominus}\xspace}
\newcommand{\bop}{\ensuremath{\otimes}\xspace}
\newcommand{\addop}{\texttt{Add}\xspace}
\newcommand{\minusop}{\texttt{Minus}\xspace}
\newcommand{\mulop}{\texttt{Mul}\xspace}
\newcommand{\divop}{\texttt{Div}\xspace}
\newcommand{\modop}{\texttt{Mod}\xspace}
\newcommand{\andop}{\texttt{And}\xspace}
\newcommand{\orop}{\texttt{Or}\xspace}
\newcommand{\xorop}{\texttt{Xor}\xspace}
\newcommand{\lshiftop}{\texttt{Shl}\xspace}
\newcommand{\rshlop}{\texttt{Lshr}\xspace}
\newcommand{\rshaop}{\texttt{Ashr}\xspace}
\newcommand{\notop}{\texttt{Not}\xspace}

\newcommand{\passign}[2]{#1 \leftarrow #2}
\newcommand{\pload}[2]{\loadKywd\ #1, #2}
\newcommand{\pstore}[2]{\storeKywd\ #1, #2}
\newcommand{\pcondassign}[3]{#1 \xleftarrow{#3?} #2}
\newcommand{\pjmp}[1]{\jmpKywd\ #1}
\newcommand{\pbranch}[2]{\jzKywd\ #1, #2}
\newcommand{\pfence}{\fenceKywd}
\newcommand{\palloc}[2]{#1 \leftarrow \allocKywd\ #2} 

\newcommand{\muasm}{$\mu$\textsc{Asm}\xspace}
\newcommand{\exprEval}[2]{\ensuremath{\llbracket {#1} \rrbracket_{#2}}}

\newcommand{\conf}[1]{\langle #1 \rangle}

\newcounter{typerule}
\crefname{typerule}{}{}

\newcommand{\typeruleInt}[5]{%
	\def\thetyperule{\textsc{#1}}%
	\refstepcounter{typerule}%
	\label{tr:#4}%
  \ensuremath{\begin{array}{c}#5 \inference{#2}{#3}\end{array}}%
}
\newcommand{\typerule}[4]{%
  \typeruleInt{#1}{#2}{#3}{#4}{\text{\scriptsize[\textsc{#1}]}\\}%
}

\newcommand{\eval}[2]{\xrightarrow[{#1}]{{#2}}}

\newcommand{\updatemap}[3]{\ensuremath{#1[#2 \mapsto #3]}}
\newcommand{\trace}[4]{\ensuremath{#1(#2)\mathord\Downarrow^{#3}_{#4}}}
\newcommand{\seqtrace}[3]{\ensuremath{#1(#2)\mathord\Downarrow_{#3}}}
\newcommand{\policyeq}[3]{\ensuremath{#1\backsim_{#3} #2}}
\newcommand{\tainteq}[3]{\ensuremath{#1\backsim_{#3} #2}}
\newcommand{\num}{\text{num}}
\newcommand{\rev}[1]{\ensuremath{\text{rev}(#1)}}

\newcommand{\sni}[2]{#1 \vdash_#2 \text{SNI}}
\newcommand{\ssafety}[2]{#1 \vdash_#2 \text{SS}}

\newcommand{\tbot}{{\ensuremath{\bot}}\xspace}
\newcommand{\tzero}{\texttt{0}\xspace}
\newcommand{\tone}{\texttt{1}\xspace}
\newcommand{\tlow}{{\texttt{L}}\xspace}
\newcommand{\thigh}{{\texttt{H}}\xspace}

\newcommand{\abdpl}[2]{\ensuremath{\mathcal{#1}_{ #2}^{\sharp}}}
\newcommand{\lattice}[1]{\ensuremath{\mathcal{#1}}}
\newcommand{\prolattice}[2]{\ensuremath{\lattice{#1}_{#2}}}
\newcommand{\parlattice}[2]{\ensuremath{\langle{#1},{#2}\rangle}}
\newcommand{\simplevector}[4]{\ensuremath{({#1}_{#2},\allowbreak{#1}_{#3},\allowbreak\cdots,\allowbreak{#1}_{#4})}}

\newcommand{\uopondomain}[3]{\ensuremath{#1_{#2}({#3})}}
\newcommand{\bopondomain}[4]{\ensuremath{#1_{#2}({#3, #4})}}
\newcommand{\counttupleup}[2]{\ensuremath{\texttt{Cnt}_{#1\sqsubseteq}{#2}}}

\newcommand{\minintv}[1]{\ensuremath{\min_{#1 \sqsubseteq }}}
\newcommand{\opondomain}[2]{\ensuremath{#1_{#2}}\xspace}

\newcommand{\isinstance}[2]{\ensuremath{{#1}\vdash{#2}}}

\newcommand{\consn}{\ensuremath{\mathtt{n}}\xspace}
\newcommand{\abd}[1]{\ensuremath{#1^\sharp}}

\newcommand{\base}{\kywd{Base}}
\newcommand{\basemap}{\ensuremath{B}\xspace}

\newcommand{\di}{\ensuremath{\mathcal{DI}}\xspace}
\newcommand{\interval}{\ensuremath{\mathcal{I}}\xspace}
\newcommand{\diorder}{\ensuremath{\sqsubseteq}\xspace}

\newcommand{\powerset}[1]{\ensuremath{\mathcal{P}({#1})}\xspace}
\newcommand{\vd}{\ensuremath{\mathcal{\nu }}\xspace}
\newcommand{\VD}{\ensuremath{\mathcal{V}}\xspace}

\newcommand{\memory}{\ensuremath{\mathcal{M}}\xspace}
\newcommand{\memmap}{\ensuremath{\mathcal{M}_\mathcal{R}}\xspace}
\newcommand{\sizemap}{\ensuremath{\mathcal{M}_\mathcal{S}}\xspace}

\newcommand{\topof}[1]{\ensuremath{\top_{#1}}\xspace}

\newcommand{\abf}[1]{\ensuremath{\alpha^{#1}}\xspace}
\newcommand{\cof}[1]{\ensuremath{\gamma^{#1}}\xspace}
\newcommand{\conff}[2]{\ensuremath{\gamma^{#1}_{#2}}\xspace}
\newcommand{\abff}[2]{\ensuremath{\alpha^{#1}_{#2}}\xspace}

\newcommand{\statecell}[2]{{#1},\ {#2}\ }

\newcommand{\hardened}[2]{\ensuremath{\mathtt{L}_{#2}({#1})}\xspace}
\newcommand{\hardenp}{\hardened{p}{P}}

\lstdefinestyle{intro}{
    language=c,
    basicstyle=\ttfamily,
    escapechar=|,
    frame=single
}
\usepackage{ifthen}
\usepackage[normalem]{ulem} 

\newboolean{showdeleted}
\setboolean{showdeleted}{false} 

\newboolean{showadded}
\setboolean{showadded}{true} 

\newif\ifshowdeleted

\newif\ifshowredeleted

\newcommand{\deleted}[1]{%
  \ifshowdeleted
    \textcolor{blue}{#1}%
  \else
    \kern-\fontdimen2\font 
  \fi
}

\newcommand{\redeleted}[1]{%
  \ifshowredeleted
    \textcolor{blue}{#1}%
  \else
    \kern-\fontdimen2\font 
  \fi
}

\newcommand{\curdeleted}[1]{%
  \ifshowredeleted
  \kern-\fontdimen2\font 
  \else
    \ifshowdeleted
    \textcolor{blue}{#1}%
    \else
    \kern-\fontdimen2\font 
    \fi
  \fi
}

\newcommand{\added}[1]{%
  \ifshowdeleted
    \textcolor{red}{#1}%
  \else
    #1
  \fi
}

\newcommand{\readded}[1]{%
  \ifshowredeleted
    \textcolor{red}{#1}%
  \else
    #1
  \fi
}

\newcommand{\commentcolor}{\color{gray}}

\lstdefinestyle{asmstyle}{
    language=[x86masm]Assembler, 
    basicstyle=\ttfamily\footnotesize,        
    commentstyle=\color{gray},  
    numberstyle=\tiny\color{gray},
    stepnumber=1,                
    breaklines=true
}

\lstdefinestyle{Cstyle}
{
    xleftmargin=2em,frame=single,framexleftmargin=1.5em,
	frame = tb,
    showstringspaces = false,
    breaklines = true,
    breakatwhitespace = false,
    tabsize = 3,
    numbers = left,
    stepnumber = 1,
    numberstyle = \tiny\color{gray},
    language = {[ANSI]C},
    alsoletter={.\$},
    keywordstyle=\bfseries\ttfamily\color{black}, %
    basicstyle=\ttfamily\color{black},
    morekeywords={stype,alloc},
    stringstyle=\color{orange},
    basicstyle=\small,
}

\lstdefinestyle{Cstyle2}
{
    xleftmargin=2em,frame=single,framexleftmargin=1.5em,
	frame = tb,
    showstringspaces = false,
    breaklines = true,
    breakatwhitespace = false,
    tabsize = 3,
    numbers = left,
    stepnumber = 1,
    numberstyle = \tiny\color{gray},
    language = {[ANSI]C},
    alsoletter={.\$},
    keywordstyle=\bfseries\ttfamily\color{black}, %
    basicstyle=\ttfamily\color{black},
    morekeywords={stype,alloc},
    keywordstyle=[2]{\ttfamily\underline}, 
    morekeywords = [2]{key, bits, width, window, wvalue, k},
    stringstyle=\color{orange},
    basicstyle=\small,
}

\lstdefinestyle{MUASMstyle}
{
    xleftmargin=2em,frame=single,framexleftmargin=1.5em,
	frame = tb,
  	showstringspaces = false,
  	breaklines = true,
  	breakatwhitespace = true,
  	tabsize = 3,
  	numbers = left,
    stepnumber = 1,
    numberstyle = \tiny\color{gray},
    alsoletter={.\$\%},
    basicstyle={\footnotesize\ttfamily\color{black}},
    keywordstyle={\footnotesize\ttfamily\color{Blue3}},
}

\newcommand{\musk}{\ensuremath{\mu_{s_k}}\xspace}
\newcommand{\muskone}{\ensuremath{\mu_{s_{k+1}}}\xspace}
\newcommand{\muspk}{\ensuremath{\mu_{s'_k}}\xspace}
\newcommand{\muspkone}{\ensuremath{\mu_{s'_{k+1}}}\xspace}

\newcommand{\rhosk}{\ensuremath{\rho_{s_k}}\xspace}
\newcommand{\rhoskone}{\ensuremath{\rho_{s_{k+1}}}\xspace}
\newcommand{\rhospk}{\ensuremath{\rho_{s'_k}}\xspace}
\newcommand{\rhospkone}{\ensuremath{\rho_{s'_{k+1}}}\xspace}

\newcommand{\minI}{\ensuremath{\interval_{\min}}\xspace}
\newcommand{\maxI}{\ensuremath{\interval_{\max}}\xspace}
\newcommand{\condm}{\ensuremath{\powerset{\integer}}\xspace}
\newcommand{\integer}{\ensuremath{\mathbb{Z}}\xspace}
\newcommand{\indi}[2]{\ensuremath{{#1} \vdash {#2}}\xspace}

\newcommand{\td}{\ensuremath{\abd{\prolattice{T}{\consn}}}\xspace}
\newcommand{\abrho}{\ensuremath{\abd{\rho}}\xspace}
\newcommand{\abs}{\ensuremath{\abd{s}}\xspace}
\newcommand{\abo}{\ensuremath{\abd{o}}\xspace}
\newcommand{\abS}{\ensuremath{\Omega}\xspace}

\newcommand{\abt}{\ensuremath{\abd{t}}\xspace}
\newcommand{\abmu}{\ensuremath{\abd{\mu}}\xspace}
\newcommand{\tmem}{\ensuremath{\memory^{\abd{\prolattice{T}{\consn}}}}\xspace}
\newcommand{\vmem}{\ensuremath{\memory^{\VD}}\xspace}
\newcommand{\seqtran}[3]{\ensuremath{#1\xRightarrow{{#3}}#2}\xspace}
\newcommand{\spectran}[3]{\ensuremath{#1\xrightarrow[]{{#3}}#2}\xspace}
\newcommand{\switchtran}[3]{\ensuremath{#1\xhookrightarrow[]{{#3}}#2}\xspace}
\newcommand{\abconf}[1]{\ensuremath{\conf{\abrho_{#1},\abmu_{#1}, \vmem_{#1}, \tmem_{#1}}}\xspace}

\newcommand{\abveclow}{\ensuremath{\ensuremath{\abd{\vec{\tlow}}}}\xspace}
\newcommand{\abhigh}{\ensuremath{\ensuremath{\abd{\thigh}}}\xspace}
\newcommand{\abvechigh}{\ensuremath{\ensuremath{\abd{\vec{\thigh}}}}\xspace}
\newcommand{\abload}[1]{\ensuremath{\mathcal{L}_{#1}}\xspace}
\newcommand{\abstore}[1]{\ensuremath{\mathcal{S}_{#1}}\xspace}

\newcommand{\emptysym}{\ensuremath{\varepsilon}\xspace}

\newcommand{\tdn}{\ensuremath{\prolattice{T}{\consn}}\xspace}
\newcommand{\ctd}{\ensuremath{\powerset{\prolattice{T}{\consn}}}\xspace}
\newcommand{\atd}{\td}
\newcommand{\ctde}{\ensuremath{T}\xspace}
\newcommand{\atde}{\ensuremath{\abd{t}}\xspace}

\newcommand{\cvd}{\ensuremath{\powerset{\integer}}\xspace}
\newcommand{\avd}{\VD}
\newcommand{\stateset}{\ensuremath{\mathcal{S}}\xspace}
\newcommand{\cstate}{\ensuremath{\powerset{\mathcal{S}}}\xspace}
\newcommand{\aconf}{\ensuremath{\abd{\mathcal{S}}}\xspace}
\newcommand{\astate}{\ensuremath{\varOmega }\xspace}
\newcommand{\cob}{\ensuremath{\mathcal{O}}\xspace}
\newcommand{\aob}{\ensuremath{\abd{\mathcal{O}}}\xspace}

\newcommand{\ubound}[1]{\ensuremath{\sqcup_{#1}}\xspace}
\newcommand{\lbound}[1]{\ensuremath{\sqcap_{#1}}\xspace}
\newcommand{\botof}[1]{\ensuremath{\bot_{#1}}\xspace}
\newcommand{\order}[1]{\ensuremath{\sqsubseteq_{#1}}\xspace}

\newcommand{\initmem}{\ensuremath{M_{\text{init}}}\xspace}
\newcommand{\allinterger}{\ensuremath{\updatemap{\bot_\VD}{\emptysym}{\topof{\interval}}}\xspace}

\newcommand{\abaddr}[1]{\ensuremath{\Gamma_{#1}}\xspace}

\newcommand{\conctrace}{\ensuremath{\tau=(p, s_1)\xrightarrow[d_1]{o_1} (p, s_2)\cdots \xrightarrow[d_{n-1}]{o_{n-1}}(p, s_n)}\xspace}
\newcommand{\conctraceseq}{\ensuremath{\tau=(p, s_1)\xrightarrow[]{o_1} (p, s_2)\cdots \xrightarrow[]{o_{n-1}}(p, s_n)}\xspace}

\newcommand{\abconftrace}{\ensuremath{\abd{\tau}=(p, \abd{s_1})\xrightarrow[]{\abd{o_1}} (p, \abd{s_2})\cdots \xrightarrow[]{\abd{o_{n-1}}}(p, \abd{s_n})}\xspace}
\newcommand{\abstatetrace}{\ensuremath{\abd{\Pi }=\abS_1\xrightarrow[]{} \abS_2\cdots }\xspace}

\newcommand{\confvd}{\ensuremath{\conff{\VD}{\tau}}\xspace}
\newcommand{\confs}{\ensuremath{\conff{\aconf}{\tau}}\xspace}

\newcommand{\abfs}{\ensuremath{\abff{\aconf}{\tau}}\xspace}

\newcommand{\absi}{\ensuremath{\abs_{i}}\xspace}
\newcommand{\absoi}{\ensuremath{\abo_{i}}\xspace}
\newcommand{\absione}{\ensuremath{\abs_{i+1}}\xspace}

\newcommand{\statepremise}{\ensuremath{s_i\in\confvd(\abs_i)}\xspace}

\newcommand{\seqfinal}{\ensuremath{\Omega^\text{seq}}\xspace}
\newcommand{\hardenlist}{\ensuremath{K}\xspace}

\newcommand{\currentstate}{\ensuremath{\abS_{\text{current}}}\xspace}

\newcommand{\currenthardenlist}{\ensuremath{K}'\xspace}
\newcommand{\lasthardenlist}{\ensuremath{\ensuremath{K}}\xspace}

\newcommand{\pred}{\textit{Pred}\xspace}
\newcommand{\fix}{\textit{Fix}\xspace}
\newcommand{\fixspec}{\ensuremath{\fix^{\text{spec}}}\xspace}
\newcommand{\fixseq}{\ensuremath{\fix^{\text{seq}}}\xspace}
\newcommand{\trans}{\textit{Trans}\xspace}
\newcommand{\hardenset}{\ensuremath{H}\xspace}
\newcommand{\hac}{\ensuremath{\finalharden^{\text{ac}}}\xspace}
\newcommand{\absprime}{\ensuremath{\abd{s'}}\xspace}
\newcommand{\finalharden}{\ensuremath{\mathcal{H}}\xspace}

\begin{document}
\pagestyle{empty}

\date{}

\title{\Large \bf Place Protections at the Right Place: Targeted Hardening for Cryptographic Code against Spectre v1}

\author{
{\rm Yiming Zhu, Wenchao Huang\thanks{Corresponding Authors}, Yan Xiong}\\
University of Science and Technology of China
} 

\maketitle

\begin{abstract}

Spectre v1 attacks pose a substantial threat to security-critical software, particularly cryptographic implementations.
Existing software mitigations, however, often introduce excessive overhead by indiscriminately hardening instructions without assessing their vulnerability.     
We propose an analysis framework that employs a novel fixpoint algorithm to detect Spectre vulnerabilities and apply targeted hardening.
The fixpoint algorithm accounts for program behavior changes induced by stepwise hardening, enabling precise, sound and efficient vulnerability detection.
This framework also  provides flexibility for diverse hardening strategies and attacker models, enabling customized targeted hardening.
We instantiate the framework as LightSLH, which hardens program with provable security.

    We evaluate LightSLH on cryptographic algorithms from OpenSSL, Libsodium, NaCL \readded{and PQClean}. 
    Across all experimental cases, LightSLH provides the lowest overhead among current provable protection strategies, including 0\% overhead in 50\% cases.
    Notably, the analysis of LightSLH reveals two previously unknown security issues:
    (1) The compiler can introduce risks overlooked by LLSCT, a hardening method proven secure at the LLVM IR level. 
    We successfully construct a side channel by exploiting compiler-inserted stack loads, confirming this risk.
    (2) Memory access patterns generated by the scatter-gather algorithm still depend on secrets, even for observers with cache line granularity.
    These findings and results highlight the importance of applying accurate protections to specific instructions.

\end{abstract}


\section{Introduction}\label{sec:intro}


Microarchitectural attacks \cite{survey-side-channel2019} have emerged as a critical security concern for programs handling sensitive data, particularly cryptographic software. 
These attacks capitalize on the ability of specific instructions to alter the internal state of the processor's microarchitecture, potentially leading to the leakage of sensitive information.
Mitigation strategies involve either excluding sensitive data from the operands of certain instructions (adhering to the constant-time principle \cite{ct-verif2015}) or employing techniques like scatter-gather \cite{cache-audit-rigorous2017} that ensure state modifications are independent of secret data.


The disclosure of Spectre attacks \cite{spectre2019,spectre-overflow2018,ret2spec2018,sptctre-return2018,horn2018speculative} has considerably broadened the attack surface for microarchitectural attacks. 
Spectre attacks exploit the fact that speculative execution can modify the processor's internal state even for instructions ultimately discarded. 
While mitigations have been developed for Spectre v2-5 through hardware \cite{intelspectremitigation, intel:ssb, intel:psfd, amd:psfd, intel:cet}, operating systems \cite{intel:bhi,intel:mds}, or software  \cite{retpoline:skylake, intelspectremitigation}, effectively defending against Spectre v1 remains an ongoing challenge.



While various methods have been proposed to defend against Spectre v1 attacks, they suffer from three limitations.
First, software mitigations \cite{declassified2023,declassiflow2023,serberus2024,llvm:slh,uslh2023,exorcise2021} often indiscriminately harden instructions without assessing their vulnerability, leading to unnecessary performance overhead.
Second, while detection tools based on dynamic analysis \cite{spectaint2021,kasper2022,SpecFuzz2020} and gadget scanning \cite{kleespectre} can analyze entire programs and identify potential vulnerabilities, they lack precision and security guarantees due to false positives and negatives.
Third, sound detection tools \cite{pitchfork2020,hunter2021,spectector2020,cats2022,typing2023,jasmin-spectre} can only determine whether a program is secure, instead of identifying all vulnerable instructions.

To minimize the overhead of Spectre v1 defenses while maintaining security, a natural approach combines sound analysis (ensuring security) with precise detection (minimizing unnecessary hardening) to identify vulnerabilities for targeted protection. However, this approach faces two key challenges.

\textit{Challenge 1}: Analysis paralysis issue.
Sound program analysis is challenged by speculative out-of-bounds stores during speculative execution, which significantly impact analysis precision. 
Specifically, an out-of-bounds memory store can invalidate the entire memory state, hampering the detection of other vulnerabilities. 
While hardening these stores effectively resolves the memory state invalidation, the resulting changes in program behavior necessitate re-analysis. 
This iterative process of hardening and re-analysis introduces considerable computational overhead.



\textit{Challenge 2}: Pinpointing vulnerable instructions.
Identifying vulnerable instructions generally requires analyzing whether \textit{specific bits} in the operands of certain instructions contain sensitive information according to the attacker models.
While most existing analysis methods can detect whether a memory access address is sensitive assuming a full-address observer, they fail to assess the sensitivity of the \textit{cache-line index}, a more realistic scenario for cache-based attacks. 
This limitation hinders the analysis of defenses specifically designed to mitigate cache-based attacks, such as scatter-gather techniques.

To tackle challenge 1, we propose an analysis framework called Light Hardening (LightH), which employs a fixpoint algorithm to account for the program behavior changes after hardening.
Based on the observation that memory accesses for hardened instructions are blocked during misspeculative execution, LightH directly utilizes the analysis results of \textit{sequential execution} to compute the program state during \textit{speculative execution}  when analyzing the behavior of hardened instructions.
Such utilization allows LightH to continue analyzing the program post-hardening, building upon the results from the previous iteration rather than restarting the analysis. 
Additionally, LightH offers a flexible interface for various hardening strategies and attacker models.
By independently implementing the interface without modifying any other modules, LightH can be adapted to identify vulnerable instructions under various attacker models and hardening strategies, making it effective in addressing challenge 2.

To demonstrate the feasibility of our framework, we provide LightSLH, an instantiation of LightH based on the SLH hardening strategy and the cache-based attacker model.
LightSLH employs a bit-level taint mechanism to trace sensitive data precisely. 
We establish a formalization of this mechanism and provide comprehensive taint rules based on the formalization.
Building upon this foundation, we extend the program semantics with the taint mechanism and prove that programs protected by LightSLH satisfy \textit{speculative safety} (SS).
SS essentially approximates the stricter \textit{speculative non-interference} (SNI), ensuring that the programs do not leak more information during speculative execution than during sequential execution.

We evaluate LightSLH on 14 cryptographic algorithms from OpenSSL, Libsodium, NaCL, and PQClean.
LightSLH provides the lowest overhead among current protection strategies that offer provable security, across all cases.
Compared to the conservative SSLH method, LightSLH reduces overhead by \redeleted{68.6\%} \readded{48.1\%} on average. 
Notably, in 50\% of cases, LightSLH introduces zero overhead.
Additionally, to demonstrate the flexibility and generality of LightH, we also implement LightFence \readded{and LightCut} based on fence strategy. 
\redeleted{Compared to the fence method, LightFence reduces average overhead by 74.5\%.
}

As a highlight, through analyzing the results of LightSLH, we discover two previously unknown security issues.
The first issue arises from the potential for the compiler  to introduce new speculative vulnerabilities when hardening LLVM IR codes with LLSCT strategy \cite{serberus2024}.
While LLSCT has been proven secure at the LLVM IR level, it may overlook speculative vulnerabilities introduced by the compiler when it inserts loads from stack during the transformation from LLVM IR to assembly. 
Exploiting this compiler behavior, we successfully construct a one-bit side channel with 83\% accuracy, confirming the risk. 
In contrast, LightSLH's security guarantees remain unaffected by this issue.
Second, LightSLH performs the first rigorous analysis of the security guarantees of RSA against Spectre v1. 
The analysis reveals that even for observers at the cache line granularity, the memory access patterns generated by the scatter-gather algorithm still depend on secrets, necessitating protection for such accesses.

In summary, our contributions are as follows:
\begin{itemize}[left = 0pt]
    \item We introduce LightH (\Cref{sec:overview}), an framework which employs a novel fixpoint algorithm to apply targeted hardening against Spectre v1.
    LightH offers flexibility in identifying vulnerable instructions according to various attacker models and  different hardening strategies.
    \item We instantiate LightH as LightSLH based on cache-based attacker and SLH hardening.
    LightSLH employs a formalized bit-level taint mechanism (\Cref{sec:taint-tracking}) to track sensitive data.
    We prove that programs hardened by LightSLH satisfy SS (\Cref{sec:LightSLH}), a property that safely approximates SNI (\Cref{sec:speculative-safety}).
    \readded{Based on the fence strategy, we also implement LightFence and LightCut to show the flexibility of our framework.}
    \item We implement and evaluate LightSLH on 12 cryptographic algorithms from OpenSSL, Libsodium,  NaCL, and PQClean (\Cref{sec:evaluation}).
    Our results demonstrate both the efficiency and accuracy of LightSLH's analysis, along with the lower overhead of its hardening.
    To further showcase the flexibility and generality of LightH, we additionally implement and evaluate LightFence.
    \item Through our experiments, we discover two previously unknown security issues. 
\end{itemize}


\section{Background}\label{sec:background}

\subsection{Spectre v1 Attack}\label{sec:spectre-attack}
To avoid hazards triggered by control-flow dependency, modern processors speculatively fetch or execute some instructions based on prediction mechanisms.
However, during \textit{speculative execution}, the microarchitectural state is modified, which can leak data that remains protected during \textit{sequential execution}. 
We refer to \textit{misspeculative execution} as speculative execution based on incorrect predictions. 



\begin{figure}
    \centering
    \begin{minipage}{0.40\linewidth}
        \begin{lstlisting}[basicstyle=\footnotesize,style=Cstyle, escapechar=&, captionpos=t,caption = {Spectre v1.}, label={lst:spectre-v1}]
x = 0;
if(x != 0) {   &\label{line:branch}&
    y = *secret; &\label{line:access}&
}
\end{lstlisting}   
    \end{minipage}
\begin{minipage}{0.55\linewidth}
    \begin{lstlisting}[basicstyle=\footnotesize,style=Cstyle, escapechar=&, captionpos=t,
label = {lst:slh}, caption = {Example of SLH.}]
x = 0;
if(x != 0) {
    mask = (x != 0) ? 0 : -1;
    y = *(secret | mask); &\label{line:mask}&
}
\end{lstlisting}  
\end{minipage}
\end{figure}

Our work focuses on addressing Spectre v1 attacks, which exploit speculative execution caused by branch prediction. 
\Cref{lst:spectre-v1} shows a typical code gadget vulnerable to Spectre v1.
\Cref{line:access} is never executed during sequential execution, since the branch condition at \Cref{line:branch} always evaluates to false. 
However, speculative execution might fetch and execute \Cref{line:access} based on branch prediction.
This speculative execution can leak the value of the secret variable through a cache side channel.
A security property called \emph{speculative non-interference} (SNI) \cite{declassified2023,specttre-sok2022,spectector2020} prevents such leaks by ensuring speculative execution reveals no more information than sequential execution.
We define SNI formally in \Cref{sec:operational-semantics}.

\textbf{Mitigations for Spectre v1.}\label{sec:mitigation-spectre}
A basic defense against Spectre v1 is to prevent processors from executing speculatively by inserting serializing instructions, such as \texttt{LFENCE} \cite{intel:lfence}, at every branch instruction.
\texttt{LFENCE} ensures that it will not execute until all prior instructions complete, and no following instruction will execute until \texttt{LFENCE} completes \cite{intel:reference}. 
Speculative-load-hardening (SLH) \cite{llvm:slh} offers an alternative approach, as shown in \Cref{lst:slh}. It uses a speculative flag (\texttt{mask}) that becomes -1 during misspeculative execution and 0 otherwise. By performing \orop operations between this flag and load operands (\Cref{line:mask}), SLH blocks speculative loads by forcing invalid addresses (-1). However, SLH applies protection indiscriminately and lacks formal security guarantees.
As an extension of SLH, SSLH \cite{exorcise2021} provides a provable protection by hardening all load, store and branch instructions.

\subsection{Monotone Frameworks}\label{sec:abstraction-introduction}


The monotone framework \cite{kam1977monotone} is a foundational approach for modeling and analyzing how program properties evolve during execution. In this framework, program configurations are represented as elements of a lattice. 
The program configuration, denoted by \abS, is a mapping from instructions to the domain of program states. 
Here, $\abS(n)$ denotes the program state before executing the instruction $p(n)$. 
Information about the program's behavior is propagated through a monotonic  function $\textsc{Trans}$ as follows:
$$\abS^{\text{new}}(n) = \bigsqcup\limits_{n'\in\textsc{Pred}(n)}\textsc{Trans}(\abS^{\text{old}}(n'),n)$$
Here, $\textsc{Pred}(n)$ denotes the set of predecessors of the instruction $p(n)$.
By designing the domain of program states, the framework can effectively utilize techniques such as abstract interpretation to capture and reason about the properties of the program in an abstract domain. 
This allows for the derivation of useful information regarding program behavior while maintaining efficiency in analysis.

\section{Challenges and Overview}\label{sec:mot}

In this section, we present two examples to illustrate the key challenges in accurately identifying which instructions are vulnerable to Spectre v1 attacks.
Following this, we provide an overview of our methodology.

\subsection{Analysis Paralysis}\label{sec:mov-paralysis}

\begin{lstlisting}[style=Cstyle, escapechar=|, captionpos=t,
    caption = {Speculative Out-of-Bounds Access}, label = {lst:analysis-paralysis}]
b[0] = 0;
if (x < 16)    |\label{lst:check}|
    a[x] = key; |\label{lst:store}|
z = b[0];
w = b[z];       |\label{lst:w}|
\end{lstlisting}     


The first challenge stems from out-of-bounds stores caused by Spectre vulnerabilities, which can lead to analysis paralysis and subsequently result in redundant hardening.
Consider the analysis and protection for \Cref*{lst:analysis-paralysis}.
Let \texttt{x} be a positive number, \texttt{a} be an array with  size  $16$ and \texttt{key} be sensitive data.
The boundary check in \Cref{lst:check} ensures that \Cref{lst:store} performs an in-bounds store, preventing data outside of the array \texttt{a} from being affected. 
During sequential execution, \texttt{z} is definitively assigned the value $0$.
Notably, during speculative execution, \Cref{lst:store} may perform an out-of-bounds store,  rendering it vulnerable to Spectre v1.1 attacks and necessitating hardening.
This out-of-bounds store introduces the possibility of writing \texttt{key} to any memory region, meaning that \texttt{z} could contain sensitive data.
Consequently, \texttt{key} becomes a potential value for \texttt{z}, requiring the load operation in \Cref{lst:w} to be hardened.

However, once \Cref{lst:store} is hardened, the program's behavior changes, as the memory access is blocked during misspeculative execution, and \texttt{z} can only be assigned a value of 0, even during speculative execution in the hardened program.
Therefore, hardening \Cref{lst:w} becomes redundant after hardening \Cref{lst:store}, while existing mitigation strategies \cite{llvm:slh,uslh2023,exorcise2021,declassified2023,declassiflow2023} would still apply such redundant protection.
As illustrated, an analysis that accounts for the impact of protections on program behavior can prevent analysis paralysis caused by Spectre vulnerabilities and reduce redundant protections, compared to an analysis that ignores such effects.

A straightforward approach to account for program behavior changes caused by hardening is to pause whenever a vulnerable instruction is identified, apply the necessary defenses, and restart the detection process.
However, for existing tools based on symbolic execution \cite{spectector2020, pitchfork2020}, this approach requires solving symbolic constraints for each memory access instruction, followed by regenerating the constraints and re-executing the symbolic analysis after each hardening.
For tools based on type systems \cite{typing2023} or model checking \cite{cats2022}, the entire program requires re-analysis after applying defenses, as security types are modified and models need reconstruction.
The above methods involve repeated analyses, significantly impacting efficiency.


\subsection{Cache-Aware Implementations}\label{sec:mov-tt}
\begin{lstlisting}[basicstyle=\small,style=Cstyle, escapechar=|, captionpos=t,
    caption = {Scatter-gather method from \openssl 1.0.2f.}, label = {lst:scatter-gather}]
align (char* buf ){ |\label{lst:sg-align}|
    return buf - ( buf & ( block_size - 1 )) + block_size;
}
scatter (char* buf, char* p, int k, int window ){ |\label{lst:sg-scatter}|
    for ( i = 0; i < N; i++)
        buf[k + i * window] = p[i];
}
gather ( char* buf, char* p, int k, int window ){ |\label{lst:sg-gather}|
    for ( i = 0; i < N; i++)
        p[i] = buf[k + i * window];
}
\end{lstlisting}   
Another challenge is determining which instructions need to be marked for protection, which is non-trivial, especially when dealing with cache-aware code implementations.
Such implementation aims to avoid secret-dependent memory access at the cache line level.
A notable example is the \openssl 1.0.2f's implementation of RSA  \cite{cache-audit-rigorous2017, cachebleed2017}. 

In \openssl, a  strategy to reduce the performance overhead of computing the powers of a large number $a$ is to precompute and store these powers in a lookup table \texttt{buf}. 
This allows direct retrieval of $a^k$ from \texttt{buf[$k$]}, simplifying the computation.
However, such memory access patterns (\texttt{buf[$k$]}) can inadvertently leak the value of $k$ through cache side channels. 
To mitigate this vulnerability, \openssl employs a scatter-gather technique, as shown in \Cref{lst:scatter-gather}.           

It works by first aligning (\Cref{lst:sg-align}) a buffer to a cache line boundary.
Next, it determines  a window size that can be evenly divided by the cache line size and partitions the buffer into distinct indices.
The secret value is then scattered (\Cref{lst:sg-scatter}) across different indices, e.g., the \texttt{i}-th byte of \texttt{p} is scattered into the \texttt{i}-th indices of \texttt{buf}.
During retrieval (\Cref{lst:sg-gather}), the indices are accessed in sequential order to reconstruct the complete value of \texttt{p}, independently of the specific secret value \texttt{k}.
Such access pattern is used to prevent leakage of the secret value \texttt{k} through timing analysis.
This example illustrates the necessity of accurately tracking which bits of a value might expose sensitive information when analyzing whether an instruction is at risk of side-channel attacks.

To our knowledge, \cite{cache-audit-rigorous2017} is the only work that has conducted a rigorous analysis of scatter-gather implementations. 
It employs a custom masked symbol domain to estimate the number of possible memory access patterns. 
However, it focuses exclusively on code gadgets related to scatter-gather implementations without examining the broader context in which these gadgets operate.
Other approaches, such as symbolic execution \cite{cached-2017} and software verification \cite{jasmin-ct}, while capable of analyzing cache-aware code, do not analyze scatter-gather implementations in RSA due to scalability limitations.
Overall, designing a scalable and sound analysis method for cache-aware implementations remains a significant challenge.

\begin{figure*}[h]
    \centering
    \includegraphics[width=\linewidth]{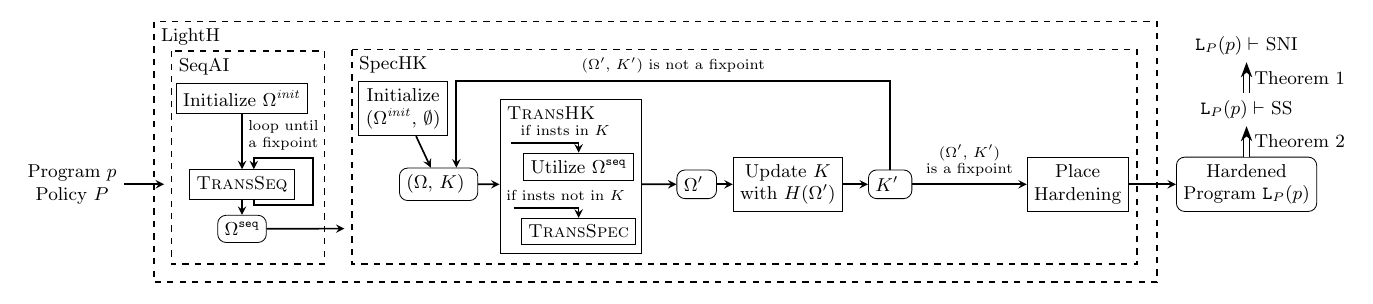}
    \caption{The Overview of LightH.}
    \label{fig:overview}
\end{figure*}

\subsection{Overview}\label{sec:overview}


To address the first challenge, we propose a framework called LightH (Light Hardening). 
To account for the changes in program behavior due to hardening, LightH iteratively identifies the instructions that require protection by employing an analysis that incorporates knowledge of which instructions have been hardened.
This approach allows LightH to continue analyzing the program post-hardening, building upon the results from the previous iteration rather than restarting the analysis. 
For the second challenge, LightH provides a flexible interface for modeling different attackers and  various protection strategies.


The workflow of LightH is illustrated in \Cref{fig:overview}.
LightH takes as input a program $p$ and a security policy $P$. 
Through a two-phase process consisting of sequential abstract interpretation (SeqAI) and speculative analysis with hardening knowledge (SpecHK), LightH produces a hardened program.
We implement LightH as LightSLH, based on the SLH protection strategy and the cache-based attacker.
LightSLH  employs a bit-level taint tracking method to analyze the program.
We also prove that \hardened{p}{P}, the program produced by LightSLH, satisfies SS, a property which soundly approximates SNI.


\textbf{SeqAI.} 
LightH first performs abstract interpretation under sequential semantics (\textsc{TranSeq}) to obtain a sound approximation of the program's behavior under sequential execution, i.e., the fixpoint $\seqfinal$. 
$\seqfinal$ is then utilized to compute the program states of the instructions that need protection in the next phase.
This utilization is based on the observation that memory accesses for hardened instructions are blocked during misspeculative execution.
Consequently, the maximum attainable program state in the domain, after executing these hardened instructions under speculative semantics, is identical to the state under sequential semantics.
Note that SeqAI does not perform vulnerability identification but instead serves as a pre-analysis for the next phase.

\textbf{SpecHK.} 
SpecHK iterates on both the program configuration $\Omega$ and the hardening set $K$, enabling on-the-fly identification and hardening of targeted instructions.
It repeats the following three steps until a fixpoint is reached:
\begin{itemize}
    \item Update $\Omega$ according to current hardening set $K$, and $\seqfinal$, which is produced by SeqAI.
    \item Compute $H(\abS')$ to identify instructions vulnerable to Spectre v1 and therefore require hardening.
    \item Assume the instructions identified in the previous step have been hardened, and incorporate this assumption in subsequent analysis.
\end{itemize}
To account for changes in program behavior after hardening, SpecHK employs \textsc{TransHK}, a transition function with hardening knowledge, to compute the transition of program configurations.
\textsc{TransHK} updates the program configuration $\Omega$ as follows:
for instructions in $K$, which are identified for hardening, SpecHK directly uses $\seqfinal$ to compute $\Omega'$.
For instructions not in $K$, \textsc{TransHK} computes $\Omega'$ using speculative semantics (\textsc{TransSpec}).

SpecHK ultimately converges to a fixpoint $(\Omega', K')$, indicating that once all instructions in $K'$ are hardened, applying the transfer function to  $\Omega'$ will not introduce any changes in $\Omega'$ or $K'$.
This confirms that no additional instructions beyond those in $K'$ require protection.

\textbf{Flexibility.} 
LightH provides flexible support for varying attacker models and protection strategies by offering an interface function $H()$, which updates the hardening set based on the current analysis results.
Given an attacker model, $H()$ determines whether a specific instruction might lead to a Spectre vulnerability by checking if its operands or parts of the operands contain sensitive information.
Additionally, by designing different $H()$ functions to update the hardening set, various protection strategies can be implemented. 
These strategies include applying SLH to identified instructions, inserting fences at the beginning of some blocks, or using previously proposed techniques such as calculating the min-cut \cite{serberus2024} or protection boundaries \cite{declassiflow2023 }.

Based on \redeleted{the SLH and fence} \readded{different} strategies, we implement LightH as \textbf{LightSLH}, \textbf{LightFence} and \textbf{LightCut}, \redeleted{both} \readded{all} of which are built on a cache-line observer model (detailed in \Cref{sec:operational-semantics}).
Based on the experimental results, we ultimately select LightSLH as our primary strategy.
Our subsequent discussion is based on LightSLH's specific $H()$ function.

\textbf{Fine-Grained Security Analysis.}
For performing a rigorous analysis of cache-aware implementations like scatter-gather method, we employ a precise bit-level taint tracking method, allowing us to trace sensitive data at the bit level.
While previous researches \cite{yadegari2014bit,cache-audit-rigorous2017} have explored bit-level taint tracking methods, these methods typically support only simple operations like addition or bitwise operations. 
For more complex operations such as multiplication or shifts, they revert to standard taint tracking, which assigns high or low labels to entire data and thus loses bit-level precision.
Our work distinguishes itself by (1) formally defining the property necessary for bit-level taint tracking, (2) providing a comprehensive tracking scheme that allows the precision of the tracking rules to be adjusted according to the formalization, and (3) enabling the fixpoint computation with guaranteed soundness based on this formalization.

The details of LightH are explained in \Cref{sec:abstract-interpretation}, and the bit-level taint tracking is described in \Cref{sec:taint-tracking}.

\section{Threat Model and Notations}\label{sec:operational-semantics}

\subsection{Threat Model}\label{sec:threat-model}

We establish a threat model where the attacker and the victim co-reside on the same hardware platform. 
This allows the attacker to exploit side channels to observe the victim's program execution. 
Specifically, the attacker can observe the targets of branch instructions and infer memory access addresses by monitoring the cache state. 
While this setting resembles prior research \cite{declassified2023,exorcise2021} on side-channel analysis, the key distinction is that we only allow the attacker to infer memory access addresses at cache line granularity, which enables us to analyze more practical attack scenarios where attackers cannot infer the full address \cite{flush-reload2014,prime-probe2015}.
\readded{Note that in our cache-based model, when the cache size is set to 1, the attacker effectively gains the ability to observe the full memory address.
}

Furthermore, the shared environment allows the attacker to fully control the branch prediction targets. 
This can be achieved through techniques like training the branch predictor. 
Our model incorporates all possible execution traces that can arise due to the influence of branch prediction.

Similar to existing work \cite{declassified2023,exorcise2021}, our attack model takes a conservative stance, aiming to identify potential information leaks within a program.
While such leaks may not immediately facilitate practical attacks, adopting a conservative model is crucial for ensuring robust security assurances in defensive strategies.

\readded{Our goal is to protect programs from leaking sensitive information during speculative execution. Leaks during sequential execution are not in the scope of our protection.
}

\subsection{Language and Security Notations}\label{sec:language}

\begin{figure}
    \small
    \begin{gather*}
        \begin{align*}
            \text{(Register) } x\in &\ \Reg\\
            \text{(Values) }  n, l \in &\ \Val = \Nat \cup \{\bot\}\\
            \text{(Unary Operator) } \ominus \coloneqq &\ \notop \\
            \text{(Binary Operators) } \otimes \coloneqq &\ \addop \mid \minusop \mid \mulop \mid \divop \mid \modop\\
            \mid &\ \andop \mid \orop \mid \xorop \mid \lshiftop \mid \rshlop \mid \rshaop \\
            \text{(Expressions) } e \coloneqq &\ n\mid x \mid \ominus e \mid e_1 \otimes e_2 \\
            \text{(Instructions) } i \coloneqq &\ \passign{x}{e} \mid \pload{x}{e} \mid \pstore{x}{e} \mid \pjmp{l}\\
            \mid &\  \pbranch{x}{l} \mid \pcondassign{x}{e}{e'} \mid \pfence \\
            \text{(Programs) } p  \coloneqq &\ n:i \mid p_1;p_2
        \end{align*}
    \end{gather*}
    
    \caption{Syntax of \muasm.}
    \label{fig:core-language}
\end{figure}

Our methodology is explained through a simplified core language called \muasm,  defined in \Cref{fig:core-language}.
For a program $p$, we denote the instruction labeled with $n$ by $p(n)$.

We use notations and operational semantics with \textit{adversarial directives} similar to  \cite{declassified2023} to model speculative execution of a program.
Specifically, the one-step execution of a program $p$ is modeled using a transition relation of the form $(p, s)\xrightarrow[d]{o} (p, s')$.
Here, $s$ represents a program state and is modeled by a tuple $\conf{\rho, f}$, 
where $\rho$ denotes a mapping from $\Reg \cup \Nat$ to $\integer$, representing the values assigned to registers and memory addresses, and $f$ is the \textit{misspeculative flag}, indicating the current state to be in misspeculative execution when $f$ evaluates to $\top$, and otherwise when $\bot$.
$o\in \Obs$ represents the attacker's observations during execution, and  $d\in \Dir$  models the attacker's ability to fully control the prediction of conditional branches in  threat model. 
$\Obs$ and $\Dir$ are taken from the following syntax:  
\begin{gather*}
    \begin{align*}
        &\Obs \coloneqq \Obsnone \mid \Obsbranch{n} \mid \Obsload{n}{a}{b} \mid \Obsstore{n}{a}{b}\\
        &\Dir \coloneqq \DirStep \mid \DirForce
    \end{align*}
\end{gather*}

Here, $\Obsload{n}{a}{b}$ and $\Obsstore{n}{a}{b}$ denote  observations of memory access: when the program reads from or writes to memory address $n$, the attacker can observe bits $a$ to $b$ of that address. 
Additionally, $\Obsbranch{n}$ denotes the exposure of a conditional branch.
$\epsilon$ means no observation is generated during this specific execution.
Directive $\DirStep$ denotes that the current instruction simply executes as intended, without any prediction or in a right prediction. 
Directive $\DirForce$, on the other hand, represents the attacker's ability to manipulate program execution at a branch instruction, forcing it into misspeculative execution.

The operational semantics is given in \Cref{app:semantics}.
We write \trace{p}{s_1}{D}{O} = $(p, s_1)\xrightarrow[d_1]{o_1} (p, s_2)\cdots$ to represent an execution trace of program $p$ that generates observations $O$ given directives $D$, starting from an initial state $s_1$ where the misspeculative flag is $\bot$.
Here, $D$ is the concatenation of $d_1,d_2\cdots$ and $O$ is the concatenation of $o_1, o_2,\cdots$.
A trace without $\DirForce$ (and thus without misspeculative execution) is referred to as a sequential trace, denoted as \seqtrace{p}{s}{O}.

Speculative leakage is formalized through \textit{SNI} \cite{declassified2023,specttre-sok2022,spectector2020}.
SNI ensures a program's speculative execution reveals no more information than its sequential execution.
A security policy $P \subseteq \Reg \cup \Val$ defines \textit{Public} registers and memory addresses, with states being \textit{equivalent} under $P$ (\policyeq{s}{s'}{P}) if they match on all values in $P$.
A program $p$ satisfies SNI ($\sni{p}{P}$) if equivalent initial states producing identical sequential observations also produce identical speculative observations under any directives $D$. Specifically,

\begin{definition}[SNI]\label{def:speculative-non-interference}
    $\sni{p}{P}$, iff
    for any pair of $ \trace{p}{s_1}{D}{O}$ and $\trace{p}{s'_1}{D}{O'}$, s.t.
    \policyeq{s_1}{s'_1}{P},
    we have
    $\overline{O} = \overline{O}' \Rightarrow O = O'$, 
    where \seqtrace{p}{s_1}{\overline{O}} and \seqtrace{p}{s'_1}{\overline{O}'}.
\end{definition}

\section{Light Hardening}\label{sec:abstract-interpretation}

\subsection{\textsc{TransSeq} and \textsc{TransSpec}}\label{sec:seq-spec-ai}
Based on the monotone framework, the transition function \textsc{Trans} can be instantiated as either \textsc{TransSeq} for sequential execution or \textsc{TransSpec} for speculative execution, allowing the fixpoint algorithm to be applied to both sequential abstract interpretation (SeqAI) and its speculative variants.
While \textsc{TransSeq} rules are well-established in existing works \cite{abstract-interpretation1977,valuedomain}, we explain how \textsc{TransSpec} works for speculative execution.

The primary distinction between \textsc{TransSeq} and \textsc{TransSpec} in our approach lies in how branch conditions are handled. 
\textsc{TransSeq} constrains the program state based on branch conditions when processing branch instructions. 
In contrast, \textsc{TransSpec} avoids computations that constrain the program state based on these conditions. 
Specifically, branch jump instructions are treated like skip instructions when applying the \textsc{TransSpec} function, except for updating the program counter.

Consider the example in \Cref{fig:seq-spec-ai}, where \textsc{TransSeq} and \textsc{TransSpec} are applied to process the statement \texttt{if (x<5) {y=x;} else {y=5;}}.
\textsc{TransSeq} uses the branch condition \texttt{x < 5} to deduce that the final value of \texttt{y} must be within [0, 5]. 
In contrast, \textsc{TransSpec} ignores the branch condition, resulting in the inferred range for \texttt{y} being [0, $\infty$).

While existing works \cite{ai-spec2019,spectector2020} primarily address the challenge of merging abstract states from sequential execution and speculative rollbacks to model speculative execution, our approach leverages a fundamental characteristic of monotone frameworks: computing the least upper bound of all predecessor states. 
When considering an infinite speculative window, semantics for rollback become unnecessary, making this scenario particularly well-suited for monotone frameworks, which inherently possesses this merging property.
Consequently, \textsc{TransSpec} seamlessly integrates the additional effects introduced by speculative execution.

\begin{figure}[h] 
    \centering
    \begin{tikzpicture}[%
        every node/.style={font=\scriptsize},
        block/.style={%
            draw, 
            rectangle split, 
            rectangle split parts=2, 
            rectangle split horizontal=false, 
            text centered,
            font=\scriptsize 
        },
        sblock/.style={%
        draw, 
        text centered,
        font=\scriptsize 
    },
        arrow/.style={->, thick}
    ]
    
    \node[block] (bb1) at (0,3)
    {
        \nodepart{one} if(x$<$5)
        \nodepart{two} x:[0,$\infty$)
    };
    \node[block] (bb2) at (-1,1.5)
    {
        \nodepart{one} y=x
        \nodepart{two} y:[0,5)
    };
    \node[block] (bb3) at (1,1.5)
    {
        \nodepart{one} y=5
        \nodepart{two} y:[5,5] 
    };
    \node[sblock] (bb4) at (0,0)
    {
        y:[0,5)$\cup$[5,5]=[0,5] 
    };
    \node (spec) at (0,-0.5) {\textsc{TransSeq}};
    
    \draw[arrow] (bb1) -- (bb2) node[midway, left] {x$<$5}; 
    \draw[arrow] (bb2) -- (bb4); 
    \draw[arrow] (bb1) -- (bb3) node[midway, right] {x$\ge$5};
    \draw[arrow] (bb3) -- (bb4);

    \node[block] (sbb1) at (4,3)
    {
        \nodepart{one} if(x$<$5)
        \nodepart{two} x:[0,$\infty$)
    };
    \node[block] (sbb2) at (3,1.5)
    {
        \nodepart{one} y=x
        \nodepart{two} y:[0,$\infty$)
    };
    \node[block] (sbb3) at (5,1.5)
    {
        \nodepart{one} y=5
        \nodepart{two} y:[5,5] 
    };
    \node[sblock] (sbb4) at (4,0)
    {
        y:[0,$\infty$)$\cup$[5,5]=[0,$\infty$)
    };
    
    \draw[arrow] (sbb1) -- (sbb2) node[midway, left] {no constraints};
    \draw[arrow] (sbb2) -- (sbb4);
    \draw[arrow] (sbb1) -- (sbb3) node[midway, right] {no constraints};
    \draw[arrow] (sbb3) -- (sbb4);
    \node (spec) at (4,-0.5) {\textsc{TransSpec}};
    \end{tikzpicture}
    \caption{\textsc{TransSeq} vs. \textsc{TransSpec}.} 
    \label{fig:seq-spec-ai} 
\end{figure}

\subsection{\textsc{TransHK}}
Accounting for the changes in program behavior caused by hardening, LightH employs a transition function with hardening knowledge, denoted by \textsc{TransHK}, to compute the transition of states.
\textsc{TransHK} leverages the observation that the memory accesses of hardened instructions are blocked during misspeculative execution. 
Consequently, the maximum attainable program state within the state domain, as determined by the speculative semantics after executing these hardened instructions, is identical to the maximum attainable state under sequential semantics.

When processing instructions not marked for hardening, \textsc{TransHK} updates the program state in the same manner as \textsc{TransSpec}. 
For instructions marked for hardening, \textsc{TransHK} utilizes the fixpoint of SeqAI (\seqfinal), which provides an over-approximation of the program's states during sequential execution, as the program state during speculative execution.

Specifically, for a load instruction $\pload{x}{e}$, let SeqAI determine the range of $x$ as $V$. 
Assume the instruction is hardened as $\pload{x}{e\ |\ \text{mask}}$, we consider two scenarios for the range of $x$.
In sequential execution,  the range of $x$ remains $V$, as previously established by the analysis.
However, during misspeculative execution, the value of $e\ |\ \text{mask}$ becomes -1, indicating an invalid address and preventing the memory access from proceeding.
Consequently, the instruction cannot complete, and any subsequent computations involving $x$ also fail to execute.
Therefore, after hardening, the range of $x$ remains $V$, even taking into account the effects of speculative execution.
Hence, in the analysis, for the $\pload{x}{e}$ instruction, \textsc{TransHK} directly reuses the previously determined result, $V$, as the range of $x$.

Similarly, for a store instruction $\pstore{x}{e}$, let SeqAI determine the range of $x$ as $V_x$, the range of $e$ as $V_e$.
After hardening, this instruction will be blocked during misspeculative execution for the same reason as for load instructions.
Therefore, LightH uses $V_x$ to update only the contents of addresses within the range of $V_e$.

\begin{algorithm}
    \begin{algorithmic}[1]
        \item[\textbf{Input}] $p$: a program, $\abs$: an initial state, \seqfinal: the result of sequential abstract interpretation
        \item[\textbf{Output}] \hardened{p}{P}: a  hardened program
        \STATE Initialize $\abS(n)$ to \botof{\aconf} for every $n$
        \STATE Initialize $\abS'(n)$ to \abs if $n=0$, and to \botof{\aconf} otherwise
        \STATE $\currenthardenlist \leftarrow \{\}$, $\lasthardenlist \leftarrow \{\}$
        \WHILE{$\exists\ n$ s.t.  $\abS(n)\neq\abS'(n)$ or $\currenthardenlist \neq \lasthardenlist$} 
        \STATE $\abS \leftarrow \abS'$, $\lasthardenlist \leftarrow \currenthardenlist$
        \FOR{every instruction $p(n)$}
            \STATE $\abS'(n) = \bigsqcup\limits_{n'\in\textsc{Pred}(n)}\textsc{TransHK}(\abS(n'),n,\seqfinal,\currenthardenlist)$ 
        \ENDFOR
        \STATE $\currenthardenlist \leftarrow  \hardenset(\abS')$  
        \ENDWHILE
        \STATE \hardened{p}{P} $\leftarrow$ Harden $p$ according $K$
        \RETURN \hardened{p}{P}
    \end{algorithmic}
    \caption{SpecHK}
    \label{algo:ai-hk}
\end{algorithm}
\subsection{Speculative Analysis with Hardening Knowledge}\label{sec:lighth}
To address the issue in \Cref{sec:mov-paralysis}, we propose a novel fixpoint algorithm called SpecHK (speculative analysis with hardening knowledge) based on the monotone framework.
The algorithm for SpecHK is detailed in \Cref{algo:ai-hk}. 
SpecHK takes as inputs a program $p$, an initial state $\abs$, and the fixpoint of SeqAI ($\seqfinal$), and returns the hardened program \hardened{p}{P}.
The algorithm iterates over both the program configuration $\abS$ and the hardened instruction set $K$. 
In each iteration,  SpecHK employs \textsc{TransHK} to compute the transition of program states, and uses $H(\abS)$ (discussed  in \Cref{sec:pinpoint}) to update the hardening set. 
This process continues until both the program configuration and the hardening set converge to a fixpoint. 
Once this fixpoint is reached, hardening is applied to the identified instruction set.

By incorporating hardening knowledge into the analysis, we address the issue discussed in  \Cref{sec:mov-paralysis}.
When analyzing \Cref{lst:store} in \Cref{lst:analysis-paralysis},  SpecHK  determines that \texttt{a[x]} might cause an out-of-bounds store, marking it as requiring protection. 
Subsequently, SpecHK directly uses the result from SeqAI, i.e., \texttt{a[x]} is an in-bound access after hardening.
Therefore, when analyzing, the range of \texttt{key} is used to update only the contents within the array \texttt{a}, without affecting other memory regions. 
As illustrated, for \textsc{TransSpec}, processing out-of-bounds stores leads to updates across every memory region to maintain the soundness of the analysis, resulting in a loss of precision. 
In contrast, \textsc{TransHK} avoids the aforementioned issue, allowing the analysis to proceed while considering the effects of protecting the instruction.

Note that SpecHK leverages the knowledge of which instructions are going to be hardened, rather than directly analyzing the instructions that have already been hardened.
To ensure the soundness of the analysis and avoid the aforementioned analysis paralysis issues, it is crucial to handle the different values of $\texttt{mask}$ (the misspeculative flag) in both sequential and misspeculative execution, along with the resulting behaviors of the hardened address $e\ |\ \texttt{mask}$.
SpecHK enables the algorithm to incorporate information from the conditional assignment, enabling the analysis to account for the varying behaviors of memory address value ranges in both sequential and speculative execution.

\section{LightSLH: Pinpoint Vulnerable Instruction}\label{sec:pinpoint}

To address the issue in \Cref{sec:mov-tt}, we employ a taint-tracking approach to mark which piece of data carries sensitive information.
According to the attacker model, $H()$ determines whether a specific instruction requires hardening by evaluating whether certain bits in the instruction contain sensitive information.

In \Cref{sec:taint-tracking}, we present our formalization of bit-level taint tracking.
In \Cref{sec:speculative-safety}, we extend the program semantics with taint tracking and provide the security properties for the extended semantics.
In  \Cref{sec:LightSLH}, we introduce the abstract domains and memory models for the analysis of LightSLH, and provide a running example to demonstrate how different transfer functions operate on the domains.

\subsection{Bit-Level Taint Tracking} \label{sec:taint-tracking}
\begin{figure}
    \begin{minipage}{0.20\linewidth}
        \begin{center}
            \begin{tikzpicture}
              \node (H) at (0,1.2) {\thigh};
              \node (L) at (0,0.6)  {\tlow};
              \node (0) at (-0.6,0)  {\tzero};
              \node (1) at (0.6,0)  {\tone};
              \node (bottom) at (0,-0.6) {\tbot};
              \draw (H)--(L)--(0)--(bottom);
              \draw (L)--(1)--(bottom);
            \end{tikzpicture} 
        \end{center}
    \end{minipage}
    \begin{minipage}{0.80\linewidth}
        \begin{itemize}
            \item \tbot: an undefined value. \tbot serves as the bottom in the lattice and represents the value read from an invalid address.
            \item \tzero: a bit that is guaranteed to be zero.
            \item \tone: a bit that is guaranteed to be one.
            \item \tlow: a bit that is independent of secrets.
            \item \thigh: a bit that relates to some secrets.
        \end{itemize}
    \end{minipage}
    \vspace{1em}
    \caption{Lattice of Bit-Level Taint Labels.}
    \label{fig:taint-label-lattice}
    \end{figure}

\textbf{Bit-Level Taint Label.}
In our taint tracking method, every bit carries one of the labels in the lattice \lattice{T} in \Cref{fig:taint-label-lattice}.
We refer to \tzero and \tone as \textit{concrete labels}, and \tbot, \tlow, \thigh as \textit{non-concrete labels}.
For an $n$-bit value $v$ labeled with a taint label vector $t$ of $n$ elements, let $t[i]$ represent the taint label of the $i$-th least significant bit of $v$, and $v[i]$ denote the corresponding bit of $v$.
We can derive a lattice on $n$-bit labels, i.e., the product of \lattice{T} taken $n$ times, denoted by \prolattice{T}{n}.
For $l\in\lattice{T}$, we use $\vec{l}$ as shorthand for a vector with all elements evaluating to $l$.

\textbf{Formalized Property.}\label{sec:taint:wd}
We denote by \bopondomain{\op}{\prolattice{T}{n}}{a}{b} the operations for two taint vectors $a$ and $b$ on lattice \prolattice{T}{n}, where $\op$ is an operator.
We say a label vector $t$ of length $n$ is \textit{legal} for an $n$-bit value $v$ if all concrete labels in $t$ match the value of corresponding bits in $v$, denoted by \isinstance{v}{t}.
For $v_1$ and $v_2$ that satisfy $v_1 \vdash t$ and  $v_2 \vdash t$, we define \tainteq{v_1}{v_2}{t} if $v_1[i] = v_2[i]$ holds for any $0\le i \le n-1$ such that $t[i]\neq \thigh$.
For example, $(\tone, \tzero, \thigh, \tlow)$ is legal for $\texttt{b1001}=5$, $(\tone, \tone, \thigh, \tlow)$ is not legal for $\texttt{b1001}=5$, and \tainteq{\texttt{b1001}}{\texttt{b1000}}{(\tlow,\tzero,\tlow,\thigh)}.
We present the definition of \textit{well-defined} as below:

\begin{definition}[Well-Defined]\label{def:well-defined}
An operator \op{} on \prolattice{T}{n} is  well-defined iff for any $ t_1 , t_2 \in \prolattice{T}{n}$, and \isinstance{v_1}{t_1}, \isinstance{v_2}{t_2}, then \bopondomain{\op}{\prolattice{T}{n}}{t_1}{t_2} satisfies
\begin{itemize}[left=0pt] 
    \item \textbf{Legality}: \isinstance{v_1\op v_2}{\bopondomain{\op}{\prolattice{T}{n}}{t_1}{t_2}}. 
    \item \textbf{Non-Interference}: For any $v'_1$ s.t. \tainteq{v_1}{v'_1}{t_1}, and any $v'_2$ s.t. \tainteq{v_2}{v_2'}{t_2}, $\tainteq{v_1\op v_2}{v'_1\op v'_2}{\bopondomain{\op}{\prolattice{T}{n}}{t_1}{t_2}} $ holds.
\end{itemize}
    
\end{definition}

The first requirement stipulates the result of the operator on \prolattice{T}{n} should be legal for the natural result of its application to concrete values, indicating that the bit with a \tzero (or \tone) label should always evaluate to 0 (resp. 1).
The second provides a formal description for \tlow and \thigh. 
It states that when any operand (e.g., $v_1$) is replaced with a new one (e.g., $v_1'$) having the same values on all non-\thigh bits, the result of the computation will remain unchanged in the corresponding non-\thigh bits of the resulting taint label vector. 
This essentially enforces a non-interference property, ensuring that modifications to \thigh bits do not influence the non-\thigh bits.


\textbf{Operator Rules.}
Based on the definition of the well-defined operator, we can build establish the operator rules for complex computations such as multiplication. 
For $\mulop$, a conservative approach involves reverting to the traditional taint tracking method, where each bit in a taint vector is assigned the same label.
A more refined approach leverages the fact that the $i$-th bit of $v_1 \cdot v_2$ is influenced only by the 0-th to $i$-th bits of $v_1$ and $v_2$. 
Consequently, the $i$-th bit of the product is set to a concrete value if all lower bits of $v_1$ and $v_2$ are labeled as concrete, or it is set to \tlow if none of the lower bits of $v_1$ and $v_2$ are labeled as \thigh.
It can be proven that both operators designed for $\mulop$ are well-defined. 
In fact, as long as the well-defined property is satisfied, we can design bit-level tracking methods with different levels of precision, tailored to the specific needs of the analysis.


We provide  complete rules in \Cref{app:sec:taint-operation} for each operator in \Cref{fig:core-language}, and demonstrate that they are well-defined.




\subsection{Reducing Security to Safety}\label{sec:speculative-safety}

Similar to previous work \cite{ct-verif-ai2017,exorcise2021}, we extend taint tracking into the speculative semantics to track sensitive information. 
The full semantics are detailed in \Cref{app:sec::full-semantics}.
Specifically, for the state transition $(p, s)\xrightarrow[d]{o} (p, s')$, we use $t(o)$ to represent the taint label of the observation.
Using the taint tracking method, we can trace data that carries sensitive information and apply protections accordingly. 

We propose a new property parametric in the taint tracking method defined in \Cref{sec:taint-tracking}, called \textit{speculative safety (SS)}. 
In short, SS imposes a restriction on the taint labels in observations generated during misspeculative execution, prohibiting \thigh labels.
The technique of reducing a hyperproperty to a safety property parametric in taint tracking is common in previous work  \cite{ct-verif-ai2017,exorcise2021,caches-2019}.
However, SS defined in our work distinguishes itself in the following two ways.
First, it features the bit-level taint tracking mechanism, facilitating reevaluation of the security guarantee of speculative safety.
Second, by requiring that observations generated during misspeculative execution carry no \thigh labels, we can circumvent the requirement to track implicit information flows, while maintaining security guarantees.

A program $p$ satisfies SS w.r.t $P$, written $\ssafety{p}{P}$, if no observation with a taint vector containing \thigh labels is generated during misspeculative execution in any execution trace.
\begin{definition}[SS]\label{def:speculative-safety}
    $\ssafety{p}{P}$, iff
    for any $\trace{p}{s_1}{D}{O}=(p, s_1)\xrightarrow[d_1]{o_1} (p, s_2)\cdots(p, s_n)$,
     s.t. $s_1 \vdash P $,
     we have,
    $$\forall 1\le i\le n, f_{s_i} = \top \Rightarrow\thigh\notin t(o_i) $$
\end{definition}

The following theorem (proved in \Cref{app:sec:sni-ss}) describes the security guarantee of SS.
Specifically, a program $p$ that satisfies SS w.r.t $P$, also satisfies SNI w.r.t $P$.
\begin{restatable}{theorem}{SStoSNI}\label{the:sni-ss}
    $\ssafety{p}{P} \Rightarrow \sni{p}{P}$.
\end{restatable}

\subsection{Domain for Program State}\label{sec:LightSLH}
In our work, the program state consists of both a taint vector from \textit{taint domain}, and an abstract value from \textit{value domain}. 
The taint domain, as described in \Cref{sec:taint-tracking}, is employed to track the propagation of sensitive information.
The value domain represents variable and pointer values, facilitating memory access computations and enabling the detection of out-of-bounds accesses.

\textbf{Value Domain.}
For the analysis of the range of each variable and memory address, we employ the value domain (denoted by \VD) as in \cite{valuedomain}.
In short, we apply a representation scheme where a value is associated with a set of ranges, each defined by an offset relative to a base.
We carefully track the values computed from such base addresses, recording their corresponding offsets within the associated memory region.
For values that lack a traceable origin to a specific memory allocation, we consider the base address to be empty, denoted by $\emptysym$. In such cases, the offset solely represents the range itself.
Particularly, for bases other than $\emptysym$, we represent their offsets using several non-overlapping intervals, while for ranges with an empty base, we use only a single interval to represent the offset. 
For example, $(s:[3,5]\cup[9,11])$ represents the values at addresses with $s$ as the base address and offsets within the range $[3,5]\cup[9,11]$.
This approach allows for a precise description of memory accesses to arrays of structures when representing pointers, while simplifying the representation when handling numbers.

\textbf{Memory Model.}
Since pointers are represented by abstract values, processing memory access instructions requires determining the potentially accessed memory cells.
Our abstract memory model, similar to the value domain, consists of multiple memory regions, each with a base address, size, and values stored at offsets ranging from $0$ to $\text{size}-1$.
Memory accesses using abstract pointers determine content based on their base-offset pairs.
Particularly, memory accesses to abstract addresses of base-offset pairs with \emptysym as bases, or out-of-bounds offsets, are handled as follows: load operations yield the top element of the lattice; store operations are treated as to every memory cell.

\begin{figure}
    \begin{lstlisting}[mathescape=true, numbers=none]
    if ( x < 8 ){ y = a[x]; z = b[y]; w = c[z]; }
    \end{lstlisting}    
    \footnotesize
    \begin{tabular}{llll}
        Expr            &\textsc{TransSeq}                                               & \textsc{TransSpec}                                                  & \textsc{TransHK}\\\hline
        \texttt{x}      & \statecell{${(\emptysym,[0,7])}$}{$\vec{\tlow}$}     & \statecell{${(\emptysym,[0,15])}$}{$\vec{\tlow}$}         & \statecell{${(\emptysym,[0,15])}$}{$\vec{\tlow}$} \\
        \texttt{a+x}    & \statecell{${(\texttt{a},[0,7])}$}{$\vec{\tlow}$}   & \statecell{${(\texttt{a},[0,15])}$}{$\vec{\tlow}$}       & \statecell{${(\texttt{a},[0,15])}$}{$\vec{\tlow}$}\\
        \texttt{y=a[x]} & \statecell{${(\emptysym,[0,255])}$}{$\vec{\tlow}$}       & \statecell{\topof{\VD}}{$\vec{\thigh}$}               & \statecell{\topof{\VD}}{$\vec{\thigh}$} \\
        \texttt{b+y}    & \statecell{${(\texttt{b},[0,255])}$}{$\vec{\tlow}$} & \fbox{\statecell{\topof{\VD}}{$\vec{\thigh}$}}     & \fbox{\statecell{\topof{\VD}}{$\vec{\thigh}$}} \\
        \texttt{z=b[y]} & \statecell{${(\emptysym,[0,255])}$}{$\vec{\tlow}$}       & \statecell{\topof{\VD}}{$\vec{\thigh}$}               & \colorbox{gray!30}{\statecell{${(\emptysym,[0,255])}$}{$\vec{\tlow}$}} \\
        \texttt{c+z}    & \statecell{${(c,[0,255])}$}{$\vec{\tlow}$}          & \fbox{\statecell{\topof{\VD}}{$\vec{\thigh}$}}     & \statecell{${(c,[0,255])}$}{$\vec{\tlow}$} \\
        \texttt{w=c[z]} & \statecell{${(\emptysym,[0,255])}$}{$\vec{\tlow}$}       & \statecell{\topof{\VD}}{$\vec{\thigh}$}               & \statecell{${(\emptysym,[0,255])}$}{$\vec{\tlow}$} \\
    \end{tabular}
    
    \caption{
        An example to illustrate how different \textsc{Trans} functions work on the abstract domain.
        Pointers that generate an observation with \thigh labels when performing memory access are \fbox{ boxed }.
        The results utilized from the analysis under abstract sequential semantics are marked with \colorbox{gray!30}{colorbox}.
    }
    \label{fig:LightSLH}
\end{figure}
    
    \textbf{A Running Example.}
    We explain how different \textsc{Trans} functions work on the abstract domains using \Cref{fig:LightSLH} as an example. 
    The variable \texttt{x} initially holds the abstract value ${(\emptysym,[0,15])}$ with a taint vector $\vec{\tlow}$.
    \texttt{a} is an array of 8-bit unsigned integers with a size of 8, while \texttt{b} and \texttt{c} are both arrays of 8-bit unsigned integers with a size of $2^8=256$.
    The contents of \texttt{a}, \texttt{b} and \texttt{c} are all labeled with $\vec{L}$.
    Each item in the table represents the abstract value and the taint label vector corresponding to the variables in the first column.
    
    For \textsc{TransSeq}, since the branch condition restricts \texttt{x}'s range to $(\emptysym,[0,7])$, the memory accesses of \texttt{a[x]}, \texttt{b[y]} and \texttt{c[z]} are all within bounds.  
    In contrast, for \textsc{TransSpec}, which ignores branch conditions, \texttt{a[x]} may result in an out-of-bounds access, causing  \texttt{y=a[x]} to be set to  \topof{\VD} (the top value of \VD), and labeled with $\vec{\thigh}$.
    As a result, the pointer \texttt{b+y} is also labeled with $\vec{\thigh}$ and accessing $\texttt{b[y]}$ generates an observation with \thigh labels, indicating that \texttt{b+y} should be hardened.
    Similarly, \texttt{c+z} will also be marked as requiring hardening.
    \textsc{TransHK} behaves as \textsc{TransSpec} except for instructions marked for hardening. 
    When processing \texttt{z=b[y]}, since \texttt{b+y} is marked, \textsc{TransHK} directly utilizes the results of SeqAI, assigning $(\emptysym,[0,255])$ and $\vec{\tlow}$ as the value of $\texttt{z}$.

    \textbf{Soundness.}
    Based on the monotone framework and the abstract domain, we prove the soundness of LightSLH.
    Specifically, for the fixpoint $(\Omega,K)$ obtained by LightH,  $\Omega$ provides an over-approximation of the behavior of the hardened program \hardened{p}{P}. As a result (with proof sketch detailed in \Cref{sec:sketch}), we have 

    \begin{restatable}{theorem}{lightslh}\label{thm:LightSLH}
        $\ssafety{\hardened{p}{P}}{P}$
    \end{restatable}


    \readded{\subsection{Variants of LightSLH}
    \label{sec:LightSLH-variants}
    While LightSLH is based on the SLH strategies, to show the flexibility of our framework to integrate different hardening strategies, we also propose two variants which based on the fence strategies: LightFence and LightCut.}

    \readded{\textbf{LightFence.} Like LightSLH, LightFence places protection (\texttt{lfence}) right before an instruction when SpecHK identifies it (1) performs an out-of-bound store, (2) performs an secret-dependent memory access of branch.}

    \readded{\textbf{LightCut.} 
    Similar to LightFence, LightCut employs \texttt{lfence} instructions to harden the program. The key difference lies in its approach to handling out-of-bounds stores. Inspired by LLSCT \cite{serberus2024}, LightCut introduces a deferred hardening strategy: out-of-bounds stores (referred to as sources) are only hardened when a memory operation or branch (referred to as sinks) dependent on these sources is detected. 
    Specifically, upon identifying a source, LightCut does not immediately insert a fence. Instead, it records the source and waits for the corresponding sink to be identified. When a sink is detected, LightCut inserts a fence between the source and the sink. If the source resides within a loop, the fence is strategically placed at the exit of the outermost loop that does not contain any sinks. This approach minimizes unnecessary fences while ensuring robust protection.}

\readded{\textbf{Security Guarantees.}
As described in \Cref{sec:lighth}, the algorithm runs until reaching the fixpoint $(\Omega, K)$, where $\Omega$ provides an over-approximation of program states after applying protections based on the set $K$. And at this fixpoint, no further protections are required according to the current state. Since different hardening strategies only affect where protections are placed, \Cref{thm:LightSLH} remains valid for both LightFence and LightCut.}
\readded{However, as we will demonstrate in \Cref{sec:chacha}, LightFence and LightSLH offer more robust protection by accounting for the compiler's inserted load when lowering high level code to low level binary.}

\begin{figure*}[!tp]
    \centering
    \fontsize{7.5}{9}\selectfont
    \setlength{\tabcolsep}{2.5pt}
    \redeleted{
    \begin{tabular}[tp]{@{}c|c|c|c|c|c|c|c|c|c@{}}
        \multirow{2}{*}{Algorithm}   & \multirow{2}{*}{Analysis Starting Function} &\multirow{2}{*}{\makecell{\# of Analyzed \\ Inst.}} &\multirow{2}{*}{\makecell{\# of Processed Inst. \\ (SeqAI + LightH)}} &\multirow{2}{*}{\makecell{\# of Analyzed \\ Func.\\}} & \multirow{2}{*}{\makecell{\# of Processed Func.\\ (SeqAI + LightH)}} & \multirow{2}{*}{\makecell{Analysis \\Time}}  & \multicolumn{3}{c}{\# hardened / \# total} \\
        \cline{8-10}
        &&&&&&& load & store & branch
        \\\hline
        AES(NI)*            & \texttt{AES\_encrypt}                     &0.4k&1.4k+1.3k&2&5+5& $<0.1$s           &0/34       &0/30       &0/3       \\ 
        RSA*            & \texttt{BN\_mod\_exp\_mont\_consttime}    &7.8k&263.7k+156.4k&74&1362+841& $43$s      &396/715    &195/361    &290/560\\
        ChaCha20*       & \texttt{ChaCha20\_ctr32}                  &0.2k&0.7k+0.7k&1&1+1& $<0.1$s           &0/20       &7/12       &0/15      \\ 
        Curve25519*     & \texttt{ossl\_x25519}                     &2.6k&7.8k+7.2k&2&21+21& $0.1$s            &0/125      &0/156      &0/9       \\ 
        Poly1305*       & \texttt{Poly1305\_Update}                 &0.3k&1.2k+0.8k&2&3+3& $<0.1$s           &0/27       &0/6        &0/6       \\ 
        SHA256*         & \texttt{SHA256\_Update}                   &1.3k&6.6k+4.8k&2&3+3& $<0.1$s            &0/112      &0/36       &0/8       \\\hline 
        RSA\textdegree  & \texttt{BN\_mod\_exp\_mont\_consttime}    &9.5k&407.2k+223.6k&57&1379+650& $27$s             &379/816    &182/398    &231/604 \\\hline
        Salsa20$^+$     &\texttt{crypto\_core\_salsa20}             &0.7k&0.9k+0.9k&2&2+2&$<0.1$s&0/64&0/64&0/3\\
        SHA256$^+$      &\texttt{crypto\_hash\_sha256\_update} &1.1k&3.3k+3.2k&2&3+3    &$<0.1$s &0/102&37/81&0/42\\\hline            
        Salsa20$^-$     &\texttt{crypto\_core}             &0.7k&0.9k+1.0k&1&1+1&$<0.1$s&0/64&0/64&0/1\\
        SHA256$^-$     &\texttt{crypto\_hashblocks}             &2.8k&5.3k+5.3k&1&1+1&$<0.1$s&0/97&0/32&0/2\\
        Poly1305$^-$     &\texttt{crypto\_onetimeauth}             &0.6k&1.4k+1.3k&1&1+1&$<0.1$s&0/69&10/40&0/21
    \end{tabular}}
\readded{
    \begin{tabular}[tp]{@{}c|c|c|c|c|c|c|c|c|c@{}}
        \multirow{2}{*}{Algorithm}   & \multirow{2}{*}{Analysis Starting Function} &\multirow{2}{*}{\makecell{\# of Analyzed \\ Inst.}} &\multirow{2}{*}{\makecell{\# of Processed Inst. \\ (SeqAI + LightH)}} &\multirow{2}{*}{\makecell{\# of Analyzed \\ Func.\\}} & \multirow{2}{*}{\makecell{\# of Processed Func.\\ (SeqAI + LightH)}} & \multirow{2}{*}{\makecell{Analysis \\Time}}  & \multicolumn{3}{c}{\# hardened / \# total} \\
        \cline{8-10}
        &&&&&&& load & store & branch
        \\\hline
        AES(NI)*            & \texttt{AES\_encrypt}                     &0.4k&1.4k+1.3k&2&5+5& $<0.1$s           &0/34       &0/30       &0/3       \\ 
        RSA*            & \texttt{BN\_mod\_exp\_mont\_consttime}    &9.7k&301.7k+231.1k&74&1284+725& $35.9$s      &488/719    &248/364    &341/568\\
        ChaCha20*       & \texttt{ChaCha20\_ctr32}                  &0.4k&1.4k+1.2k&1&1+1& $<0.1$s           &0/35       &11/27       &0/18      \\ 
        Curve25519*     & \texttt{ossl\_x25519}                     &2.6k&6.9k+7.5k&2&18+24& $0.2$s            &0/125      &0/154      &0/9       \\ 
        Poly1305*       & \texttt{Poly1305\_Update}                 &0.3k&1.2k+0.8k&2&3+3& $<0.1$s           &0/27       &0/6        &0/6       \\ 
        SHA256*         & \texttt{SHA256\_Update}                   &1.3k&5.4k+5.2k&2&3+3& $<0.1$s            &0/112      &0/36       &0/8       \\\hline 
        RSA\textdegree  & \texttt{BN\_mod\_exp\_mont\_consttime}    &10.0k&392.2k+291.7k&57&1386+772& $19.2$s             &557/852    &286/437    &361/629 \\\hline
        Salsa20$^+$     &\texttt{crypto\_core\_salsa20}             &0.7k&0.9k+0.8k&2&2+2&$<0.1$s&0/64&0/64&0/3\\
        SHA256$^+$      &\texttt{crypto\_hash\_sha256\_update} &1.3k&3.7k+3.5k&2&3+3    &$0.2$s &0/124&50/99&0/48\\\hline            
        Salsa20$^-$     &\texttt{crypto\_core}             &0.7k&0.9k+0.9k&1&1+1&$<0.1$s&0/64&0/64&0/1\\
        SHA256$^-$     &\texttt{crypto\_hashblocks}             &2.8k&5.3k+5.3k&1&1+1&$<0.1$s&0/97&0/32&0/2\\
        Poly1305$^-$     &\texttt{crypto\_onetimeauth}             &1.0k&2.5k+2.1k&1&1+1&$<0.1$s&0/132&8/102&0/23\\\hline
        \readded{Kyber$^\bullet$} &\texttt{kyber\_clean\_crypto\_kem\_enc}             &3.1k&36.1k+40.2k&39&152+127&$0.5$s&0/284&125/234&5/83\\
        \readded{McEliece$^\bullet$} &\texttt{mceliece\_clean\_encrypt}             &0.6k&1.1k+0.9k&5&9+9&$<0.1$s&0/97&5/13&3/14\\
    \end{tabular}}
    \caption{Analysis information of LightSLH.
     *: \openssl3.3.0. \textdegree: \openssl1.0.2f. $^+$: Libsodium 1.0.20. $^-$: NaCL. \readded{$^\bullet$: PQClean}}
    \label{fig:analysis} 
\end{figure*}

\section{Evaluation}\label{sec:evaluation}

Following the methodology outlined in previous sections, we implement LightSLH, \readded{LightFence and LightCut} as an LLVM \cite{lattner2004llvm} IR pass.
Details about our implementation is presented in \Cref{sec:impl}.
\redeleted{In this section, we evaluate the performance of LightSLH, including its analysis efficiency and the overhead associated with its protective measures.
}\redeleted{To demonstrate the flexibility and generality of our methodology, we additionally implement LightH as LightFence based on the fence protection strategy.}
\readded{In this section, we evaluate the overhead associated with different protective measures, and the analysis efficiency of LightSLH.
}
We present our main results in \Cref{sec:main-result} and discuss two case studies that highlight interesting findings and insights encountered during the analysis in the next two subsections.

\textbf{Workloads.}
We evaluate \LightSLH on several cryptographic primitives from mainstream cryptographic libraries, including AES(NI), ChaCha20, Poly1305, SHA256, Curve25519 and RSA from OpenSSL 3.3.0, Salsa20 and SHA256 from Libsodium 1.20.0, and Poly1305, Salsa20, and SHA256 from NaCL, \readded{and Kyber512 and McEliece348864 from PQClean}.
While OpenSSL 3.3.0's RSA implementation adopts a conservative approach against cache attacks, we also analyze RSA from OpenSSL 1.0.2f which uses scatter-gather techniques as its cache side-channel defense mechanism.

\textbf{Baselines.}
\readded{We compare against three baselines: two basic defenses (Fence and SSLH) and LLSCT \cite{serberus2024}, a state-of-the-art tool for hardening against Spectre attacks.
Fence implements speculation barriers at basic block entries, while SSLH protects all load, store, and branch instructions, improving upon its predecessor's incomplete protection \cite{exorcise2021}.}
\readded{While LLSCT offers comprehensive protection for 5 different Spectre variants, in our settings, we evaluated LLSCT with Spectre PHT (v1), BTB \cite{spectre2019} and RSB \cite{ret2spec2018} protections enabled, as LLSCT couples v1 protection with BTB/RSB mitigations. 
For cryptographic primitives, Spectre BTB attacks are inherently mitigated due to the absence of indirect branches \cite{typing2023}.}
\curdeleted{, and protections against Spectre RSB are predominantly implemented at the hardware level (discussed in \Cref{sec:discussion}).}
\added{For RSB-based attacks, while LLSCT offers software-level protection against initial RSB overflow scenarios, we discuss defenses against more comprehensive RSB attacks, including recently discovered underflow variants, in \Cref{sec:discussion}.}
\redeleted{We evaluate against three baselines: two basic defenses (Fence and SSLH) and LLSCT \cite{serberus2024}, an optimized hardening strategy. 
Fence implements speculation barriers at basic block entries, while SSLH protects all load, store, and branch instructions, improving upon its predecessor's incomplete protection \cite{exorcise2021}. 
LLSCT offers comprehensive protection against Spectre v1-4 for programs which satisfy constant-time in sequential semantics, We select LLSCT as our primary baseline due to its demonstrated lower performance overhead compared to existing approaches \cite{blade2021,uslh2023,llvm:slh}.
Modern processors and operating systems have mitigated Spectre v3-4 vulnerabilities through system-level mechanisms including microcode updates and other mechanisms \cite{intel:ssb}. 
Additionally, cryptographic primitives are inherently resistant to Spectre v2 attacks \cite{typing2023} due to the absence of indirect branches.
In our settings, we evaluated LLSCT with Spectre v1-3 protections enabled, as LLSCT couples v1 protection with v2/v3 mitigations.
While Declassiflow \cite{declassiflow2023} targets overhead reduction for Spectre v1 defenses, we exclude it from our evaluation due to functional incorrectness in its implementation\footnote{The hardened program produces incorrect outputs compared to the original.}.}

\textbf{Experiment Setup.} 
We run experiments on a machine with 2.40GHz Intel Xeon\textregistered{} CPU E5-2680 and 128 GB memory.
We compile all workloads using Clang \redeleted{16.0.0 } \readded{14.0.4} and LLVM \redeleted{16.0.0 } \readded{14.0.4} with optimization flag \texttt{-O2}.
The only manual intervention required is annotating critical function parameters with security attributes to identify secret variables in the program's initial state.
\redeleted{In accordance with the implementation in the cryptographic libraries, our analysis assumes a cache size of 64B, allowing the attacker to observe the 58 most significant bits of a 64-bit address. }
\readded{We conducted our analysis under two settings: (1) a 64B cache line size, where the attacker observes the top 58 bits of a 64-bit address, consistent with the RSA implementation in OpenSSL, and (2) the attacker observes the entire address (equivalent to a cache line size of 1).}

\subsection{Main Results}\label{sec:main-result}

\readded{\textbf{Different Attack Models.}
Interestingly, we observe that the two attack models yield identical results in terms of overhead and protection. 
Upon analysis, we identify two key reasons: 
(1) For the code in our experiments that is not specifically designed to be cache-aware, LightSLH's protections are entirely attributed to out-of-bounds stores or secret-dependent branches, and there is no secret-dependent memory access. The protections for these cases are independent of the cache size used. 
(2) For RSA (i.e., cache-aware code), as detailed in \Cref{sec:rsa}, even at the cache-line granularity, the cache-aware components require protection (and so under the full-address model). Consequently, LightSLH applies the same protections to RSA under both attacker models.
To avoid redundancy, we only present results under one attacker model.}

\textbf{Efficiency of Analysis.} \Cref{fig:analysis}  presents the analysis information for each workloads. 
LightSLH completes analysis in 0.1 second for most cases, with the complex RSA implementation taking under 1 minute. 
Notably, LightSLH is the first tool to analyze RSA implementations against Spectre attacks.

Unlike the previous mitigations (including the three baselines) to apply hardening regardless of vulnerability, LightSLH can pinpoint the vulnerable instructions and then applying targeted hardening accordingly.
We do not directly compare analysis efficiency with other provable detection tools, as these tools are primarily focused on detecting the presence of vulnerabilities without the capability to identify all potential issues \cite{spectector2020,pitchfork2020,cats2022,hunter2021,kleespectre, typing2023, ai-spec2019}.
Nevertheless, LightSLH demonstrates higher analysis efficiency than these tools based on their reported experimental results.
For instance, although \textsc{Binsec/HAUNTED} \cite{hunter2021} employs a specially optimized symbolic execution algorithm that achieves improved analysis efficiency compared to other tools \cite{pitchfork2020, kleespectre}, it still takes over 30 minutes to verify OpenSSL's curve algorithm. 
Furthermore, \textsc{Spectector}\cite{spectector2020} and Kaibyo \cite{cats2022} only perform analyses on small cases of no more than 100 lines.

Moreover, existing tools \cite{spectector2020,pitchfork2020,cats2022,hunter2021,kleespectre, ai-spec2019} require repeated analysis after each vulnerability fix.
For example, according to \Cref{fig:analysis}, ChaCha20 from OpenSSL, SHA-256 from Libsodium, and RSA require 11, 50, and over 100 rounds of analysis, respectively. 
This iterative re-analysis further exacerbates their performance.
In contrast, LightSLH's framework enables continuous analysis post-hardening, avoiding costly re-analysis, leading to substantial efficiency gains.

\begin{figure*}[t]
    \centering
    \scriptsize
    
    \redeleted{
        \begin{tabular}[tp]{@{}c|c|c|c|c|c@{}}
            \multirow{2}{*}{Algorithm}    &   \multicolumn{5}{c}{Overhead}\\
            \cline{2-6}&Fence&  SSLH&   LLSCT & LightFence&  LightSLH\\\hline
            AES(NI)*        &   8.8\%&  6.8\%&  4.1\%&  \textbf{0\%} &  \textbf{0\%}\\
            RSA*        &   548.0\%&    145.4\%& $\times$    &131.3\%&\textbf{101.6\%}\\
            ChaCha20*   &   67.8\%& 6.8\%&\textcolor{gray}{0\%}   &20.4\%& 4.8\%\\
            Curve25519* &   8.5\% &20.8\%&5.1\%&\textbf{0\%}&\textbf{0\%} \\
            Poly1305*   &   97.9\%&33.6\%&33.9\%&\textbf{0\%}&\textbf{0\%}\\
            SHA256*     &   18.8\%&28.0\%&4.2\%&\textbf{0\%}&\textbf{0\%}\\\hline
            RSA\textdegree& 466.7\%&136.2\%&$\times$&153.8\%&\textbf{95.7\%}\\\hline
            Salsa20$^+$ &62.3\%&15.1\%&6.2\%&\textbf{0.0\%}&\textbf{0\%}\\
            SHA256$^+$  &44.5\%&45.7\%&\textcolor{gray}{0\%}&38.1\%&5.6\%\\\hline
            Poly1305$^-$&638.5\%&87.0\%&\textcolor{gray}{7.9\%}&123.7\%&50.9\%\\
            Salsa20$^-$ &3.2\% &19.1\% &7.8\%&\textbf{0\%}&\textbf{0\%}\\
            SHA256$^-$  &7.7\% &9.5\%  &3.7\%&\textbf{0\%}&\textbf{0\%}
        \end{tabular}}

    \readded{\begin{tabular}[tp]{@{}c|c|c|c|c|c|c@{}}
        \multirow{2}{*}{Algorithm}    &   \multicolumn{6}{c}{Overhead}\\
        \cline{2-7}&Fence&  SSLH&   LLSCT & LightFence&  LightSLH&\readded{LightCut}\\\hline
        AES(NI)*        &   9.5\%&  5.8\%&  2.4\%&  \textbf{0\%} &  \textbf{0\%}&  \textbf{0\%}\\
        RSA*        &   476.1\%&    125.8\%& $\times$    &315.4\% &\textbf{93.5\%} & 299.3\%\\
        ChaCha20*   &   98.5\%    & 14.6\%    & 6.6\%     & 6.4\%         & \textbf{5.8\%}     &6.6\%    \\
        Curve25519* &   6.6\%     & 18.5\%    & 0.7\%     & \textbf{0\%}            & \textbf{0\%}      & \textbf{0\%}            \\
        Poly1305*   &   104.5\%   & 40.6\%    & 31.5\%&\textbf{0\%}&\textbf{0\%} & \textbf{0\%}   \\
        SHA256*     &   17.3\%    & 16.4\%    & \textbf{<0.1\%}        & \textbf{0\%}            & \textbf{0\%}    & \textbf{0\%}          \\\hline
        RSA\textdegree& 466.7\%&136.2\%&$\times$&323.7\%&\textbf{95.7\%} & 309.3\% \\\hline
        Salsa20$^+$ &53.6\%    & 13.1\%    & 4.6\%     &\textbf{0.0\%}&\textbf{0\%} & \textbf{0\%}   \\
        SHA256$^+$  &135.5\%   & 58.2\%    & 2.2\%     & 103.5\%       & 20.5\%  & \textbf{1.3\%}  \\\hline
        Poly1305$^-$&612.5\%   & 79.1\%    & 3.6\%     & 84.7\%        & 46.2\%  &\textbf{0.5\%}\\
        Salsa20$^-$ &53.4\%    & 17.8\%    & 4.1\%&\textbf{0\%}&\textbf{0\%}&\textbf{0\%}\\
        SHA256$^-$  &9.3\%     & 12.4\%    & 6.5\% &\textbf{0\%}&\textbf{0\%}&\textbf{0\%}
        \\\hline
        Kyber$^\bullet$ & 190.6\%   & 83.3\%    & 71.3\%    & 161.6\%       & 69.8\% &  \textbf{40.8\%} \\
        McEliece$^\bullet$ & 93.5\%    & 42.6\%    & 39.4\%    & 37.4\%        & 33.5\% & \textbf{4.1\%} \\\hline
        \readded{geomean (ct)} & 82.3\% & 31.2\% & 12.8\% & 24.9\% & 12.8\% & \textbf{3.9\%} \\\hline
        \readded{geomean (all)} & 114.6\% & 42.2\% & $\times$& 48.4\% & \textbf{21.9\%} & 26.2\% 
    \end{tabular}}
    \caption{Runtime overhead of mitigations for Spectre v1.
    *:  \openssl3.3.0. \textdegree:  \openssl1.0.2f. $^+$:  Libsodium 1.0.20. $^-$:  NaCL. $^\bullet$: PQClean. 
    $\times$ indicates that the algorithm does not meet the prerequisites for analysis by LLSCT, specifically the constant-time property in sequential semantics. \readded{Geomean (ct) and geomean (all) represent the geometric mean of the overhead for mitigations applied to constant-time algorithms and all algorithms, respectively.} 
    \readded{The lowest overhead is highlighted in \textbf{bold}. 
    0\% indicates that no protection is needed, and <0.1\% indicates that although there are protections placed, the overhead is negligible.}}
    
    \label{fig:overhead}
\end{figure*}

\textbf{Flexibility.} 
\redeleted{Following the methodology in \Cref{fig:overview}, we implement LightH as LightFence with an $H()$ based on the fence protection strategy. 
Building on LightSLH's implementation, LightFence requires fewer than 100 lines of C++ code, demonstrating the flexibility of our methodology.
}
\readded{Building on LightSLH's implementation, LightFence and LightCut requires fewer than 300 lines of C++ code, demonstrating the flexibility of our methodology.
}
As shown in \Cref{fig:overhead}, LightFence consistently outperforms the baseline Fence strategy in all evaluated cases, and reduces average overhead by \redeleted{74.5\%} \readded{57.8\%}, \readded{and LightCut further reduces overhead by 69.5\% compared to LLSCT}.
This improvement stems from our framework's ability to apply protections selectively, targeting only areas that require hardening instead of applying them universally.
\redeleted{Despite this, LightSLH demonstrates superior performance compared to LightFence in all cases, leading us to adopt LightSLH as our primary strategy.
}
Despite LightCut demonstrating superior performance compared to LightSLH, we adopt LightSLH as our primary strategy due to its robust protection when considering the compiler's transformation.

\textbf{Overhead.} \Cref{fig:overhead} presents the overhead caused by the various protection methods.
Among all the 14 algorithms, LightSLH outperforms both Fence and SSLH, and reduce overhead by \redeleted{68.6\%} \readded{48.1\%} on average compared to SSLH, including 0\% overhead in 7 cases.
LightSLH achieves the lowest overhead in 12 of 14 cases.
Through manual inspection and analysis of the remaining two algorithms, where LLSCT \readded{LightCut} and LightCut exhibits lower overhead, we identify that the reduced overhead stems from \redeleted{LLSCT's} \readded{their} incomplete protection strategy, which overlooks potential vulnerabilities arising from the LLVM IR to assembly transformation. 

\textbf{Security Issues.}
Our results reveals two previously unknown security issues.
In \Cref{sec:chacha}, we explain why LLSCT's protection is considered incomplete. 
In \Cref{sec:rsa}, we analyze the security properties of the scatter-gather method in RSA from OpenSSL1.0.2f.

\subsection{Case Study: Hardening for Speculative Out-Of-Bounds Store}\label{sec:chacha}

\begin{figure}[h] 
\begin{minipage}{0.30\linewidth}
    \centering
    \begin{lstlisting}[language=C,basicstyle=\footnotesize] 
for (i = 0; i < len; i++) 
    out[i] = buf[i];
    \end{lstlisting}
    \end{minipage}
    \hfill
    \vrule width 0.5pt
    \hfill
\begin{minipage}{0.60\linewidth}
    \begin{lstlisting}[style=asmstyle]
.Loop: 
mov  -8(%rbp), %eax     ; Load 'i' from the stack into %eax; 
mov  buf(,%eax,4), %edx  ; Load buf[i] into %edx; 
mov  %edx, out(,%eax,4) ; Store %edx into out[i]; may overwrite stack if out-of-bounds
...
    \end{lstlisting}
\end{minipage}
\caption{Out-of-Bounds Store in a Loop}
\label{fig:loop}
\end{figure}

Based on our analysis, we observe that in the three workloads where LLSCT exhibits lower overhead, the reduced overhead results from its incomplete protection against speculative out-of-bounds stores within loops. 
Specifically, LLSCT's protection strategy overlooks the threats introduced by the compiler when inserting constant access (CA) loads during the transformation from LLVM IR to assembly, where CA refers to memory accesses that combine constants with the program counter (PC) or stack pointer (RSP).

LLSCT protects programs through a series of LLVM passes. 
For speculative out-of-bounds stores, LLSCT prevents the exploitation of out-of-bounds data written during speculative execution by identifying source-sink pairs in the program and inserting fences between them in an LLVM IR pass.
One type of source-sink pairs consists of NCA stores paired with CA loads.
While LLSCT can mitigate speculative out-of-bounds risks at the LLVM IR level, the LLVM documentation \cite{LLVM} highlights that the compiler can introduce CA accesses when lowering LLVM IR to assembly.
Our experiments confirm that the compiler indeed generates CA loads, potentially creating new source-sink pairs and leading to Spectre vulnerabilities that LLSCT cannot protect against.

In the C-style loop in \Cref{fig:loop}, where the loop variable \texttt{i} is considered in a register, each iteration contains only one NCA store (\texttt{out[i]}) and one NCA load (\texttt{buf[i]}), which LLSCT does not recognize as source-sink pairs.
As a result, no fence is inserted within the loop, even though \texttt{out[i]} executes an out-of-bounds store during speculative execution when $\texttt{i} \ge \texttt{len}$.
This strategy is reasonable because when \texttt{i} is treated as a register, the address of \texttt{buf[i]} remains unaffected by the out-of-bounds store. 
Therefore, the value stored by the speculative out-of-bounds store does not leak via the cache side channel within the loop.

However, such security guarantees are compromised when lowering code from a higher level to assembly. 
When using lower optimization levels or when there are insufficient registers, the program retrieves the loop variable from the stack, as illustrated in the assembly-style code in \Cref{fig:loop}. 
Consequently, during misspeculative execution, the out-of-bounds store \texttt{out[i]} may overwrite values in the stack, including the loop variable \texttt{i}. 
If the region of the stack used to store the loop variable is overwritten by \texttt{out[i]} with a sensitive value \texttt{key}, the loop variable loaded in the next iteration will be \texttt{key}, and \texttt{key} will be leaked through the cache side channel during the memory access \texttt{in[i]} (e.g, \texttt{in[key]}).
Based on this insight, we successfully construct a one-bit side channel with an accuracy of 83\% by exploiting the primitive of speculative out-of-bounds stores in loops that rewrite variables on the stack.
Such side channel confirms this security threat.

In contrast, the security provided by LightSLH remains unaffected by compiler-introduced CA accesses, as the protections in LightSLH directly prevent such out-of-bounds stores from being misspeculatively executed. 
We examine the two cases where LightSLH exhibits higher overhead and find that \textit{all} involve protections applied to prevent such out-of-bounds stores.
Thus, LightSLH pinpoints all vulnerabilities without false positives in these two cases.

Note that in cases where LightSLH introduces 0\% overhead, the source codes also utilize loops to write data to a contiguous address space. 
However, since the number of iterations is constant, the compiler will unroll these loops, and the generated codes will contain no loops. 
This ensures that the addresses involved in the memory access operations can be determined at compile time, precluding out-of-bounds stores during speculative execution. 

\readded{Recent concurrent works \cite{arranz2025preservation,van2025snip} have highlighted that compiler register allocation can lead to speculative out-of-bounds writes overwriting other arrays \cite{arranz2025preservation} or stack addresses allocated for local variables \cite{van2025snip}. Such overwriting introduces data dependencies absent in high-level code, undermining Spectre defenses implemented at higher abstraction levels. 
Similar to our findings, \cite{arranz2025preservation} also identifies that such out-of-bounds writes within loops may write sensitive information to other address spaces. 
While in their examples, the writes within the loop are leaked by load instructions outside the loop, we further demonstrate (with a Poc) that such out-of-bounds instructions can cause data leakage within the loop itself. This underscores the necessity of applying protections directly at the out-of-bounds write.}
\readded{Nevertheless, assuming a property-preserving compiler in the future, both LLSCT and LightCut hold promise as low-overhead solutions.
}

\subsection{Case Study: RSA}\label{sec:rsa}
    \begin{lstlisting}[float=false, style=Cstyle2,basicstyle=\footnotesize, escapechar=|, captionpos=t, 
    caption={RSA in OpenSSL 1.0.2f, secrets are \underline{underlined}.}, label={lst:rsa}]
void gather ( char* buf, char* p, int k, int window ){
for ( i = 0; i < N; i++ )
    p[i] = buf[k + i * window];
}
int BN_mod_exp_mont_consttime ( ..., BIGNUM *key, ... ){
...
bits = BN_num_bits ( key ); ...
width = BN_window_bits_for_ctime_exponent_size ( bits );
/* width is assigned a value ranging in [1,6] based on bits */
window = 1 << width; ...
gather ( buf, p, wvalue, window); ...
/* buf is aligned to the cache line boundary, p is a buffer to store gathered values, wvalue is the index to be gathered, ranging in [0, window - 1] */
}       
\end{lstlisting}

\readded{In this case study, we showcase a pre-Spectre era side-channel issue, which demonstrates the effectiveness of the bit-level taint tracking mechanism in \LightSLH.} 
Analyzing RSA comprehensively and efficiently is a challenging task due to its intricate data structures (e.g., linked lists), convoluted function calls (e.g., recursive calls), and cache-aware code design (e.g., scatter-gather method).
Methods like symbolic execution \cite{spectector2020,pitchfork2020,hunter2021} struggle to scale well for analyzing RSA. 
In the analysis of \LightSLH, the value domain is adept at handling these intricate data structures, and the monotone framework proves instrumental in efficiently analyzing such complex function calls.
\redeleted{In this case study, we will demonstrate how our taint mechanism enables us to perform program analysis in a more precise manner.
}

\readded{Previous studies \cite{cachebleed2017,cache-audit-rigorous2017,caches-2019} suggested that the scatter-gather code in RSA exhibits constant behavior at the cache-line granularity. 
Based on this assumption, such code was considered not to leak sensitive information and, therefore, might not require protection under speculative execution, potentially reducing the overhead of defenses. 
To analyze this cache-aware code and evaluate the potential for reduced overhead, we employed bit-level taint tracking and conducted experiments under two settings: cache-line granularity and full address observation. 
However, our experimental results reveal that the protection overhead remains consistent across both settings.
Upon analyzing the results, we identify, for the first time, that even for an attacker who can only observe memory accesses at cache line granularity, the observations generated by the scatter-gather algorithm  also depend on secrets.
This finding highlights the critical need for protection, even in scenarios constrained to cache-line granularity observations.}

\redeleted{Our work presents a rigorous analysis that demonstrates, for the first time, that even for an attacker who can only observe memory accesses at cache line granularity, the observations generated by the scatter-gather algorithm  also depend on secrets.
According to previous work \cite{cachebleed2017,cache-audit-rigorous2017,caches-2019}, the memory access patterns generated by the scatter-gather algorithm are considered to be independent of secret information at the granularity of cache lines, thus these memory accesses should not be hardened.
While  \cite{caches-2019} flags memory accesses in scatter-gather as potentially vulnerable to cache attacks, it attributes these flags to false positives caused by analysis imprecision.}
\readded{Among previous analyses of RSA, while \cite{caches-2019} flags memory accesses in scatter-gather as potentially vulnerable to cache attacks, it attributes these flags to false positives caused by analysis imprecision.}
Analysis by  \cite{cache-audit-rigorous2017} concludes that the scatter-gather algorithm does not leak information to attackers capable of observing cache lines. 
This conclusion stems from its analysis of the core code (as presented in \Cref{lst:scatter-gather}) of the scatter-gather algorithm with assumed input parameters, rather than examining the entire \texttt{BN\_mod\_exp\_mont\_consttime} function.

A simplified code of \texttt{BN\_mod\_exp\_mont\_consttime} is presented in \Cref{lst:rsa} with secret-dependent variables marked with underline.
The variable \texttt{bits}, which is computed from the secret value \texttt{key}, represents the length of the big number corresponding to \texttt{key}.
Depending on the value of \texttt{bits}, \texttt{width} ranges from 1 to 6. 
After taint propagation and sanitization, \texttt{window} is assigned a taint vector where the least significant $7$ bits are \thigh and the others are \tzero.
Specifically, the seventh least significant bit of \texttt{window} is set to 1 only when \texttt{width} equals 6, and 0 otherwise.
This indicates that the value of the seventh bit depends on the secret, and our computation accurately reflects this dependency by setting the seventh bit to \thigh.
When \texttt{window} is passed as an argument to the \texttt{gather} function and the memory access \texttt{buf[k+i*window]} is performed, the seventh bit of \texttt{window} influences the cache line index. 
This implies that at the cache line granularity, the memory access pattern is also dependent on  secrets.
Consequently, the memory access in \texttt{gather}  also requires protection.

Alternatively, when setting \texttt{width} to a fixed value of $6$, our analysis of the \texttt{gather} function reveals that its memory access pattern at the cache line granularity remains constant, consistent with the findings of previous work \cite{cache-audit-rigorous2017}. 
This further demonstrates \LightSLH's capability to conduct rigorous analysis of cache-aware programs. While the precise analysis enabled by bit-level taint tracking does not reduce the overhead for RSA in OpenSSL 1.0.2f, it facilitates more accurate analysis and helps eliminate unnecessary hardening in other programs.




\section{Discussion}\label{sec:discussion}
\textbf{Security Scope of LightSLH.}
LightSLH hardens instructions vulnerable to Spectre v1, including memory access and branch instructions that may leak secret data through side-channel attacks \cite{spectre2019} and instructions that perform out-of-bounds memory stores \cite{spectre-overflow2018} during speculative execution.

While USLH \cite{uslh2023} addresses time-variable instruction leaks, LightSLH does not currently implement such protections because: (1) these specific gadgets are unlikely in real-world software as noted by \cite{uslh2023}, and (2) LightH is inherently flexible and can be extended to accommodate such scenarios by extending the $H()$ function to check \thigh labels on time-variable instructions based on the target architecture's specifications.

Regarding Spectre declassified \cite{declassified2023}, which concerns speculative leakage at declassification sites, LightSLH does not explicitly handle such cases. While these leaks can be mitigated by inserting hardening at declassification points (typically at cryptographic primitive return instructions), we opted not to incorporate declassification semantics as they are peripheral to our core methodology.


\textbf{Preservation.}
We implement \LightSLH as an LLVM post-optimization pass.
\redeleted{While security properties verified in LLVM IR may be compromised during the compilation process to binary code, we consider  hardening at LLVM IR level sufficient for ensuring security.
This is because the lowering from LLVM IR to binary code only introduces memory accesses with PC or RSP values combined with constants \cite{LLVM,serberus2024}, which consistently generate constant observations and do not cause out-of-bounds stores. 
Nevertheless, formal verification of security property preservation during compilation remains important, which can be built upon verified compilers like CompCert \cite{leroy2016compcert}, and is left as future work.}
\readded{While LightSLH's hardening at LLVM IR level do not suffer from the compiler-introduced memory access as described in \Cref{sec:chacha}, its protection may still be compromised by other compilation transformation such as dead code elimination and register allocation \cite{arranz2025preservation,van2025snip}. Proving preservation for compilers is both important and challenging, which goes beyond the scope of our research and is left as future work.}In addition, our methodology can also be implemented for binary programs. 
The decision to implement at the LLVM IR level is due to LLVM IR's single static-assignment form, which simplifies data flow analysis and facilitates code transformations.

\textbf{Protection for Non-Constant-Time Programs.}
While existing hardening tools \cite{serberus2024,typing2023,declassified2023} focus on protecting constant-time programs, we argue that non-constant-time programs also require protection against speculative execution attacks.
Many implementations intentionally deviate from strict constant-time requirements due to implementation complexity and performance considerations, accepting certain controlled information leaks.
For example, \openssl's RSA implementation deliberately leaks big number lengths as an acceptable trade-off.
The emergence of Spectre attacks raises a critical question: could these vulnerabilities amplify existing controlled leaks into more severe security breaches?
\redeleted{To address this concern, automated hardening mechanisms that enforce speculative non-interference ensure that speculative execution cannot reveal more information than what is already disclosed during sequential execution.
}
\readded{In our work, properties including SNI and SS, and formalization does not require constant-time execution in the sequential model. This means that even for non-constant-time code, LightSLH ensures the SNI property, guaranteeing that speculative execution cannot reveal more information than what is already disclosed during sequential execution.}


\readded{\textbf{Protection against Spectre RSB and STL.}}
\curdeleted{For Spectre RSB, although original attacks \cite{ret2spec2018} have been mitigated through hardware protections and microcode updates, 
subsequent researches \cite{barberis2022branch,wikner2022retbleed,wikner2024breaking} have uncovered additional vulnerabilities in the prediction mechanisms for return instructions, which can cause programs to execute instructions at arbitrary locations.
For these vulnerabilities, hardware-level protections \cite{intel:bhi,intel:ibrs,intel:ibpb,intel:rsb} remain the primary defense mechanism.} 
\added{For Spectre-RSB, recent studies~\cite{barberis2022branch,wikner2022retbleed,wikner2024breaking} have shown that, when the RSB underflows, the processor may fall back to alternative prediction mechanisms to generate the target address for a \texttt{ret} instruction.
While the initial RSB attacks~\cite{ret2spec2018} exploit overflows to redirect speculation to return addresses from other function calls, these newer attacks can trigger speculative execution at arbitrary code locations.
A commonly proposed mitigation is to ``stuff the RSB'' before executing a \texttt{ret} instruction by pre-filling it with safe return targets \cite{intel:bhi,wikner2024breaking}}.

\readded{Regarding Spectre STL (v4) \cite{horn2018speculative}, while SLH's mitigation for Spectre v1 is known to incur significant performance overhead when disabling STL predictions via SSBD \cite{serberus2024}, our framework could potentially model STL's prediction mechanisms in the future to provide software-based defenses. 
Specifically, this could involve analyzing load instructions by performing a lub operation between the values from preceding stores and the loaded value from memory to determine the possible values accessed under the Spectre v4 prediction model. 
However, this extension lies beyond the scope of our current work and is left as future research.}



\section{Related Work}\label{sec:related-work}
\textbf{Speculative Vulnerability Detection.}
Existing detections tools base on symbolic execution \cite{pitchfork2020,spectector2020,hunter2021,automatic-detection2022},
type system \cite{typing2023}, model checking \cite{cats2022} and abstract interpretation \cite{ai-spec2019}.
These tools, though sound, can only determine whether a program is secure, rather than identifying all potentially vulnerable instructions.
Other methods include dynamic analysis \cite{spectaint2021, kasper2022, SpecFuzz2020} and gadget scanning \cite{oo72019}.
In the absence of rigorous verification, their detection methods are prone to both false positives and false negatives.

\textbf{Software Mitigation.}
Various approaches have been proposed to optimize the performance overhead of Spectre attack mitigations.
Blade \cite{blade2021} employs a type system to classify expressions as transient or stable, preventing speculative leaks by determining minimal protection placements that block information flow from transient values to transmitters.
However, Blade is proven incomplete \cite{declassiflow2023}, as it only protects speculatively-accessed values while failing to prevent speculative leaks of sequentially-accessed values \cite{declassiflow2023}.

LLSCT \cite{serberus2024} mitigates Spectre attacks by eliminating dependencies from four classes of taint primitives. However, it requires programs to be \textit{static constant-time}, limiting its applicability (e.g., for RSA), and its security guarantees may not hold after compilation to binary code \Cref{sec:chacha}.


SelSLH \cite{declassified2023} uses a type system (L/H types) to identify expressions dependent on secrets and selectively protects load instructions of type L.
While efficient for array accesses, SelSLH has limited support for pointer operations and, like LLSCT \cite{serberus2024}, only works with constant-time programs.

Declassiflow \cite{declassiflow2023} avoids incurring multiple overhead penalties for a single repeated hardening by carefully relocating protection to its ``knowledge frontier'', where the protected value will inevitably be leaked.
LightSLH is orthogonal to Declassiflow, thus integrating Declassiflow's strategy in the $H()$ function of LightH may result in lower overhead, for which we leave as future work.

Jasmin \cite{jasmin,jasmin-ct,jasmin-spectre,typing2023} employs type systems to ensure speculative constant-time security in cryptographic implementations. However, it is limited to the Jasmin language and cannot automatically repair type-checking failures.

\section{Conclusion}

In this paper, we propose LightH, a framework designed to apply targeted hardening to protect programs against Spectre v1 vulnerabilities.
LightH is flexible, allowing for the modeling of different attacker scenarios and the integration of various hardening strategies.
Based on SLH hardening strategy and the cache-based attacker model, we implement LightH as LightSLH, which hardens program with provable security.
Our results demonstrate the efficiency of LightSLH's analysis and its low-overhead protection.
Through our experiments, we discover two previously unknown security issues, highlighting the necessity of applying accurate protections to certain instructions.


\section{Ethics considerations}
This work adheres to ethical research practices, with all resources and tools utilized in this study sourced from open-source repositories. Specifically, the code framework, cryptographic libraries, and benchmarks employed in our experiments were obtained from publicly available and widely recognized open-source repositories. The use of these resources complies with their respective licenses and ensures reproducibility.

Our research focuses on defensive measures and aims to enhance security by identifying and mitigating potential vulnerabilities. In particular, we uncovered two previously unreported instances of information leakage through side channel. While these instances do not currently constitute a complete exploit, they highlight areas where further investigation and proactive defenses are warranted.

As a defensive work, our approach emphasizes a rigorous security model and conservative protective measures. This is justified by the goal of preemptively addressing potential vulnerabilities before they evolve into more severe threats. By sharing these findings, we aim to contribute to the broader security community’s efforts in strengthening cryptographic software against emerging attack vectors.

\section{Open Science}
The artifact for this work is publicly available at https://doi.org/10.5281/zenodo.15569395. It includes the following components:

\textbf{Tools. }
The LightFence and LightSLH tools developed in this research are included in the Zenodo repository. Comprehensive documentation is provided to facilitate ease of use and integration into other projects.

\textbf{Experimental Data. }
All experimental data, including benchmarks and results, are shared in a well-structured and accessible format, enabling reproducibility and further validation by the research community.

\textbf{Reproducibility. }
The repository contains detailed, step-by-step instructions for replicating the experiments presented in this paper. This ensures that researchers can independently validate our findings and build upon this work.

We are committed to maintaining these resources for long-term accessibility and welcome feedback and collaboration from the community to further advance this research.

\section*{Acknowledgments}
We would like to thank our shepherd and reviewers for their helpful comments. The research is supported by the National Natural Science Foundation of China under Grant No. 62372422, No.61972369, No.62272434, and the Fundamental Research Funds for the Central Universities (Grant Number: WK2150110024).



\bibliographystyle{plain}
\bibliography{bib/cite}
\onecolumn
\appendix

\section{Details of Implementation}\label{sec:impl}


\subsubsection*{Interprocedural Analysis}\label{sec:impl-inter}
Our interprocedural analysis performs in a context-sensitive way.
For non-recursive calls, we set the callee's entry states (including memory state and states of arguments) based on the current context.
The analysis is then performed on the callee function, followed by collecting its possible leaving states (including memory and return value).
The lub of these states is returned to the caller.
Recursive calls are treated as control flow.
Specifically, we update the callee's entry state with the lub of its current entry state and the state at the call site.
Then, the entry block is marked pending to force the analyzer to reprocess it with the updated state.

To further improve efficiency, we record a summary of each function's analysis, including  entry and leaving states.
When encountering subsequent calls to the same function with the same entry state, we can skip re-analyzing it and directly use the recorded return state.
This significantly reduces the analysis time. 
For example, applying this method to the analysis of RSA in OpenSSL eliminates 90\% of redundant re-evaluations.

\subsubsection*{Convergency}
The taint domain is evidently finite, and given the upper and lower bounds of all 64-bit integers, the value domain is finite as well. 
The finiteness of these two lattices ensures the convergence of fixpoint computation.
However, convergence in loops can be slow due to the size of the value domain.
For instance, for a loop like \texttt{for (i=0; i<16; i++)\{...\}}, computing a fixpoint requires analyzing the loop 16 times before the abstract value of the loop variable \texttt{i} converges.
To improve analysis efficiency, we directly use the loop's boundary conditions to determine the abstract value of the loop variable in \textsc{TransSeq}, while setting its abstract value to \topof{\VD} in \textsc{TransSpec}.
In the previous example, \texttt{i} has an abstract value of $(\emptysym,[0,15])$ when  using \textsc{TransSeq}, and to \topof{\VD} when using \textsc{TransSpec}.

\readded{\subsubsection*{Memory Representation}
In contrast to theoretical models, real-world memory access often deals with variable-sized data. 
A single memory access might span multiple addresses.
For instance, in 64-bit systems, pointers (typically 8 bytes) and 64-bit integers are stored in 8 consecutive memory addresses.
However, byte-granular access remains possible. 
For example, the AES implementation within OpenSSL entails reading a 64-bit integer in a bytewise manner.
This raises the question: how should we handle byte-granular access for memory units that should be treated as a whole?
}

\readded{In our implementation, we track both the size and an alignment attribute for memory regions and pointers. 
When an aligned pointer accesses an aligned memory region with matching sizes, load and store operations retrieve the corresponding values from the memory region. 
For unaligned regions, we only record the highest taint label, denoted as $t$. 
A value loaded from an unaligned region or accessed by an unaligned pointer is assigned $(\topof{\VD},\vec{t})$, while write operations to unaligned regions only update $t$.
}

\readded{An additional challenge  arising in memory operations is the computational cost of analysis.
For instance, processing a load instruction from array \texttt{a}, where the address has an abstract value of $(\texttt{a},[0,15])$, requires 16 lub computations, which is inefficient.
To mitigate it, we maintain and update a single lub for all elements in an array. 
This optimization ensures that each memory access to an array involves just one lub computation.
}
\readded{\subsubsection*{Implementation-Model Gap}}

\readded{
While our formalization is based on \muasm and our implementation on LLVM IR, notable differences emerge due to the diverse instruction set in LLVM IR. 
LLVM IR introduces additional instructions, such as cast and vector operations, which require special handling. 
To address these, we record the size of each operand and, if the operand is a pointer, its alignment attribute, employing conservative methods to compute their states.
For example, when casting an `i8*' pointer to a pointer of a struct type, we disable the alignment attribute, set the value domain of the casted pointer to $\top$, and assign its taint domain to the highest label within the struct region.
For vector operations, we conservatively record the least upper bound (lub) of all elements in the vector. 
This approach minimizes analysis overhead while ensuring the soundness of the analysis.}

\section{Operator Semantics For Taint Tracking}\label{app:sec:taint-operation}
We set the result of operations to $\vec{\tbot}$ when \tbot appears in any operand, thus when introducing the rules in this section, we will set aside situations where the taint label evaluates to \tbot.

Let $a=\simplevector{a}{n-1}{n-2}{0}$ and $b=\simplevector{b}{n-1}{n-2}{0}$ be two taint label vectors with $n$ elements. 
We denote by $\bopondomain{\op}{\prolattice{T}{n}}{a}{b}=\simplevector{r}{n-1}{n-2}{0}$ the operations on lattice \prolattice{T}{n}, where $\op \in \uop \cup \bop$.
For a unary operator, the second operand $b$ can be omitted.
With a slight abuse of notation, we use \tzero, \tone, $0$, and $1$ to represent both taint labels and the corresponding values.

Let \counttupleup{t}{(x_1,\cdots,x_m)} denote the count of $x_i$ such that $t\sqsubseteq x_i$.
Let $\min_{t\sqsubseteq i}(a)$ denote the least index in a taint vector $a$ such that $t\sqsubseteq a_i$. 
In particular, $\min_{t\sqsubseteq i}(a)=\text{len}(a)$ if there is no index $i$ such that $t\sqsubseteq a_i$.
For a taint vector $a$, $\num(a)$ is a concrete value defined as $\sum_{k=0}^{\min_{\tlow\sqsubseteq i}(a)-1} 2^{k}a_k$. 
For a concrete number $n$, $n_i$ denotes the digit at the $i$-th position in the binary representation of $n$, i.e., $\overline{n_kn_{k-1}\cdots n_0}$.

The operator semantics are given below:

\begin{itemize}
    \item $\op=\notop$.
    $$
    r_i = 
      \begin{cases}
          r_i, & \text{if} \ a_i \in \{L,H\}\\
          \tone,& \text{if}\ a_i = \tzero\\
          \tzero,& \text{if}\ a_i = \tone
      \end{cases}
    $$
    \item $\op=\andop$.
    $$
    r_i = 
      \begin{cases}
          \tzero & \text{if} \ a_i =\tzero\ \text{or}\ b_i = \tzero\\
          \tone & \text{else if} \ a_i =\tone\ \text{and}\ b_i = \tone\\
          a_i \sqcup b_i &\text{otherwise}
      \end{cases}
    $$
    \item $\op=\orop$.
    $$
    r_i = 
      \begin{cases}
          \tone & \text{if} \ a_i =\tone\ \text{or}\ b_i = \tone\\
          \tzero & \text{else if} \ a_i =\tzero\ \text{and}\ b_i = \tzero\\
          a_i \sqcup b_i &\text{otherwise}
      \end{cases}
    $$
    \item $\op=\xorop$.
    $$
    r_i = 
      \begin{cases}
          a_i\ \xorop\ b_i & \text{if $a_i$ and $b_i$ are both concrete labels}\\
          a_i \sqcup b_i &\text{otherwise}
      \end{cases}
    $$
    \item $\op=\addop$.
    $$
    c_i = 
    \begin{cases}
        \tzero & \text{if\ \ $\counttupleup{\tone}{(a_{i-1},b_{i-1},c_{i-1})} \le 1$ or $i=0$} \\
        \thigh  & \text{else if\ \ $\thigh\in\{a_{i-1},b_{i-1},c_{i-1}\}$} \\
        \tlow & \text{else if\ \ $\tlow\in\{a_{i-1},b_{i-1},c_{i-1}\}$} \\
        \tone & \text{otherwise}
    \end{cases}
$$
$$
    r_i =
    \begin{cases}
        \thigh  & \text{if\ \ $\thigh\in\{a_i,b_i,c_{i}\}$} \\
        \tlow & \text{else if\ \ $\tlow\in\{a_i,b_i,c_{i}\}$} \\
        (a_i+b_i+c_{i})(\text{mod }2) & \text{otherwise}
    \end{cases}
$$

    where $c_i$ denotes the carry bit.
    \item $\op=\minusop$.
    $$
        c_{i+1} = 
        \begin{cases}
            \tzero & \text{if\ \ ($\counttupleup{\tone}{(b_i,c_{i})} \le 1$ and $a_i = \tone$) or $i=0$} \\
            \thigh  & \text{else if\ \ $\thigh\in\{a_i,b_i,c_{i}\}$} \\
            \tlow & \text{else if\ \ $\tlow\in\{a_i,b_i,c_{i}\}$} \\
            \tone & \text{otherwise}
        \end{cases}
    $$
    $$
        r_i =
        \begin{cases}
            \thigh  & \text{if\ \ $\thigh\in\{a_i,b_i,c_{i}\}$} \\
            \tlow & \text{else if\ \ $\tlow\in\{a_i,b_i,c_{i}\}$} \\
            (a_i-b_i-c_{i})(\text{mod }2) & \text{otherwise}
        \end{cases}
    $$
    where $c_i$ denotes the carry bit.
    \item $\op = \mulop$.
    $$
        r_i =
        \begin{cases}
            \thigh  & \text{if $i\ge \min(\minintv{\thigh}(a)+\minintv{\tone}(b), \minintv{\thigh}(b)+\minintv{\tone}(a))$ }\\
            \tlow  & \text{else if $i\ge \min(\minintv{\tlow}(a)+\minintv{\tone}(b), \minintv{\tlow}(b)+\minintv{\tone}(a))$ }\\
            (\num(a)\times\num(b))_i  & \text{otherwise}\\
        \end{cases}
    $$    
    \item $\op = \divop$.
    $$
        r_i = \begin{cases}
            \thigh & \text{if $\thigh \in a$ or $\thigh \in b$}\\
            \tlow  & \text{else if $\tlow \in a$ or $\tlow \in b$}\\
            (\num(a)\div \num(b))_i  & \text{otherwise}\\        
        \end{cases}
    $$    
    \item $\op = \modop$.
    $$
        r_i = \begin{cases}
            \thigh  &\text{if $\thigh \in a$ or $\thigh \in b$}\\
            \tlow  &\text{else if $\tlow \in a$ or $\tlow \in b$}\\
            (\num(a)\% \num(b))_i  & \text{otherwise}\\        
        \end{cases}
    $$    
    \item $\op = \lshiftop$.
    \begin{itemize}
        \item If $\tlow,\thigh \notin b$,
        $$
            r_i = \begin{cases}
                a_{i+\num(b)} & \text{if $i+\num(b)<n$}\\
                \tzero & \text{otherwise}
            \end{cases}
        $$
        \item  If $\tlow \in b$ or $\thigh \in b$,
        $$
            r_i = \begin{cases}
                \thigh & \text{if $i\ge \num(b)+\minintv{\thigh}(a)$} \\
                \tlow & \text{else if $i\ge \num(b)+\minintv{\tone}(a)$} \\
                \tzero & \text{otherwise}
            \end{cases}
        $$
    \end{itemize}
    \item $\op = \rshlop$.\\
    Let $\rev{a}=\simplevector{a}{0}{1}{n-1}$,
    \bopondomain{\rshlop}{\prolattice{T}{n}}{a}{b} = \rev{\bopondomain{\lshiftop}{\prolattice{T}{n}}{\rev{a}}{b}}.
    \item $\op = \rshaop$.
    Let $\rev{r}=\simplevector{r}{0}{1}{n-1}=\simplevector{r'}{n-1}{n-2}{0}$, then $r'_i$ is given by following rules.
    \begin{itemize}
        \item If $\tlow,\thigh \notin b$,
        $$
            r'_i = \begin{cases}
                \rev{a}_{i+\num(b)} & \text{if $i+\num(b)<n$}\\
                a_{n-1} & \text{otherwise}
            \end{cases}
        $$
        \item  If $\tlow \in b$ or $\thigh \in b$,
        $$
            r'_i = \begin{cases}
                \thigh & \text{if $i\ge \num(b)+\minintv{\thigh}(\rev{a})$} \\
                \tlow & \text{else if $i\ge \num(b)+\minintv{\uopondomain{\notop}{\lattice{T}}{a_{n-1}}}(\rev{a})$} \\
                a_{n-1} & \text{otherwise}
            \end{cases}
        $$
    \end{itemize}
\end{itemize}

Given the rules above, we have
\begin{restatable}{theorem}{welldefined}\label{thm:well-defined}
    For any operator $\op$ in \muasm{} with  \bopondomain{\op}{\prolattice{T}{n}}{a}{b} defined in our work, \opondomain{\op}{\prolattice{T}{n}} is a well-defined operator on \prolattice{T}{n}.
\end{restatable}
\begin{proof}
    For the first requirement in the definition of well-defined operators, we need to prove that for any $ t_1 , t_2 \in \prolattice{T}{n}$, and \isinstance{v_1}{t_1}, \isinstance{v_2}{t_2}, we have $\isinstance{(v_1\op v_2)_i}{r_i}$ for all $0\le i\le n-1$, where $r =\bopondomain{\op}{\prolattice{T}{n}}{t_1}{t_2} $. \hfill\mbox{(*)}

    \begin{itemize}
        \item  For $\op=\notop,\andop,\orop,\xorop,\divop\text{ and }\modop$, the proof is straightforward.
        \item For $\op=\addop$, let $s_i$ denote the carry bit, we prove $\isinstance{s_i}{c_i}$ and (*) holds. \hfill\mbox{(**)}
        
        It is obvious that (**) holds for $i=0$. 
        Assume (**) holds for $i=k(k\ge 0)$. 
        For $i=k+1$, if $c_{k+1}=0$, then $\counttupleup{\tone}{(v_{1,k},v_{2,k},c_{k})} \le 1$ holds. 
        Considering the inductive hypothesis, we have that at most one of $v_{1,k}$, $v_{2,k}$ and $s_k$ equal to 1, which indicates that $s_{k+1}=0$, thus \isinstance{s_{k+1}}{c_{k+1}} holds.
        If $c_{k+1}=1$, there is no \thigh and \tlow in $v_{1,k}$, $v_{2,k}$ and $s_k$, which indicates that at least two of them equal to 1. So $s_{k+1}$ equals to 1 and \isinstance{s_{k+1}}{c_{k+1}} holds.
        Given \isinstance{s_{k+1}}{c_{k+1}}, it is obvious that \isinstance{(v_1\op v_2)_{k+1}}{r_{k+1}}.

        Therefore, by mathematical induction, (**) holds.

        The case for $\op=\minusop$ is similar to $\op=\addop$.
        \item For $\op=\mulop$,  
        $$v_1\cdot v_2 = (\sum_{i=0}^{n-1} 2^i v_{1,i})(\sum_{i=0}^{n-1} 2^i v_{2,i}) = \sum_{i=0}^{2n-2}2^i(\sum_{j=0}^{i}v_{1,j}v_{2,i-j})$$
        So the $i$-th bit of $v_1\cdot v_2$ is determined by only $v_{1,j}\cdot v_{2,i-j}(0\le j\le i)$.
        Consequently, for $0\le i< s$, where $s=\min(\minintv{\tlow}(v_1)+\minintv{\tone}(v_2), \minintv{\tlow}(v_2)+\minintv{\tone}(v_1))$, \isinstance{(v_1\op v_2)_i}{r_i} follows from  \isinstance{v_1}{t_1} and \isinstance{v_2}{t_2}.
        \item For $\op=\lshiftop$, if $b$ does not contain \thigh and \tlow, the number of bits to shift left is concrete and thus the proof is straightforward.
        Otherwise, the least $s$ bits of $v_1\op v_2$ equal to $0$, where $s = \num(v_2)+\minintv{\tone}(v_1)$. 
        This is because that $v_1$ is left-shifted by at least $\num(v_2)$ bits, and the least $\minintv{\tone}(v_1)$ bits of $v_1$ is $0$.
        Thus $\isinstance{(v_1\op v_2)_i}{r_i}$ holds for all $0\le i\le n-1$.

        The case for $\op=\rshlop$ and $\rshaop$ is similar to $\op=\lshiftop$.

    \end{itemize}

    For the second requirement, we only need to prove that given any $ t_1, t_2, t_1^\tlow, t_2^\tlow \in \prolattice{T}{n}$, there exists an $r^\tlow$, such that for any \isinstance{v_1}{t_1^\tlow}, \isinstance{v_2}{t_2^\tlow}, $\isinstance{(v_1\op v_2)}{r^\tlow}$ holds for all $0\le i\le n-1$, where $t^\tlow$ can be obtained by replacing all \tlow labels in $t$ with concrete labels.

    The proof of the second requirement can be obtained using a similar approach as the proof of the first requirement.
\end{proof}
\section{Operational Semantics}\label{app:semantics}

\begin{figure*}
    \begin{center}
        \scriptsize
        \typerule{Asgn}{
            p(\rho(\pc)) = \passign{x}{e} & \rho' = \updatemap{\rho}{x}{\exprEval{e}{\rho}}
        }{
            (p, \conf{\rho, f}) \eval{\DirStep}{\Obsnone} (p, \conf{\updatemap{\rho'}{\pc}{\rho(\pc)+1}, f})
        }{semantic:assign}
        \typerule{Jmp}{  
            p(\rho(\pc)) = \pjmp{l} & \rho' = \updatemap{\rho}{\pc}{l}
        }{
            (p, \conf{\rho, f}) \eval{\DirStep}{\Obsnone} (p, \conf{\rho', f})
        }{semantic:jump}
        \typerule{Alloc}{
            p(\rho(\pc)) = \palloc{x}{n} \\ \rho' = \rho[x\mapsto\rho(\mem), \mem \mapsto \rho(\mem)+n]
        }{
            (p, \conf{\rho, f}) \eval{\DirStep}{\Obsnone} (p, \conf{\rho'[\pc \mapsto \rho(\pc)+1], f})
        }{semantic:alloc}
        \typerule{Fen}{
             p(\rho(\pc)) = \pfence & 
             l = {\begin{cases}
                \bot & \text{if $f = \top$}\\
                \rho(\pc) + 1 & \text{if $f = \bot$}
             \end{cases}}
         }{
             (p, \conf{\rho, f}) \eval{\DirStep}{\Obsnone} (p, \conf{\updatemap{\rho}{\pc}{l}, f})
         }{semantic:fence}
         \typerule{Ld}{
             p(\rho(\pc)) = \pload{x}{e} & n = \exprEval{e}{\rho} \\
            \rho' = \rho[x\mapsto \rho(n), \pc\mapsto\rho(\pc)+1]
         }{
             (p, \conf{\rho, f}) \eval{\DirStep}{\Obsload{n}{a}{b}} (p, \conf{\rho', f})
         }{semantic:load}
         \typerule{St}{
             p(\rho(\pc)) = \pstore{x}{e} & n = \exprEval{e}{\rho} \\  
            \rho' = \rho[n\mapsto \rho(x), \pc\mapsto\rho(\pc)+1]
         }{
             (p, \conf{\rho, f}) \eval{\DirStep}{\Obsstore{n}{a}{b}} (p, \conf{\rho', f})
         }{semantic:store}
         \typerule{CndAsgn}{
             p(\rho(\pc)) = \pcondassign{x}{e}{e'} \\
             \rho' = {
                \begin{cases}
                    \rho & \text{if \exprEval{e'}{\rho} = 0}\\
                    \updatemap{\rho}{x}{\exprEval{e}{\rho} } & \text{if \exprEval{e'}{\rho} $\neq$ 0}
                \end{cases}
             }
         }{
             (p, \conf{\rho, f}) \eval{\DirStep}{\Obsnone} (p, \conf{\updatemap{\rho'}{\pc}{\rho(\pc)+1}, f})
         }{semantic:conditional-assign}
         \typerule{Br-Step}{
            p(\rho(\pc)) = \pbranch{x}{l} \\ n = \rho(x) &
            l' = {\begin{cases}
                \rho(\pc) + 1 & \text{if $n \neq 0$}\\
                l & \text{if $n = 0$}
            \end{cases}}
        }{
            (p, \conf{\rho, f}) \eval{\DirStep}{\Obsbranch{n}} (p, \conf{\updatemap{\rho}{\pc}{l'}, f})
        }{semantic:branch-step}
        \typerule{Br-Force}{
            p(\rho(\pc)) = \pbranch{x}{l} \\ n = \rho(x) &
            l' = {\begin{cases}
                \rho(\pc) + 1 & \text{if $n= 0$}\\
                l & \text{if $n\neq 0$}
            \end{cases}}
        }{
            (p, \conf{\rho, f}) \eval{\DirForce}{\Obsbranch{n}} (p, \conf{\updatemap{\rho}{\pc}{l'}, \top})
        }{semantic:branch-force}
        \vspace{-1em}
        \caption{Speculative Semantics of \muasm.}
        \label{fig:semantics}
    \end{center}
\end{figure*}    

The operational semantics of \muasm{} is given in \Cref{fig:semantics}.

\section{Speculative Semantics with Taint Tracking}\label{app:sec::full-semantics}
The speculative semantics with taint tracking mechanism is presented in \Cref{fig:tsemantics}.
Notably, $\pc$ and $\mem$ are labeled as $\vec{\tlow}$.
\vspace{2em}
\begin{figure*}[!h]
    \begin{center}
        \small
        \typerule{Asgn}{
            p(\rho(\pc)) = \passign{x}{e} 
            \\ \mu' = \updatemap{\mu}{x}{\exprEval{e}{\mu}}
            \\ \rho' = \rho[x\mapsto\exprEval{e}{\rho},\pc\mapsto\rho(\pc)+1]
        }{
            (p, \conf{\rho, \mu, f}) \eval{\DirStep}{\Obsnone:\Obsnone} (p, \conf{\rho', \mu', f})
        }{app:tsemantic:assign}
         \typerule{Ld}{
             p(\rho(\pc)) = \pload{x}{e} & n = \exprEval{e}{\rho} & t = \exprEval{e}{\mu} \\ 
            t' = {\begin{cases}
            \vec{\thigh} & \text{if $\thigh \in t$}\\
            \mu(n) &\text{if $\thigh \notin t$}
            \end{cases}} & {\begin{aligned}
                \rho'=\rho[x\mapsto \rho(n),\pc\mapsto\rho(\pc)+1] \\
                \mu'=\updatemap{\mu}{x}{t' }
            \end{aligned}}
         }{
             (p, \conf{\rho, \mu, f}) \eval{\DirStep}{\Obsload{n}{a}{b}:t_{[a,b]}} (p, \conf{\rho',\mu', f})
         }{app:tsemantic:load}
         \typerule{St}{
             p(\rho(\pc)) = \pstore{x}{e} & n = \exprEval{e}{\rho} & t = \exprEval{e}{\mu} \\ 
            t' = {\begin{cases}
            \vec{\thigh} & \text{if $\thigh \in t$}\\
            \mu(x)      &\text{if $\thigh \notin t$}
            \end{cases}} & {\begin{aligned}
                \rho'=\rho[n\mapsto \rho(x),\pc\mapsto\rho(\pc)+1] \\
                \mu'=\updatemap{\mu}{n}{t' }
            \end{aligned}}
         }{
             (p, \conf{\rho, \mu, f}) \eval{\DirStep}{\Obsstore{n}{a}{b}:t_{[a,b]}} (p, \conf{\rho',\mu', f})
         }{app:tsemantic:store}
        \typerule{Fen}{
             p(\rho(\pc)) = \pfence \\ 
             l' = {\begin{cases}
                \bot & \text{if $f = \top$}\\
                \rho(\pc) + 1 & \text{if $f = \bot$}
             \end{cases}}
         }{
             (p, \conf{\rho,\mu, f}) \eval{\DirStep}{\Obsnone} (p, \conf{\updatemap{\rho}{\pc}{l'},\mu, f})
         }{app:tsemantic:fence}
         \typerule{CondAsgn}{
             p(\rho(\pc)) = \pcondassign{x}{e}{e'} 
             \\ n = \exprEval{e}{\rho} 
             & n' = \exprEval{e'}{\rho} 
             & t = \exprEval{e}{\mu} 
             & t' = \exprEval{e'}{\mu} \\
             \rho' = {
                \begin{cases}
                    \rho& \text{if $n' = 0$}\\
                    \updatemap{\rho}{x}{n} & \text{if $n'\neq$ 0}
                \end{cases}
             }
             & \mu'={
                \begin{cases}
                    \updatemap{\mu}{x}{\vec{\thigh}} & \text{if $\thigh \in t'$}\\
                    \mu & \text{else if $n' = 0$} \\
                    \updatemap{\mu}{x}{t} & \text{otherwise }
                \end{cases}
             }
         }{
             (p, \conf{\rho,\mu, f}) \eval{\DirStep}{\Obsnone} (p, \conf{\updatemap{\rho'}{\pc}{\rho(\pc)+1}, \mu',f})
         }{app:tsemantic:conditional-assign}
         \typerule{Jmp}{  
            p(\rho(\pc)) = \pjmp{l} 
            & \rho' = \updatemap{\rho}{\pc}{l}
        }{
            (p, \conf{\rho,\mu,f}) \eval{\DirStep}{\Obsnone} (p, \conf{\rho',\mu,f})
        }{app:tsemantic:jump}
         \typerule{Br-Step}{
            p(\rho(\pc)) = \pbranch{x}{l}  
            & n = \rho(x)
            & t = \mu(x)
            \\ l' = {\begin{cases}
                \rho(\pc) + 1 & \text{if $n \neq 0$}\\
                l & \text{if $n = 0$}
            \end{cases}}
        }{
            (p,\conf{\rho,\mu, f}) \eval{\DirStep}{\Obsbranch{n}:t} (p, \conf{\updatemap{\rho}{\pc}{l'},\mu, f})
        }{app:tsemantic:branch-step}
        \typerule{Br-Force}{
            p(\rho(\pc)) = \pbranch{x}{l} 
            & n = \rho(x)
            & t = \mu(x) 
            \\ l' = {\begin{cases}
                \rho(\pc) + 1 & \text{if $n= 0$}\\
                l & \text{if $n\neq 0$}
            \end{cases}}
        }{
            (p, \conf{\rho,\mu, f}) \eval{\DirForce}{\Obsbranch{n}:t} (p, \conf{\updatemap{\rho}{\pc}{l'}, \mu,\top})
        }{app:tsemantic:branch-force}
        \typerule{Alloc}{
            p(\rho(\pc)) = \palloc{x}{n} 
            & \rho' = \rho[x\mapsto\rho(\mem), \mem \mapsto \rho(\mem)+n]
            & \mu' = \mu[x\mapsto \vec{\tlow}]
        }{
            (p, \conf{\rho,\mu, f}) \eval{\DirStep}{\Obsnone} (p, \conf{\rho'[\pc \mapsto \rho(\pc)+1], \mu', f})
        }{app:tsemantic:alloc}
        \caption{Speculative Semantics of \muasm.}
        \label{fig:tsemantics}
    \end{center}
\end{figure*}    
\section{Reduce Speculative Non-Interference to Speculative Safety}\label{app:sec:sni-ss}
In this section we prove \Cref{the:sni-ss}.

The following lemma establishes a connection between taint tracking and non-interference.
It captures the property that the value which corresponds to a non-\thigh label is determined only by its preceding trace, and the initial values of registers and memory addresses in $P$. 

\begin{lemma}\label{lemma:eval-value}
    For an expression $e$, let $X$ be the set of all registers that appear in $e$. 
    For any pairs of value mapping $\rho$ and $\rho'$ such that $\rho(x)=\rho'(x)$ holds for every $x\in X$, then we have $\exprEval{e}{\rho}=\exprEval{e}{\rho'}$.
\end{lemma}
The proof is straightforward by using mathematical induction to the length of $e$. 
Similarly, we have
\begin{lemma}\label{lemma:eval-taint}
    For an expression $e$, let $X$ be the set of all registers that appear in $e$.
    For any pairs of taint mapping $\mu$ and $\mu'$ such that $\mu(x)=\mu'(x)$ holds for every $x\in X$, then we have $\exprEval{e}{\mu}=\exprEval{e}{\mu'}$.
\end{lemma}

Furthermore, \Cref{thm:well-defined} and \Cref{def:well-defined} imply that 
\begin{lemma}\label{lemma:eval-non-inter}
    For an expression $e$, let $X$ be the set of all registers that appear in $e$.
    For any pairs of value mapping $\rho$ and $\rho'$, and a taint mapping $\mu$ such that \isinstance{\rho(x)}{\mu(x)}, \isinstance{\rho'(x)}{\mu(x)}, and $\tainteq{\rho(x)}{\rho'(x)}{\mu(x)}$ hold, for every $x\in X$, then we have $\tainteq{\exprEval{e}{\rho}}{\exprEval{e}{\rho'}}{\exprEval{e}{\mu}} $.
\end{lemma}

It follows from \Cref{lemma:eval-value}, \Cref{lemma:eval-taint} and \Cref{lemma:eval-non-inter} that
\begin{lemma}\label{lemma:label-non}
    For any pair of execution traces 
    \begin{gather*}
        \trace{p}{s_1}{D}{O}=(p, s_1)\xrightarrow[d_1]{o_1} (p, s_2)\xrightarrow[d_2]{o_2} (p, s_3)\cdots \xrightarrow[d_{n-1}]{o_{n-1}}(p, s_n)\\
        \trace{p}{s'_1}{D'}{O'}=(p, s'_1)\xrightarrow[d'_1]{o'_1} (p, s'_2)\xrightarrow[d'_2]{o'_2} (p, s'_3)\cdots \xrightarrow[d'_{n-1}]{o'_{n-1}}(p, s'_n)
    \end{gather*}
    such that $\policyeq{s_1}{s'_1}{P}$, $s_1 \vdash P$ and $s'_1 \vdash P$, 
    then for any $r\in \Reg \cup \Nat$, $1\le t\le n$, we have
        \begin{align*}
            (\forall 1\le i\le t,\ &\rho_{s_i}(\pc)=\rho_{s_i'}(\pc)) \Rightarrow \\
             &((\mu_{s_t}(r) = \mu_{s_t'}(r)) 
             \land  (\tainteq{\rho_{s_t}(r)}{\rho_{s'_t}(r)}{\mu_{s_t}(r)})).
        \end{align*}
    
\end{lemma}
\begin{proof}
    We will use mathematical induction to complete the proof. The conclusion holds when $t = 1$. Suppose the conclusion holds when $t = k < n$. Then, when $t = k+1$, we consider the different cases of $\rho_{s_k}(\pc)$ ($=\rho_{s'_k}(\pc)$, by prerequisite).
    \begin{enumerate}
        \renewcommand{\labelenumi}{\theenumi.}
        \item $\rho_{s_k}(\pc)=\passign{x}{e}$. 
        As \Cref{tr:app:tsemantic:assign} states, $\mu_{s_k}$ and $\mu_{s_{k+1}}$ only disagree on $x$ and $\pc$, and $\mu_{s'_k}$  similarly with $\mu_{s'_{k+1}}$.
        By the inductive hypothesis,  for $r\neq x$, $\mu_{s_{k+1}}(r)=\mu_{s_{k}}(r)=\mu_{s'_{k}}(r) = \mu_{s'_{k+1}}(r)$.
        By \Cref{lemma:eval-taint}, we have $\mu_{s_{k+1}}(x) = \exprEval{e}{\mu_{s_k}}=\exprEval{e}{\mu_{s'_k}} = \mu_{s_{k+1}'}(x)$.

        The inductive hypothesis implies that $\tainteq{\rho_{s_k}(r)}{\rho_{s'_k}(r)}{\mu_{s_k}(r)}$. 
        Then by \Cref{lemma:eval-non-inter}, $\tainteq{\rho_{s_{k+1}}(x)}{\rho_{s'_{k+1}}(x)}{\mu_{s_{k+1}}(x)}$.
        And for $r\neq x$, $\rho_{s_{k+1}}(r)= \rho_{s_k}(r) \backsim_{\mu_{s_k}(r)} \rho_{s'_k}(r)= \rho_{s'_{k+1}}(r)$.

        \item $\rho_{s_k}(\pc)= \pload{x}{e}$.     
        As \Cref{tr:app:tsemantic:load} states, apart from $\pc$, only $x$ will be modified in $\rho_{s_k}$ and $\mu_{s_k}$. 

        Let $t = \exprEval{e}{\mu_{s_k}} (= \exprEval{e}{\mu_{s'_k}})$.
        If $\thigh \in t$, then $\mu_{s_{k+1}}(x)=\vec{\thigh}=\mu_{s'_{k+1}}(x)$ and naturally $\tainteq{\rho_{s_{k+1}}(x)}{\rho_{s'_{k+1}}(x)}{\mu_{s_{k+1}}(x)}$ (since there is no non-\thigh label in $\mu_{s_{k+1}}(x)$) .
        If $\thigh \notin t$, $\tainteq{\exprEval{e}{\rho_{s_k}}}{\exprEval{e}{\rho_{s'_k}}}{t}$ (by the inductive hypothesis) implies that $\exprEval{e}{\rho_{s_k}} = \exprEval{e}{\rho_{s'_k}}\triangleq l$. 
        Then, $\muskone(x)=\musk(l)=\muspk(l)=\muspkone(x)$ and $\rhoskone(x)=\rhosk(l)\backsim_{\musk(l)(=\muskone(x))}\rhospk(l)=\rhospkone(x)$.

        \item  $\rho_{s_k}(\pc)= \pstore{x}{e}$. 
        Let $t = \exprEval{e}{\mu_{s_k}} (= \exprEval{e}{\mu_{s'_k}})$. 
        Similar to $\loadKywd$, we only need to consider the case where $\thigh \notin t$.
        
        If $\thigh \notin t$, $\tainteq{\exprEval{e}{\rhosk}}{\exprEval{e}{\rhospk}}{t}$ (by the inductive hypothesis) implies that $\exprEval{e}{\rho_{s_k}} = \exprEval{e}{\rho_{s'_k}}\triangleq l$. 
        Then, $\muskone(l)=\musk(x)=\muspk(x)=\muspkone(l)$ and $\rhoskone(l)=\rhosk(x)\backsim_{\musk(x)(=\muskone(l))}\rhospk(x)=\rhospkone(l)$.

        \item  $\rho_{s_k}(\pc)= \pfence$ or $\pbranch{x}{l}$. 
        Conclusion holds since no register other than $\pc$ or memory address is modified.
        \item $\rho_{s_k}(\pc)= \pcondassign{x}{e}{e'}$. 
        As \Cref{tr:app:tsemantic:conditional-assign} states, apart from $\pc$, only $x$ will be modified in $\rho_{s_k}$ and $\mu_{s_k}$. 
        Let $t = \exprEval{e'}{\musk} (= \exprEval{e'}{\muspk})$. 
        Similar to $\loadKywd$, we only need to consider the case where $\thigh \notin t$.
        
        If $\thigh \notin t$, $\tainteq{\exprEval{e'}{\rhosk}}{\exprEval{e'}{\rhospk}}{t}$ (by the inductive hypothesis) implies that $\exprEval{e'}{\rho_{s_k}} = \exprEval{e'}{\rho_{s'_k}}\triangleq l$. 
        If $l=0$, we have $\muskone(x)=\musk(x)=\muspk(x)=\muspkone(x)$ and $\rhoskone(x)=\rhosk(x)\backsim_{\musk(x)(=\muskone(x))}\rhospk(x)=\rhospkone(x)$.
        If $l\neq 0$,  by \Cref{lemma:eval-taint}, we have $\muskone(x)=\exprEval{e}{\musk}=\exprEval{e}{\muspk}=\muspkone(x)$.
        By \Cref{lemma:eval-non-inter}, we have $\rhoskone(x)=\exprEval{e}{\rhosk}\backsim_{\exprEval{e}{\musk}(=\muskone(x))}\exprEval{e}{\rhospk}=\rhospkone(x)$.

        \item $\rho_{s_k}(\pc)= \palloc{x}{n}$.
        It is clear that $\muskone(x)=\vec{\tlow}=\muspkone(x)$.
        We also have $\rhoskone(x)=\rhosk(\mem)=\rhospk(\mem)=\rhospkone(x)$, where the second equation stems from the assumption that the $\mem$ register is always marked as $\vec{\tlow}$.
        Furthermore,  $\rhoskone(\mem)=\rhosk(\mem)+n=\rhospk(\mem)+n=\rhospkone(\mem)$.
    \end{enumerate}

    So conclusion holds when $t=k+1$.
    Using mathematical induction, we can prove that the conclusion holds for $1\le t\le n$.

\end{proof}

\SStoSNI*
\begin{proof}
    Suppose there is a program $p$ satisfying SS with a violation of SNI.
    Then we can find a pair of traces, $\tau=\trace{p}{s_1}{D}{O}$ and $\tau'=\trace{p}{s'_1}{D}{O'}$, with the corresponding sequential traces \seqtrace{p}{s_1}{\overline{O}} and \seqtrace{p}{s'_1}{\overline{O}'}, such that $\policyeq{s_1}{s'_1}{P}$, $s_1 \vdash P$ , $s'_1 \vdash P$ and $\overline{O} = \overline{O}'$, but $O\neq O'$.
    Unwind $O$ as $o_1o_2\cdots o_n$, $O'$ as $o_1'o_2'\cdots o_{n'}'$.

    Let $i$ denote the least index such that $o_i\neq o_i'$. 
    It is obvious that $o_i$ and $o_i'$ must be generated during misspeculative execution.
    The equality of $\overline{O}$ and $\overline{O}'$, and the shared directives $D$ imply that $\tau$ and $\tau'$ have the same start point of misspeculative execution, i.e., speculative flags of $s_j$ and $s_j'$ have a same evaluation given any $s_j$ in $\tau$ and $s_j'$ in $\tau'$.
    
    Consider the value of $\pc$ in $s_j$ and $s_j'$, where $j\le i$. 
    For $j$ such that $s_j$ and $s_j'$ are in sequential part of the trace, $\overline{O} = \overline{O}'$ derives the equality of the value of $\pc$ in $s_j$ and $s_j'$.
    For $j$ such that $s_j$ and $s_j'$ are in misspeculative part of the trace, by applying mathematical induction to $j$, we can also  derive the equality of $\pc$'s value  in $s_j$ and $s_j'$ from \Cref{def:well-defined}, \Cref{lemma:label-non} and the definition of speculative safety which requires the absence of \thigh labels during speculative execution.
    
    Now we have (1) the preceding traces of $o_i$ and $o_i'$ have the same value of $\pc$ at each state, (2) $o_i$ and $o_i'$ is free of \thigh label.
    By \Cref{lemma:label-non}, we have $o_i=o_i'$, which leads to a contradiction.
\end{proof}

As can be seen from the proof, even though we do not track taint caused by implicit information flow, speculative safety still successfully guarantees the soundness of speculative non-interference.
This is due to the fact that, firstly, the premise of speculative non-interference requires that any two traces produce the same observations in sequential execution.
The premise inherently excludes the possibility of an H label leaking information during such execution. 
Secondly, the property of speculative safety prohibits the H label from appearing in observations generated during misspeculative execution, further eliminating the potential for implicit information flow of \thigh labels.

\section{Value Domain}\label{app:sec:vd}
Without loss of generality, we assume that the length of each register and the length of each memory address are both \consn, a given constant.

Let the concrete domain be $\condm$, where \integer denotes all \consn-bit integers. 
The abstract function from $\condm$ to abstract domain $L$ is denoted by $\abf{L}$, and the concretization function is denoted by $\conff{L}{}$.

We first present the operator rules on interval domain of \consn-bit integers, written $\interval$. 
Abstract function and concretization function between $\powerset{\integer}$ and $\interval$ are straightforward.
We denote the maximum and minimum value in $\interval$ by $\interval_{\max}(=2^{n-1}-1)$ and $\interval_{\min}(=-2^{n-1})$, respectively.
Let $I_1=[a_1,b_1], I_2=[a_2,b_2]\in\interval$.

\begin{figure*}
\begin{itemize}
    \item $\uopondomain{\notop}{\interval}{I_1}=[-1-b_1,-1-a_1]$.
    \item $\bopondomain{\addop}{\interval}{I_1}{I_2}=
    \begin{cases}
        [\minI,\maxI] & \text{if $a_1+a_2<\minI$ or $b_1+b_2>\maxI$}\\
        [a_1+a_2,b_1+b_2] & \text{otherwise}
    \end{cases}$.
    \item $\bopondomain{\minusop}{\interval}{I_1}{I_2}=
    \begin{cases}
        [\minI,\maxI] & \text{if $a_1-b_2<\minI$ or $b_1-a_2>\maxI$}\\
        [a_1-b_2,b_1-a_2] & \text{otherwise}
    \end{cases}$.
    \item $\bopondomain{\mulop}{\interval}{I_1}{I_2}=
    \begin{cases}
        [\minI,\maxI] & \text{if $\exists c\in C$ s.t. $c>\minI$ or $c<\minI$}\\
        [\min(C),\max(C)] & \text{otherwise}
    \end{cases}$, where $C=\{a_1\cdot a_2, a_1\cdot b_2, b_1\cdot a_2, b_1\cdot b_2\}$.
    \item $\bopondomain{\divop}{\interval}{I_1}{I_2}=[\minI,\maxI]$.
    \item $\bopondomain{\modop}{\interval}{I_1}{I_2}=I_1$. 
    \item $\bopondomain{\rshaop}{\interval}{I_1}{I_2}=
    \begin{cases}
        [\minI,\maxI] & \text{if $b_1< 0$}\\
        [\min(C),\max(C)] & \text{otherwise}
    \end{cases}
    $, where $C=\{\bopondomain{\rshaop}{\integer}{s}{t}\mid s\in\{a_1,b_1\},t\in\{a_2,b_2\} \}$.
    \item $\bopondomain{\andop}{\interval}{I_1}{I_2}=
    \begin{cases}
        [\minI,\maxI] & \text{if $a_2< 0$ and $b_2\ge0$, or if $a_1<0$}\\
        [0,\min(b_1,b_2)] & \text{if $a_2\ge 0$}\\
        [a_1-\uopondomain{\notop}{\integer}{a_2},b_1 - \uopondomain{\notop}{\integer}{b_2}] & \text{if $b_2<0$}
    \end{cases}$.
\end{itemize}
For the following bitwise operators, we assume that $a_i,b_i>0(i=1,2)$, otherwise, the result is set to $[\minI,\maxI]$.
\begin{itemize}
    \item $\bopondomain{\orop}{\interval}{I_1}{I_2}=[\max(a_1,a_2),\maxI]$.
    \item $\bopondomain{\xorop}{\interval}{I_1}{I_2}=[\minI,\maxI]$.
    \item $\bopondomain{\lshiftop}{\interval}{I_1}{I_2}=
    \begin{cases}
        [\minI,\maxI] & \text{if $(b_1<<b_2) >\maxI$}\\
        [a_1<<a_2,b_1<<b_2] & \text{otherwise}
    \end{cases}$.
    \item $\bopondomain{\rshlop}{\interval}{I_1}{I_2}=[\bopondomain{\rshaop}{\integer}{a_1}{b_2},\bopondomain{\rshaop}{\integer}{a_2}{b_1}]$.
\end{itemize} 
\caption{Operator rules on interval domain.}
\label{fig:itv}
\end{figure*}

The operator rules on interval domain are shown in \Cref{fig:itv}.


Moving on, let us examine the disjoint interval set domain ($\di$).
\begin{restatable}[Disjoint Interval Set]{definition}{dis}\label{def:di}
    A disjoint interval set is a set of intervals $\{[a_i,b_i]\mid 1\le i\le n, a_i, b_i\in\Int \}$ where $n\in\Nat$, and $\Int$ denotes the set of integers, such that $a_i\le b_i$ and $b_i < a_{i+1} - 1$ holds for all $1 \le i \le n$..
\end{restatable}

For $n\in \integer$ and $D\in \di$, we denote by \indi{n}{D} if $\exists [a,b]\in D$ such that $n\in [a,b]$.
We define a relation \diorder on $\di$: for $d,d'\in \di$, $d \diorder d'$ holds iff for any $[a,b]\in d$, there is $[a',b']\in d'$ such that $a'\le a\le b\le b'$.

It is clear that any set of integers can be uniquely represented by a corresponding disjoint interval set.
Therefore, there is an isomorphism between $\condm$ and $\di$, and we have a natural $\abf{\di}$ and $\conff{\di}{}$.
Furthermore, $\diorder$ and $\subseteq$ are isomorphic functions.
We denote by $\di(Z)$ the corresponding disjoint interval set of $Z\in\condm$.

In the domain of disjoint interval sets, the greatest upper bound and least lower bound can be naturally derived from the corresponding concepts in the interval domain.
Furthermore, the operator rules for $\di$ can be largely inherited from the interval domain, with the exception of multiplication, which requires special handling.

Let $D_1=\{[a_{1,1},b_{1,2}],\cdots,[a_{1,s},b_{1,s}]\}$,$D_2=\{[a_{2,1},b_{2,2}],\cdots,[a_{2,t},b_{2,t}]\}$ be two disjoint interval sets.
For \bopondomain{\mulop}{\di}{D_1}{D_2}, if $t=1$ and $a_{2,1}=b_{2,1}\triangleq l$ (which indicates that $D_2$ contains an only integer $l$), we define 
$$
    \bopondomain{\mulop}{\di}{D_1}{D_2} = \di(\{n\mid n\in \bopondomain{\mulop}{\interval}{[a_{1,i},b_{1,i}]}{[l,l]}
    \})
$$
holds for some $1\le i\le s$

Now we can define the abstract value domain (denote by \VD) of a program $p$.
For a program $p$, we collect all the $\allocKywd$ instructions into the set $\base_p=\{i\mid p(i) = \palloc{x}{n}\}$. 
With each $b\in\base_p\cup \{\emptysym\}$ and the corresponding offset denotes a possible range of value, the abstract interpretation of the value is represented by a mapping $\vd: \base_p\cup\{\emptysym\}\to\di$ with $\vd(\emptysym)$ being a singleton set (i.e., $\vd(\emptysym)$ contains only one interval).

We can derive a lattice on \VD from \parlattice{\di}{\diorder} by requiring that for any $\vd_1,\vd_2\in\VD$, $\vd_1\diorder\vd_2$ holds iff for any $ b\in \base_p\cup\{\emptysym\}$, we have $\vd_1(b)\diorder \vd_2(b)$.
Furthermore, we can obtain the greatest lower bound and the least upper bound in the lattice \VD.
An abstract value $\vd$ is called \textit{abstract number} iff for any $b\in\base_p$ , $\vd(b)=\emptyset$.


Let $\vd_1,\vd_2$ be two abstract values, and $\vd=\bopondomain{\op}{\VD}{\vd_1}{\vd_2}$ be the computation result in \VD.

\begin{itemize}
    \item For $\op=\addop$, 
    \begin{itemize}
        \item If $\vd_1$ and $\vd_2$ are both abstract numbers, \vd is given by
            $$
                \vd(x) = \begin{cases}
                    \bopondomain{\addop}{\di}{\vd_1(\emptysym)}{\vd_2(\emptysym)} & \text{if $x = \emptysym$}\\
                    \emptyset & \text{otherwise}
                \end{cases}
            $$
        \item 
        If only one of $\vd_1$ and $\vd_2$ is an abstract number, w.l.o.g., we let $\vd_2$ be the abstract number, then
        $$
            \vd(x) = \bopondomain{\addop}{\di}{\vd_1(x)}{\vd_2(\emptysym)}
        $$
        \item If $\vd_1$ and $\vd_2$ are neither abstract numbers,  \vd is set to \topof{\VD}.
    \end{itemize}
    \item For $\op=\minusop$, 
    \begin{itemize}
        \item If $\vd_1$ and $\vd_2$ are both abstract numbers, \vd is given by
            $$
                \vd(x) = \begin{cases}
                    \bopondomain{\minusop}{\di}{\vd_1(\emptysym)}{\vd_2(\emptysym)} & \text{if $x = \emptysym$}\\
                    \emptyset & \text{otherwise}
                \end{cases}
            $$
        \item 
        If  $\vd_2$ is an abstract number and $\vd_1$ is not. Then,
        $$
            \vd(x) = \bopondomain{\minusop}{\di}{\vd_1(x)}{\vd_2(\emptysym)}
        $$
        \item Otherwise,  \vd is set to \topof{\VD}.
    \end{itemize}
    \item For $\op=\andop$, 
    \begin{itemize}
        \item If $\vd_1$ and $\vd_2$ are both abstract numbers, \vd is given by
            $$
                \vd(x) = \begin{cases}
                    \bopondomain{\andop}{\di}{\vd_1(\emptysym)}{\vd_2(\emptysym)} & \text{if $x = \emptysym$}\\
                    \emptyset & \text{otherwise}
                \end{cases}
            $$
        \item If only one of $\vd_1$ and $\vd_2$ is an abstract number, w.l.o.g., we let $\vd_2$ be the abstract number,  
        \begin{itemize}
            \item If there exists a non-negative $z_1$ and a negative $z_2$, such that $\indi{z_1}{\vd_2(\emptysym)}$ and $\indi{z_2}{\vd_2(\emptysym)}$, \vd is given by
            $$
                \vd(x) = \begin{cases}
                    \topof{\di} & \text{if $x = \emptysym$}\\
                    \emptyset & \text{otherwise}
                \end{cases}
            $$
            \item If there does not exist a negative $z$, such that $\indi{z_1}{\vd_2(\emptysym)}$, then \vd is assigned with $\vd_2$.
            \item If there does not exist a non-negative $z$, such that $\indi{z_1}{\vd_2(\emptysym)}$, then \vd is given by
            $$
                \vd(x) = 
                    \bopondomain{\andop}{\di}{\vd_1(x)}{\vd_2(\emptysym)} 
            $$
        \end{itemize}
        \item Otherwise,  \vd is set to \topof{\VD}.
    \end{itemize}
    \item For $\op\in \bop\cup\uop/\{\addop,\minusop,\andop\}$,
    \begin{itemize}
        \item If either $v_1$ or $v_2$ is not abstract number, then the result $\vd$ is the top element $\topof{\VD}$ of the domain $\VD$.
        \item If $v_1$ and $v_2$ are both abstract numbers, we have 
        $$
        \vd(x)=\begin{cases}
            \bopondomain{\op}{\di}{\vd_1(\emptysym)}{\vd_2(\emptysym)} & \text{if $x = \emptysym$}\\
                    \emptyset & \text{otherwise}
        \end{cases}
        $$   
    \end{itemize}
\end{itemize}


\section{Abstract Semantics}\label{app:sec:ass}

In this section, we present the abstract semantics of our analysis.
First, we introduce abstract sequential semantics, and then abstract speculative semantics.
Both semantics work with a taint domain \atd, and a value domain \VD.
Before delving into the main discussion, let us establish some notations.

\textit{Abstract state}: Our abstract semantics work on abstract state $\aconf$.
$\aconf$ is defined as a domain of tuples $\conf{\abrho,\abmu, \vmem, \tmem}$ (written \abs).
Here, $\abrho$ maps $\Reg/\{\pc,\mem\}$ to \VD, and $\{\pc, \mem\}$ to $\Nat$.
And $\abmu$ is a mapping from $\Reg$ to \atd.
Memory is modeled by \vmem and \tmem. 
The former is a memory model that stores abstract values, while the latter is a memory model that stores taint vectors. 
The load and store operations on \tmem and \vmem are modeled by \abload{\vmem}, \abstore{\vmem}, \abload{\tmem} and  \abstore{\tmem}.
Each $\abd{s} = \conf{\abrho,\abmu, \vmem, \tmem}$ represents a possible state of program.
We also denote the $\abrho$ in state $\abd{s}$ by $\abrho_{\abs}$, and similarly we have the notations 
$\abmu_{\abs}$, $\vmem_{\abs}$ and $\tmem_{\abs}$.

\textit{Transition and program trace}: The one-step abstract execution of a program $p$ is modeled using a transition relation between two abstract states. 
We denote the transition within abstract sequential semantics (or abstract speculative semantics) by \seqtran{(p,\abs_1)}{(p,\abs_2)}{\abo} (\spectran{(p,\abs_1)}{(p,\abs_2)}{\abo},  resp.), for which we say program $p$ with an abstract state $\abs_1$ generates an abstract observation \abo.
Here, $\abo\in\abd{\Obs}$, and \abd{\Obs} is given by:
$$
    \abd{\Obs} \coloneqq \Obsnone \mid \Obsbranch{\vd:\abt} \mid \Obsload{\vd:\abt}{a}{b} \mid \Obsstore{\vd:\abt}{a}{b}
$$
In the formula above, $\vd\in\VD$, $\abt\in\atd$, and $[a,b]$ denotes that the attacker can observe bits from a to b.
Note that we do not model the speculative flag and directives in the transition of  abstract states. 
This is because during abstract interpretation, all possible branch addresses are computed and taken into account in the next state.

\begin{figure*}[!h]
    \begin{center}
        \small
        \typerule{Asgn-Seq}{
            p(\abrho(\pc)) = \passign{x}{e} 
            \\ \abmu_1 = \updatemap{\abmu}{x}{\exprEval{e}{\abmu}}
            \\ \abrho_1 = \abrho[x\mapsto\exprEval{e}{\abrho},\pc\mapsto\abrho(\pc)+1]
        }{
            \seqtran{(p,\abconf{})}{(p,\conf{\abrho_1,\abmu_1,\vmem,\tmem})}{\epsilon}
        }{app:tsemantic:seq:assign}
         \typerule{Ld-Seq}{
             p(\abrho(\pc)) = \pload{x}{e} & \vd = \exprEval{e}{\abrho} & \abt = \exprEval{e}{\abmu} \\ 
             \abt_1 = {\begin{cases}
            \abvechigh & \text{if $\abhigh \in \abt$}\\
            \abload{\tmem}(\vd) &\text{if $\abhigh \notin \abt$}
            \end{cases}} & {\begin{aligned}
                \abrho_1=\abrho[x\mapsto \abload{\vmem}(\vd)]\\
                \abrho_2=\abrho_1[\pc\mapsto\abrho(\pc)+1] \\
                \abmu_1=\updatemap{\abmu}{x}{\abt_1}
            \end{aligned}}
         }{
            \seqtran{(p,\abconf{})}{(p,\conf{\abrho_2,\abmu_1,\vmem,\tmem})}{\Obsload{\vd:\abt}{a}{b}}
         }{app:tsemantic:seq:load}
         \typerule{St-Seq}{
             p(\abrho(\pc)) = \pstore{x}{e} & \vd = \exprEval{e}{\abrho} & \abt = \exprEval{e}{\abmu} \\ 
            \abt_1 = {\begin{cases}
            \abvechigh & \text{if $\abhigh \in \abt$}\\
            \abmu(x)      &\text{if $\abhigh \notin \abt$}
            \end{cases}} & {\begin{aligned}
                \vmem_1 = \abstore{\vmem}(\vd, \abrho(x))\\
                \tmem_1 = \abstore{\tmem}(\vd, \abt_1)\\
                \abrho_1=\rho[\pc\mapsto\abrho(\pc)+1]
            \end{aligned}}
         }{
            \seqtran{(p,\abconf{})}{(p,\conf{\abrho_1,\abmu,\vmem_1,\tmem_1})}{\Obsstore{\vd:\abt}{a}{b}}
         }{app:tsemantic:seq:store}
         \typerule{CondAsgn-Seq}{
             p(\abrho(\pc)) = \pcondassign{x}{e}{e'} 
             & \abt = \exprEval{e}{\abmu} 
             & \abt_1 = \exprEval{e'}{\abmu} \\
             \vd = \exprEval{e}{\abrho} 
             & \abmu_1={
                \begin{cases}
                    \updatemap{\abmu}{x}{\abvechigh} & \text{if $\abhigh \in \abt_1$}\\
                    \updatemap{\abmu}{x}{ \abt\sqcup \abmu(x)} & \text{otherwise }
                \end{cases}
             }\\
             \abrho_1 = \abrho[x\mapsto \abrho(x)\sqcup \vd, \pc\mapsto\abrho(\pc)+1]
         }{
            \seqtran{(p,\abconf{})}{(p,\conf{\abrho_1,\abmu_1,\vmem,\tmem})}{\epsilon}
         }{app:tsemantic:seq:conditional-assign}
        \typerule{Fen-Seq}{
             p(\abrho(\pc)) = \pfence \\ 
             \abrho_1=\abrho[\pc\mapsto\abrho(\pc)+1]
         }{
            \seqtran{(p,\abconf{})}{(p,\conf{\abrho_1,\abmu,\vmem,\tmem})}{\epsilon}
         }{app:tsemantic:seq:fence}
         \typerule{Jmp-Seq}{  
            p(\abrho(\pc)) = \pjmp{l} 
            & \abrho_1 = \updatemap{\abrho}{\pc}{l}
        }{
            \seqtran{(p,\abconf{})}{(p,\conf{\abrho_1,\abmu,\vmem,\tmem})}{\epsilon}
        }{app:tsemantic:seq:jump}
         \typerule{Br-T-Seq}{
            p(\abrho(\pc)) = \pbranch{x}{l} 
            & \vd = \abrho(x) 
            & \abt = \abmu(x)
            \\ \abrho_1 = \abrho[x\mapsto\vd\sqcap{\{[0,0]\}},\pc\mapsto l]
        }{
            \seqtran{(p,\abconf{})}{(p,\conf{\abrho_1,\abmu,\vmem,\tmem})}{\Obsbranch{\vd:\abt}}
        }{app:tsemantic:seq:branch-true}
        \typerule{Br-F-Seq}{
            p(\abrho(\pc)) = \pbranch{x}{l} 
            & \vd = \abrho(x) 
            & \abt = \abmu(x)
            \\ \abrho_1 = \abrho[x\mapsto \vd\sqcap\{[-\infty,-1],[1,\infty]\},\pc\mapsto\abrho(\pc)+1 ]
        }{
            \seqtran{(p,\abconf{})}{(p,\conf{\abrho_1,\abmu,\vmem,\tmem})}{\Obsbranch{\vd:\abt}}
        }{app:tsemantic:seq:branch-false}
        \typerule{Alloc-Seq}{
            p(\abrho(\pc)) = \palloc{x}{n} 
            & \abrho_1 = \abrho[x\mapsto \updatemap{\bot_\VD}{\abrho(\pc)}{\{[0]\}}, \mem \mapsto \abrho(\mem)+n]
            & \abmu_1 = \abmu[x\mapsto \abveclow]
        }{
            \seqtran{(p,\abconf{})}{(p,\conf{\abrho_1,\abmu_1,\vmem,\tmem})}{\epsilon}
        }{app:tsemantic:seq:alloc}
        \caption{Abstract Sequential Semantics of \abconf{}.}
        \label{fig:tsemantics:seq:ab}
    \end{center}
\end{figure*}    

The abstract sequential semantics (\Cref{fig:tsemantics:seq:ab}) share many similarities with the concrete semantics. 
We now focus on the points of divergence between the two.

The first distinction lies in the handling of conditional assignment instructions. 
Unlike the concrete semantic that assigns values based on the value of $e'$, the abstract semantic utilizes the least upper bound of the original value and the value of $e$ to estimate all possible assignment scenarios.

The second difference lies in the handling of branch instructions.
In abstract sequential semantics, there are two rules for transitions when encountering a $\jzKywd$ instruction: \Cref{tr:app:tsemantic:seq:branch-true} and \Cref{tr:app:tsemantic:seq:branch-false}.
These two rules allow both jump and non-jump scenarios as subsequent states of an abstract state, and they constrain the values in these subsequent states using branch conditions. 
For instance, in \Cref{tr:app:tsemantic:seq:branch-true}, $\pc$ is set to the corresponding value when the branch is taken, indicating that the branch condition holds, thus ${\{[0,0]\}}$ is used to constrain the branch variable $x$. 
On the other hand, \Cref{tr:app:tsemantic:seq:branch-false}, corresponding to the case where the branch condition does not hold, uses $\{[-\infty,-1],[1,\infty]\}$ to constrain the branch variable $x$. 

Another difference lies in the handling of allocation instructions.
In abstract semantics, we do not consider the concrete value of the allocated base address. 
Instead, we assign $x$  an abstract value $\updatemap{\bot_\VD}{\abrho(\pc)}{\{[0]\}}$, where $\abrho(\pc)\in \base_p$ is the corresponding symbol for current allocation instruction.

Similarly we can define the abstract speculative semantics of \abconf{} (\Cref{fig:tsemantics:spec:ab}).
\begin{figure*}[!h]
    \begin{center}
        \small
        \typerule{Asgn-Spec}{
            p(\abrho(\pc)) = \passign{x}{e} 
            \\ \abmu_1 = \updatemap{\abmu}{x}{\exprEval{e}{\abmu}}
            \\ \abrho_1 = \abrho[x\mapsto\exprEval{e}{\abrho},\pc\mapsto\abrho(\pc)+1]
        }{
            \spectran{(p,\abconf{})}{(p,\conf{\abrho_1,\abmu_1,\vmem,\tmem})}{\epsilon}
        }{app:tsemantic:spec:assign}
         \typerule{Ld-Spec}{
             p(\abrho(\pc)) = \pload{x}{e} & \vd = \exprEval{e}{\abrho} & \abt = \exprEval{e}{\abmu} \\ 
             \abt_1 = {\begin{cases}
            \abvechigh & \text{if $\abhigh \in \abt$}\\
            \abload{\tmem}(\vd) &\text{if $\abhigh \notin \abt$}
            \end{cases}} & {\begin{aligned}
                \abrho_1=\abrho[x\mapsto \abload{\vmem}(\vd)]\\
                \abrho_2=\abrho_1[\pc\mapsto\abrho(\pc)+1] \\
                \abmu_1=\updatemap{\abmu}{x}{\abt_1}
            \end{aligned}}
         }{
            \spectran{(p,\abconf{})}{(p,\conf{\abrho_2,\abmu_1,\vmem,\tmem})}{\Obsload{\vd:\abt}{a}{b}}
         }{app:tsemantic:spec:load}
         \typerule{St-Spec}{
             p(\abrho(\pc)) = \pstore{x}{e} & \vd = \exprEval{e}{\abrho} & \abt = \exprEval{e}{\abmu} \\ 
            \abt_1 = {\begin{cases}
            \abvechigh & \text{if $\abhigh \in \abt$}\\
            \abmu(x)      &\text{if $\abhigh \notin \abt$}
            \end{cases}} & {\begin{aligned}
                \vmem_1 = \abstore{\vmem}(\vd, \abrho(x))\\
                \tmem_1 = \abstore{\tmem}(\vd, \abt_1)\\
                \abrho_1=\rho[\pc\mapsto\abrho(\pc)+1]
            \end{aligned}}
         }{
            \spectran{(p,\abconf{})}{(p,\conf{\abrho_1,\abmu,\vmem_1,\tmem_1})}{\Obsstore{\vd:\abt}{a}{b}}
         }{app:tsemantic:spec:store}
         \typerule{CondAsgn-Spec}{
             p(\abrho(\pc)) = \pcondassign{x}{e}{e'} 
             & \abt = \exprEval{e}{\abmu} 
             & \abt_1 = \exprEval{e'}{\abmu} \\
             \vd = \exprEval{e}{\abrho} 
             & \abmu_1={
                \begin{cases}
                    \updatemap{\abmu}{x}{\abvechigh} & \text{if $\abhigh \in \abt_1$}\\
                    \updatemap{\abmu}{x}{ \abt\sqcup \abmu(x)} & \text{otherwise }
                \end{cases}
             }\\
             \abrho_1 = \abrho[x\mapsto \abrho(x)\sqcup \vd, \pc\mapsto\abrho(\pc)+1]
         }{
            \spectran{(p,\abconf{})}{(p,\conf{\abrho_1,\abmu_1,\vmem,\tmem})}{\epsilon}
         }{app:tsemantic:spec:conditional-assign}
        \typerule{Fen-Spec}{
             p(\abrho(\pc)) = \pfence 
             \abrho_1=\abrho[\pc\mapsto\abrho(\pc)+1]
         }{
            \spectran{(p,\abconf{})}{(p,\conf{\abrho_1,\abmu,\vmem,\tmem})}{\epsilon}
         }{app:tsemantic:spec:fence}
         \typerule{Fen-Spec-Block}{
             p(\abrho(\pc)) = \pfence &
             \abrho_1=\abrho[\pc\mapsto\bot]
         }{
            \spectran{(p,\abconf{})}{(p,\conf{\abrho_1,\abmu,\vmem,\tmem})}{\epsilon}
         }{app:tsemantic:spec:fence}
         \typerule{Br-T-Spec}{
            p(\abrho(\pc)) = \pbranch{x}{l} 
            & \vd = \abrho(x) 
            & \abt = \abmu(x)
            \\ \abrho_1 = \abrho[\pc\mapsto l]
        }{
            \spectran{(p,\abconf{})}{(p,\conf{\abrho_1,\abmu,\vmem,\tmem})}{\Obsbranch{\vd:\abt}}
        }{app:tsemantic:spec:branch-true}
        \typerule{Br-F-Spec}{
            p(\abrho(\pc)) = \pbranch{x}{l} 
            & \vd = \abrho(x) 
            & \abt = \abmu(x)
            \\ \abrho_1 = \abrho[\pc\mapsto\abrho(\pc)+1 ]
        }{
            \spectran{(p,\abconf{})}{(p,\conf{\abrho_1,\abmu,\vmem,\tmem})}{\Obsbranch{\vd:\abt}}
        }{app:tsemantic:spec:branch-false}
        \typerule{Alloc-Spec}{
            p(\abrho(\pc)) = \palloc{x}{n} 
            & \abmu_1 = \abmu[x\mapsto \abveclow]
            \\ \abrho_1 = \abrho[x\mapsto \updatemap{\bot_\VD}{\abrho(\pc)}{\{[0]\}}, \mem \mapsto \abrho(\mem)+n]  
        }{
            \spectran{(p,\abconf{})}{(p,\conf{\abrho_1,\abmu_1,\vmem,\tmem})}{\epsilon}
        }{app:tsemantic:spec:alloc}
        \typerule{Jmp-Spec}{  
            p(\abrho(\pc)) = \pjmp{l} 
            & \abrho_1 = \updatemap{\abrho}{\pc}{l}
        }{
            \spectran{(p,\abconf{})}{(p,\conf{\abrho_1,\abmu,\vmem,\tmem})}{\epsilon}
        }{app:tsemantic:spec:jump}
        \caption{Abstract Speculative Semantics of \abconf{}.}
        \label{fig:tsemantics:spec:ab}
    \end{center}
\end{figure*}    

\Cref{tr:app:tsemantic:spec:assign}, \Cref{tr:app:tsemantic:spec:load}, \Cref{tr:app:tsemantic:spec:store}, \Cref{tr:app:tsemantic:spec:conditional-assign}, \Cref{tr:app:tsemantic:spec:jump} and \Cref{tr:app:tsemantic:spec:alloc} is identical to their sequential versions.
\Cref{tr:app:tsemantic:spec:fence} sets $\pc$ to $\bot$, blocking the execution.
\Cref{tr:app:tsemantic:spec:branch-true} and \Cref{tr:app:tsemantic:spec:branch-false} represent the branch taken and not taken scenarios, respectively. 
Unlike abstract sequential semantics, abstract speculative semantics do not use branch conditions to constrain variables. 
This is because, even if the branch condition is false, the program can still speculatively execute with the taken branch, and vice versa.


\section{Soundness of Abstract Interpretation}\label{app:sec:sound}


Before the discussion, we introduce some notations used in our analysis.
For a domain $V$, \ubound{V} and  \lbound{V} denote the least upper bound and the greatest lower bound operator in $V$, respectively. 
\topof{V} and \botof{V} denote the top and the bottom element of the lattice, respectively.
For operator \op defined in \muasm and a domain $V$, \bopondomain{\op}{V}{a}{b} denotes the operations in $V$.
\order{V} denotes the partial order of the lattice.
For a pair of abstract domain and concrete domain $(A,C)$, we denote $\abf{A}$ as the abstract function and $\conff{A}{}$ as the concretization function.

The abstract domains and their corresponding concrete domains employed in our analysis are as follows:

\begin{itemize}
    \item   Taint Tracking.\\
            Concrete Domain: \ctd. \ctd is the powerset of taint vectors. Each element $T\in \ctd$ represents a possible set of taint vectors associated with a variable.  
            \order{\ctd} is defined as an inclusion relation between sets.\\
            Abstract Domain: \atd. \atd is the product of \abd{T}.
            \order{\ctd} is derived from \order{\abd{T}}.
    \item   Value.\\
            Concrete Domain: \cvd. \cvd is the powerset of \consn-bit integers. $Z\in \cvd$ represents a possible set of integers associated with a variable. 
            \order{\cvd} is defined as an inclusion relation between sets. \\
            Abstract Domain: 
            The abstract interpretation of values is constructed hierarchically using multiple abstract domains: the interval domain (\interval), the disjoint interval set domain (\di) and the abstract value domain (\avd) introduced in \Cref{app:sec:vd}.
            \order{\interval} is given as the standard inclusion relation.
            \order{\di} and \order{\avd} is introduced in \Cref{app:sec:vd}.
    \item   State of executing a particular instruction.\\
            Concrete Domain: The concrete domain consists of program state sets \stateset that satisfies the property that for any $s_1,s_2\in\stateset$, $\rho_{s_1}(\pc)=\rho_{s_2}(\pc)$ holds. 
            Each element in this concrete domain represents a set of possible states that the program can be in when it executes a particular instruction.
            The partial order of this domain is defined as an inclusion relation between sets.\\
            Abstract Domain: \aconf, the domain of abstract states. Each element $\abs\in\aconf$ ia a quaternion \abconf{}, introduced in \Cref{app:sec:ass}. 
            From here on, we will refer to the quaternion as an \textit{abstract state}. 
            \order{\aconf} is defined by requiring that for any $\abs_1,\abs_2\in \aconf$, $\abs_1\order{\aconf}\abs_2$ holds iff $\order{}$ holds for any element pairs in \abrho, \abmu, \vmem and \tmem. 
    \item   Program State.\\
            Concrete Domain: \cstate. \cstate is the powerset of all states.
            Each element in this concrete domain represents a set of possible states of programs.
            \order{\cstate} is defined as an inclusion relation between sets.\\
            Abstract Domain: \astate.   \astate is call the \textit{abstract configuration} domain.
            Each element $\abS\in\astate$ (i.e., an abstract configuration) denotes a mapping from $\Nat$ to \aconf.
            $\abS(i)$ represents the possible abstract states of the program when it reaches instruction $p(i)$.
            Thus $\abrho_{\abS(i)}(\pc)$ is required to be $i$.
            \order{\astate} is defined by requiring that $\abS_1 \order{\astate} \abS_2$ holds iff $\abS_1(i) \order{\astate} \abS_2(i)$ for any $i\in p$, where $i\in p$ denotes that $p(i)$ is a valid instruction.
            \abS is called an initial abstract configuration when $\abS(i)=\botof{\astate}$ for any $i\neq 0$. 
    \item   Observation.\\
            Concrete Domain: \cob. The observation in our work takes the form of a value with a taint label vector.
            In our analysis, we only care about the taint label of an observation. Therefore, the concrete domain of observations is the same as the concrete domain of taint tracking (i.e., \ctd).\\
            Abstract Domain: \aob. For the same reason, the abstract domain of observations is the same as the abstract domain of taint tracking (i.e., \atd).
\end{itemize}

For taint tracking, let $\abf{\abd{\lattice{T}}}$ and $\conff{\abd{\lattice{T}}}{}$ denote the isomorphic functions between \lattice{T} and \abd{\lattice{T}}.
For \ctd and \atd, the abstract and concretization function is given by the following formula.
Let $\ctde\in\ctd$ and $\atde\in\atd$, 
\begin{gather*}
    \abf{\atd}(\ctde) = (\atde_{\consn-1},\atde_{\consn-2},\cdots,\atde_0)\ \ \\\text{where $\atde_i=\ubound{\lattice{T}}\{t[i]\mid t\in \ctde\}$ for $0\le i \le \consn -1$}\\
    \conff{\atd}{}(\atde) = \{t \mid t[i] \in \atde[i]\}
\end{gather*}

It is straightforward to show that $\abf{\atd}$ and $\cof{\atd}$ are monotonic.

\begin{lemma}[Galois Connection Between \abf{\atd} and \conff{\atd}{}]\label{lemma:gc:td}
    For any $\ctde\in\ctd$, we have
    $$
        \ctde \order{\ctd} \cof{\atd}(\abf{\atd}(\ctde))
    $$
\end{lemma}

\begin{lemma}[Local Soundness of Taint Domain]\label{lemma:ls:td}
    For any $\ctde_1,\ctde_2\in\ctd$ and operators \op, we have 
    $$\abf{\atd}(\bopondomain{\op}{\ctd}{\ctde_1}{\ctde_2}) \order{\ctd} \bopondomain{\op}{\atd}{\abf{\atd}(\ctde_1)}{\abf{\atd}(\ctde_2)}$$
\end{lemma}

The proofs of \Cref{lemma:gc:td} and \Cref{lemma:ls:td} are straightforward given the isomorphism between \lattice{T} and \abd{\lattice{T}}.

The abstract and concretization functions of value domain can be more intricate because they depend on the program's execution paths. 
In this case, it is more efficient to discuss this directly on the state domain. 
When discussing the state domain, we use \stateset instead of \cstate as the concrete domain for analysis. 
Without causing confusion, given the program trace $\tau$, we still use the $\abff{\aconf}{\tau}$ to represent the abstract function from \stateset to \aconf, and the concretization function $\conff{\aconf}{\tau}$ remains unchanged (i.e., from \aconf to \cstate).
Therefore, the Galois connection condition can be written as $s\in \conff{\aconf}{\tau}(\abff{\aconf}{\tau}(s))$, and the local soundness can be written as $\abff{\aconf}{\tau}(f(s))\sqsubseteq \abd{f} (\abff{\aconf}{\tau}(s))$.

Let us first discuss state transitions under abstract speculative semantics.

Given a program $p$, let
$\tau=(p, s_1)\xrightarrow[d_1]{o_1} (p, s_2)\cdots \xrightarrow[d_{n-1}]{o_{n-1}}(p, s_n)$ be a concrete speculative trace of states, and $\abd{\tau}=(p, \abd{s_1})\xrightarrow[]{\abd{o_1}} (p, \abd{s_2})\cdots \xrightarrow[]{\abd{o_{n-1}}}(p, \abd{s_n})$ be an abstract speculative trace of abstract states.

Note that we require all memory spaces to be allocated using the $\allocKywd$ instructions. 
However, considering that some values may already be stored in memory in the initial state of the program (we denote the set of such addresses as \initmem.), we add some $\allocKywd$ instructions before the start of the program to represent the allocation of such memory. 
It is important to note that these values already exist in memory in the initial state, and these $\allocKywd$ instructions are only formal placeholders used to assign a symbol to the base of these already used addresses.
Unlike real $\allocKywd$ instructions, which can represent multiple address regions, each of these formal $\allocKywd$ instructions corresponds to a single address region.

The set of all symbols used to represent base addresses in the program $p$ is denoted by $\base_p=\{i\mid p(i) = \palloc{x}{n}\}$.
Given the symbols, memory $\memory^V$ is formalized by a tuple $\conf{\memmap^V, \sizemap^V}$.
$\memmap^V: (\base_p\times \integer)\to V$ maps a memory base and an offset index to a value in $V$, where $V$ can be \VD or \abdpl{T}{n}.
$\sizemap^V: \base_p \to \Nat$ records the memory region size corresponding to each base address.

Given the concrete trace $\tau$, there is a corresponding $\basemap_\tau:\base_p\to \powerset{\Nat}$, which records the concrete addresses of each base.
Each allocated address can be represented as its base address plus an in-bounds offset.
Thus there is a corresponding function $\abaddr{\tau}:\Nat\to (\base_p\times\integer)$ that maps each address to its abstract interpretation.
For unallocated addresses $n$, we can choose a symbol in $\base_p$ that represents the largest base $b$ allocated and use an out-of-bounds offset relative to $b$ to interpret $n$.
Note that such a representation is unique when the trace $\tau$ is given.

Then, we have the concretization functions:
\begin{align*}
    &\conff{\VD}{\tau}(\vd) = \\&\{i\mid \indi{i}{\vd(\emptysym)}\}\cup \{ n + m \mid b\in \base_p, n \in \basemap_\tau(b), \indi{m}{\vd(b)}\}\\&
    \conff{\aconf}{\tau}(\abs) = \{s\mid  \ \rho_s(\pc) = \rho_{\abs}(\pc),\ \rho_s(\mem) = \rho_{\abs}(\mem),\\&
    \ \text{for $x\in\Reg/\{\pc,\mem\}$},\\& \text{$\rho_s(x)\in\conff{\VD}{\tau}(\abrho_{\abs}(x))$ and $\mu_s(x)\in \cof{\atd}(\abmu_{\abs}(x))$,}\\&
    \ \text{for $n\in\Nat$,}\\&\text{$\rho_s(n)\in\conff{\VD}{\tau}(\vmem_{\abs}(\abaddr{\tau}(n)))$ and $\mu_s(n)\in \cof{\atd}(\tmem_{\abs}(\abaddr{\tau}(n)))$ }\}
\end{align*}

By the definition, $\conff{\VD}{\tau}$ and $\conff{\aconf}{\tau}$ are both monotonic.

For \conff{\VD}{\tau} and \abaddr{\tau}, we have, 
\begin{lemma}\label{lemma:sound:memory}
    Let $\vd\in\VD$ and $n\in\Nat$.
    Given an abstract memory $\memory^V=\conf{\memmap, \sizemap}$ on domain $V$, for any $n\in \conff{\VD}{\tau}(\vd)$ where $\vd\in\VD$,  we have
    $$
    \memmap(\abaddr{\tau}(n)) \order{V} \abload{\memory^V}(\vd)
    $$
\end{lemma}
\begin{proof}
    Let $(b,z)=\abaddr{\tau}(n)$, where $b\in \base_p$ and $z\in\integer$.

    If $\indi{z}{\vd(b)}$,  the conclusions is trivial.

    If $z\not\vdash\vd(b)$,  then $n$ is represented by an out-of-bounds base-offset pair in $\vd$.
    Thus, $\memmap(\abaddr{\tau}(n)) \order{V} \topof{V} = \abload{\memory^V}(\vd)$.
    The conclusion also holds.
\end{proof}

A policy $P$ specifies which registers and memory addresses in \initmem contain data that will be marked as \textit{public}.
An initial state $s$ satisfying the policy $P$ is a state where the $\pc$ and $\mem$ evaluates to $0$, and  for $v\in P$, $\mu_{s}(v)=\vec{\tlow}$ and for $v\notin P$, $\mu_{s}(v)=\vec{\thigh}$.

\begin{definition}[Corresponding Initial Abstract State]\label{def:cor-abconstate}
    Let $s$ be an initial concrete state of trace $\tau$. 
    We call an abstract state $\abs$ the corresponding initial abstract state of $s$ when
\begin{enumerate}
    \item $\abrho_{\abs_i}(\pc)=0$ and $\abrho_{\abs_i}(\mem)=0$.
    \item For any $x\in\Reg$, $\abmu_{\abs}(x)=\mu_{s}(x)$.
    \item For any $n\in \initmem$, $\memmap^{\atd}(\abaddr{\tau}(n))=\mu_s(n)$.
    \item For any $x\in\Reg/\{\pc,\mem\}$, $\abrho_{\abs}(x)=\allinterger$, where \allinterger denotes such an abstract value $\vd$ such that $\vd(\emptysym)=[\minI,\maxI]$ and $\vd(v)=\emptyset$ for $v\in\base_p$. 
    \item For any $n\in \initmem$, its corresponding abstract memory address is set to $\allinterger$, i.e., $\memmap^{\avd}(\abaddr{\tau}(n))=\allinterger$.
\end{enumerate}
\end{definition}

The first rule ensures it is an initial abstract state.
The next two rules ensure that the taint vectors of the abstract state $\abs$ can correctly approximate the taint vectors of the concrete state $s$. 
The last two rules take into account all possible values of $s$ in registers and initial memory.
By the definition, we have 
\begin{lemma}\label{lemma:cor-init-conf}
    Each initial state has a unique corresponding initial abstract state.
\end{lemma}

\begin{definition}[Corresponding Speculative Abstract State Trace]\label{def:cor-abcontrace}
    We call a trace of abstract states $\abconftrace$
    the corresponding speculative abstract state trace of a concrete state trace $\tau=(p, s_1)\xrightarrow[d_1]{o_1} (p, s_2)\cdots \xrightarrow[d_{n-1}]{o_{n-1}}(p, s_n)$
    when
    \begin{enumerate}
        \item $\abs_1$ is the corresponding initial abstract state of $s_1$.
        \item For $1\le i\le n$, $\abrho_{\abs_i}(\pc)=\rho_{s_i}(\pc)$.
    \end{enumerate}
\end{definition}

\begin{lemma}\label{lemma:cor-trace-conf}
    Each concrete state trace $\tau$ has a unique corresponding speculative abstract state trace.
\end{lemma}
\begin{proof}
    Let \conctrace be a concrete state trace.
    We can construct  an abstract state trace using mathematical induction.
    Let $\abs_1$ be the corresponding initial abstract state (such abstract state can be uniquely determined by \Cref{lemma:cor-init-conf}), confirming that the conclusion holds for $n=1$.
    Suppose the second requirement holds for $n=k$, we discuss the classification based on the value of pc for the case of $n=k+1$.

    If $p(\rho_{s_{k}}(\pc))$ is not a $\jzKywd$ instruction, then there is a unique rule for $p(\rho_{s_{k}}(\pc))$ in both concrete semantics (\Cref{fig:tsemantics}) and abstract speculative semantics (\Cref{fig:tsemantics:spec:ab}).
    By applying the corresponding rules, we obtain $\abs_{k+1}$.

    If $p(\rho_{s_{k}}(\pc))$ is a $\jzKywd$ instruction. 
    Suppose $p(\rho_{s_{k}}(\pc)) = \pbranch{x}{l}$, then $\rho_{s_{k+1}}(\pc)$ can be either $\rho_{s_{k}}(\pc)+1$ or $l$.
    When $\rho_{s_{k+1}}(\pc) = \rho_{s_{k}}(\pc)+1$, we apply \Cref{tr:app:tsemantic:spec:branch-false} to get $\abs_{k+1}$; otherwise, we apply \Cref{tr:app:tsemantic:spec:branch-true}.
    In both cases we have $\rho_{s_{k}}(\pc)=\abrho_{\abs_{k+1}}(\pc)$, which implies that the second requirement holds for $n=k+1$.

    As can be seen from the construction, $\abs_{k+1}$ is uniquely determined when $s_{k+1}$ and $\abs_{k}$ is given. 
    Therefore, we obtain the unique corresponding abstract state trace of $\tau$. 
\end{proof}

With \Cref{lemma:cor-trace-conf}, we can define the abstract function \abff{\aconf}{\tau} by requiring $\abff{\aconf}{\tau}(s_i)=\abs_i$, where \abconftrace is the corresponding speculative abstract state trace of \conctrace.

\begin{lemma}[Local Soundness of Interval]\label{lemma:ls:interval}
    Let $\cof{\interval}$ be the concretization function.
    Given $z_1,z_2\in \integer$ and $I_1,I_2\in\interval$ s.t. $z_1\in I_1$ and $z_2\in I_2$, then for  $\op\in\bop\cup\uop$ we have
    $$
    \bopondomain{\op}{\integer}{z_1}{z_2}\in \cof{\interval}(\bopondomain{\op}{\interval}{I_1}{I_2})
    $$
\end{lemma}
\begin{proof}
    For $\op\in\uop\cup\bop/\{\orop,\andop\}$, the operation rules are standard and the proofs are straightforward.

    For $\op=\andop$, we only consider the case where $I_1$ does not contain negative integers (thus $z_1$ is a non-negative integer).
    \begin{itemize}
        \item If $I_2$ does not contain negative integers, $z_2$ is a non-negative integer.
        
        Considering that the \andop operation can turn certain $1$ bits in the operands to $0$ but never turn $0$ bits to $1$, we have $\bopondomain{\andop}{\integer}{z_1}{z_2} \le \min(z_1,z_2)$.
        This further implies $\bopondomain{\andop}{\integer}{z_1}{z_2}\in\cof{\interval}(\bopondomain{\andop}{\interval}{I_1}{I_2})$.
        \item If $I_2$ does not contain non-negative integers, $z_2$ is a negative integer.
         
         $\bopondomain{\andop}{\integer}{z_1}{z_2}$ shares the same sign bit with $z_1$.
         Similarly, zeros in $z_2$  will clear the corresponding ones in $z_1$, thus $\bopondomain{\andop}{\integer}{z_1}{z_2} \le z_1-z_2$.
         This further implies $\bopondomain{\andop}{\integer}{z_1}{z_2}\in\cof{\interval}(\bopondomain{\andop}{\interval}{I_1}{I_2})$.
         \item If $I_2$  contains both negative and  non-negative integers, the conclusion is straightforward.
    \end{itemize}

    The case of $\op=\orop$ can be proven using a similar approach.
\end{proof}

The local soundness of disjoint interval set can be derived from \Cref{lemma:ls:interval}.
\begin{lemma}[Local Soundness of Disjoint Interval Set]\label{lemma:ls:dis}
    Let $\cof{\di}$ be the concretization function.
    Given $z_1,z_2\in \integer$ and $d_1,d_2\in\interval$ s.t. $\indi{z_1}{d_1}$ and $\indi{z_2}{d_2}$, then for $\op\in\bop\cup\uop$ we have
    $$
    \bopondomain{\op}{\integer}{z_1}{z_2}\in\cof{\di}(\bopondomain{\op}{\di}{d_1}{d_2})
    $$
\end{lemma}

Furthermore, we have the local soundness of the abstract value domain.
\begin{lemma}[Local Soundness of Value Domain]\label{lemma:ls:vd}
    Let $\confvd$ be the concretization function.
    Given $z_1,z_2\in\integer$ and $\vd_1,\vd_2\in\VD$ s.t. $z_1\in\confvd(\vd_1)$ and $z_2\in\confvd(\vd_2)$
    $\op\in\bop\cup\uop$, we have $$\bopondomain{\op}{\integer}{z_1}{z_2}\in \confvd(\bopondomain{\op}{\VD}{\vd_1}{\vd_2})$$
\end{lemma}

By \Cref{lemma:ls:vd} and the definition of \confvd, we have
\begin{lemma}\label{lemma:expr-vd}
    For $\rho:\Reg\to\integer$ and $\abrho:\Reg\to\avd$, let \confvd be a concretization function.
    If $\rho(x)\in \confvd(\abrho(x))$ holds for any $x\in\Reg$, then for any expression $e$, 
    $\exprEval{e}{\rho}\in \confvd(\exprEval{e}{\abrho})$.
\end{lemma}
The proof can be derived by using mathematical induction to the length of $e$.

Similarly, by \Cref{lemma:ls:td} we have
\begin{lemma}\label{lemma:expr-td}
    For $\mu:\Reg\to\tdn$ and $\abmu:\Reg\to\atd$, let $\cof{\atd}$ be a concretization function.
    If $\mu(x)\in \cof{\atd}(\abmu(x))$ holds for any $x\in\Reg$, then for any expression $e$, 
    $\exprEval{e}{\mu}\in \cof{\atd}(\exprEval{e}{\abmu})$.
\end{lemma}

Combining \Cref{lemma:ls:td} and \Cref{lemma:ls:vd}, we have the local soundness of abstract state's speculative transitions.

\begin{lemma}[Local Soundness of Abstract State's Speculative Transition]\label{lemma:ls:spec}
    Let \conctrace be a concrete state trace and \abconftrace be its corresponding speculative abstract state trace.
    For $1\le i \le (n-1)$, we have $$s_i\in\confs(\abs_i) \implies s_{i+1}\in\confs(\abs_{i+1})$$
\end{lemma}
\begin{proof}
    We proceed by case distinction on the transition rules defined in \Cref{fig:tsemantics} and \Cref{fig:tsemantics:spec:ab}.
    Since the values of $\pc$ and $\mem$ are already determined by \Cref{def:cor-abcontrace}, the conclusion holds naturally for \Cref{tr:app:tsemantic:branch-force}, \Cref{tr:semantic:branch-step}, \Cref{tr:app:tsemantic:jump} and \Cref{tr:app:tsemantic:fence}.
    For the remaining cases, 
\begin{enumerate}
    \renewcommand{\labelenumi}{\textbf{RULE}}
    \item \Cref{tr:app:tsemantic:assign}. Suppose $p(\rho(\pc)) = \passign{x}{e} $.  
    Then the transition \spectran{\absi}{\absione}{\absoi} is derived by applying \Cref{tr:app:tsemantic:spec:assign}.
    We have 
    \begin{align*}
        \rho_{s_{i+1}}(x) & = \exprEval{e}{\rho} \tag*{(By \Cref{tr:app:tsemantic:assign})}\\
                    & \in \confvd(\exprEval{e}{\abrho_{s_i}}) \tag*{(By $s_i\in\confs(\abs_i)$ and \Cref{lemma:expr-vd})}\\
                    & = \confvd(\abrho_{s_{i+1}}(x)) \tag*{(By \Cref{tr:app:tsemantic:spec:assign})}
    \end{align*}
    
    Similarly,  $\mu_{s_{i+1}}(x)\in \cof{\atd}(\abmu_{s_{i+1}}(x))$ holds.
    Thus $s_{i+1}\in\confs(\abs_{i+1})$.
    \item \Cref{tr:app:tsemantic:load}. Suppose $p(\rho(\pc)) = \pload{x}{e}$. Then the transition \spectran{\absi}{\absione}{\absoi} is derived by applying \Cref{tr:app:tsemantic:spec:load}. 
    We have
    \begin{align*}
        \rho_{s_{i+1}}(x) & =  \rho_{s_{i+1}}(\exprEval{e}{\rho_{s_i}})  \tag*{(By \Cref{tr:app:tsemantic:load})}\\
                        & \in \{\rho_{s_{i+1}}(n)\mid n \in \confs(\exprEval{e}{\abrho_{s_i}})\} \tag*{(By \statepremise and \Cref{lemma:expr-vd})}\\
                        & \subseteq \bigcup_{n \in \confvd(\exprEval{e}{\abrho_{s_i}})} \confvd(\vmem_{\abs_{i}}(\abaddr{\tau}(n))) \tag*{(By \statepremise)}\\
                        & \subseteq \confvd(\abload{\vmem_{\abs_{i}}}(\exprEval{e}{\abrho_{s_i}})) \tag*{(By \Cref{lemma:sound:memory})}\\
                        & = \confvd(\abrho_{\abs_{i+1}}(x)) \tag*{(By \Cref{tr:app:tsemantic:spec:load})}
    \end{align*}
    For taint tracking, if $\vec{\thigh}\in \exprEval{e}{\mu_{s_i}}$, we have $\abvechigh\in\exprEval{e}{\abmu_{s_i}}$.
    Then, $\mu_{s_{i+1}}(x)\in \cof{\atd}(\abvechigh) = \cof{\atd}(\abmu_{\abs_{i+1}}(x))$.

    If $\vec{\thigh}\notin \exprEval{e}{\mu_{s_i}}$, similar to value domain, we have 
    \begin{align*}
        \mu_{s_{i+1}}(x)  &=  \mu_{s_{i+1}}(\exprEval{e}{\rho_{s_i}})  
                         \in \{\mu_{s_{i+1}}(n)\mid n \in \confs(\exprEval{e}{\abrho_{s_i}})\} 
                         \\&\subseteq \bigcup_{n \in \confvd(\exprEval{e}{\abrho_{s_i}})} \cof{\atd}(\tmem{\abs_{i}}(\abaddr{\tau}(n)))
                         \\&\subseteq \cof{\atd}(\abload{\tmem_{\abs_{i}}}(\exprEval{e}{\abrho_{s_i}})) 
                         = \confvd(\abrho_{\abs_{i+1}}(x)) 
    \end{align*}
    Thus $s_{i+1}\in\confs(\abs_{i+1})$.
    \item \Cref{tr:app:tsemantic:store}. Suppose $p(\rho(\pc)) = \pstore{x}{e}$. Then the transition \spectran{\absi}{\absione}{\absoi} is derived by applying \Cref{tr:app:tsemantic:spec:store}. 
    
    Let $k = \exprEval{e}{\rho_{s_i}}$, and $(b,z)=\abaddr{\tau}$ where $b\in \base_p$ and $z\in \integer$.
    For $k'\in\Nat$ where $k' \neq k$, we have
    $$
    \rho_{s_{i+1}}(k')=\rho_{s_{i}}(k')\in \conff{\VD}{\tau}(\vmem_{\abs_{i}}(\abaddr{\tau}(k))) \subseteq \conff{\VD}{\tau}(\vmem_{\abs_{i+1}}(\abaddr{\tau}(k)))
    $$

    For $\rho_{s_{i+1}}(k)$, if $z\not\vdash \abrho_{s_i}(x)(b)$, then $k$ is represented as an out-of-bounds base-offset pair in $\abrho_{s_i}(x)$.
    Therefore, $\abstore{\vmem_{\abs_{i}}}(\exprEval{e}{\abrho_{s_i}}, \abrho_{s_i}(x))$ will set $\vmem_{\abs_{i+1}}(b,z)$ to \topof{\VD}. Then we have $\rho_{s_{i+1}}(k)\in \conff{\VD}{\tau}(\vmem_{\abs_{i+1}}(\abaddr{\tau}(k)))$.
    
    If $z\vdash \abrho_{s_i}(x)(b)$, We have
    \begin{align*}
        \rho_{s_{i+1}}(k) & =  \rho_{s_i}(x) \tag*{(By \Cref{tr:app:tsemantic:store})}\\
                        & \in \confvd(\abrho_{s_i}(x)) \tag*{(By \statepremise )}\\
                        & \subseteq \confvd(\abrho_{s_i}(x) \ubound{\VD} \vmem_{\abs_{i}}(b,z)) \tag*{(By the definition of \ubound{\VD} and the monotonicity of \confvd)}\\
                        & \subseteq \confvd(\vmem_{\abs_{i+1}}(b,z)) \tag*{(By the definition of \abstore{\vmem_{\abs_{i+1}}} and  \Cref{tr:app:tsemantic:spec:store})}
    \end{align*}
    Similarly,  $\mu_{s_{i+1}}(k)\in \cof{\atd}(\vmem_{\abs_{i+1}}(\abaddr{\tau}(k)))$ holds.
    Thus $s_{i+1}\in\confs(\abs_{i+1})$.
    \item \Cref{tr:semantic:conditional-assign}. Suppose $p(\rho(\pc)) = \pcondassign{x}{e}{e'}$. Then the transition \spectran{\absi}{\absione}{\absoi} is derived by applying \Cref{tr:app:tsemantic:spec:conditional-assign}.  We have
    \begin{align*}
    \rho_{s_{i+1}}(x)& \in \{\rho_{s_{i}}(x), \exprEval{e}{\rho_{s_i}}\} \tag*{(By \Cref{tr:semantic:conditional-assign})}\\
                    & \subseteq \confvd(\abrho_{\abs_i}(x))\cup \confvd(\exprEval{e}{\abrho_{\abs_i}}) \tag*{(By \statepremise and \Cref{lemma:expr-vd})}\\
                    & \subseteq \confvd(\abrho_{\abs_i}(x) \ubound{\VD} \exprEval{e}{\abrho_{\abs_i}}) \tag*{(By the monotonicity of \confvd)}\\
                    & = \confvd(\abrho_{\abs_{i+1}}(x)) \tag*{(By \Cref{tr:app:tsemantic:spec:conditional-assign})}
    \end{align*}
    Similarly,  $\mu_{s_{i+1}}(x)\in \cof{\abt}(\abmu{s_{i+1}}(x))$ holds.
    Thus $s_{i+1}\in\confs(\abs_{i+1})$.
    \item \Cref{tr:app:tsemantic:alloc}. Suppose $p(\rho(\pc)) = \pcondassign{x}{e}{e'}$. Then the transition \spectran{\absi}{\absione}{\absoi} is derived by applying \Cref{tr:app:tsemantic:spec:alloc}.  We have
    $$
        \mu_{s_{i+1}}(x) = \vec{\tlow} \in \cof{\atd}(\abveclow) = \cof{\atd}(\abrho_{\abs_{i+1}}(x))
    $$
    Thus $s_{i+1}\in\confs(\abs_{i+1})$.
\end{enumerate}
\end{proof}

\begin{theorem}[Global Soundness of Abstract State's Speculative Transition]\label{thm:gs:acst}
    Let \conctrace be a state trace and \abconftrace be the corresponding  speculative abstract state trace.
    For $1\le i \le n$, we have $s_i\in\confs(\abs_i)$.
\end{theorem}
\begin{proof}
    By \Cref{def:cor-abconstate}, we have $s_1\in\confs(\abs_1)$.

    By the mathematical induction, we further have for $1\le i \le n$,  $s_i\in\confs(\abs_i)$.
\end{proof}

We have now obtained the global soundness of the abstract state trace. 
Based on this, we can further obtain the global soundness of the abstract configuration.
Similarly to the abstract state domain, we can define the concretization function of abstract configuration as:
$$
\conff{\astate}{\tau}(\abS)=\bigcup_{i\in p} \conff{\aconf}{\tau}(\abS(i))
$$

We define the predecessors of a location $i$ of the program $p$ as: $$\pred_p(i) = \{ j \in p\mid \text{$j = i - 1$ or $p(j)=\pbranch{x}{i}$ or $p(j)=\pjmp{i}$}\}$$ 
Note that $p(0)$ has no predecessors as it is the entry point of the program.
 
Then we define the speculative transition of an abstract configuration $\abS_1$ as $\spectran{\abS_1}{\abS_2}{}$, where $\abS_2$ is given by:
$$
    \abS_2(i)=\begin{cases}
        \abS_1(0)& i=0\\
        \bigsqcup\limits_{j\in\pred_p(i)}\{\abs_2\mid\spectran{\abS_1(j)}{\abs_2}{\abo}, \abrho_{\abs_2}(\pc)=i\}&i\neq 0
    \end{cases}
$$
This formula calculates the abstract states that are reached from each $p(i)$'s predecessor's abstract state, and computes the least upper bound of these abstract states as $p(i)$'s new abstract state.
It is obvious that $\spectran{}{}{}$ is a monotonic operator on \astate and $\abS_2$ is uniquely determined by $\abS_1$.

Then for a speculative trace of abstract configurations \abstatetrace where $\abS_1$ is an initial abstract configuration, since $\abS_1\order{\astate}\abS_2$,  $\{\abS_i\}$ is a monotonic list of abstract configurations.
Given that \astate is a finite domain, there exists an $i>0$ such that $\abS_i=\abS_j$ holds for all $j\ge i$, meaning that $\abS_i$ is a fixpoint of the operator $\spectran{}{}{}$.
We denote the fixpoint of \abd{\Pi} by $\fix^{\text{spec}}(\abS_1)$.

And we can define the \textit{corresponding initial abstract configuration} of an initial concrete state $s$.

\begin{definition}[Corresponding Initial Abstract Configuration]\label{def:cor-abstate}
    Let $s$ be an initial concrete state of trace $\tau$. 
    We call an initial abstract configuration $\abS$ a corresponding initial abstract configuration of $s$ when $\abfs(s)\order{\aconf}\abS(0)$.
\end{definition}

Now we have the soundness of our abstract speculative semantics.

\begin{theorem}[Soundness of Abstract Speculative Semantics]\label{thm:gs:spec}
    Let \conctrace be a speculative abstract state trace, $\abS$ be a corresponding initial abstract configuration of $s_1$, and $\conff{\astate}{\tau}$ be the concretization function. 
    We have 
    $$
    s_i\in \conff{\astate}{\tau}(\fixspec(\abS))
    $$
\end{theorem}
\begin{proof}
    Let \abconftrace be the corresponding abstract state trace of $\tau$, and \abstatetrace be the trace derived from $\abS_1$.
    We first proof that $\abs_i\order{\aconf}\abS_i(\abrho_{\abs_i}(\pc))$ holds for all $i\ge 0$.
    
    The proof can be done by the mathematical induction.
    For $i=1$, $\abs_1\order{\aconf}\abS_1(\abrho_{\abs_1}(\pc))$ is given by the fact that $\abS$ is a corresponding initial abstract configuration of $s_1$.
    Suppose for $1\le i\le k$,  $\abs_i\order{\aconf}\abS_i(\abrho_{\abs_i}(\pc))$ holds.
    Then for $i=k+1$, we have, 
    \begin{align*}
        \abs_{k+1} & \order{\aconf} \bigsqcup\{\abs\mid \spectran{\abs_k}{\abs}{\abo}\}\tag*{(By  $\abs_{k+1}\in \{\abs\mid \spectran{\abs_k}{\abs}{\abo}\}$) }\\
        & \order{\aconf} \bigsqcup\{\abs\mid \spectran{\abS_k(\abrho_{\abs_k}(\pc))}{\abs}{\abo}\}\tag*{(by the inductive hypothesis and the monotonicity of \spectran{}{}{})}\\
        & \order{\aconf} \bigsqcup_{j\in\pred_p(\abrho_{\abs_{k+1}}(\pc))}\{\abs\mid \spectran{\abS_k(j)}{\abs}{\abo}\}\tag*{(By $\abrho_{\abs_{k}}(\pc)\in\pred_p(\abrho_{\abs_{k+1}}(\pc))$)}\\
        & = \abS_{k+1}(\abrho_{\abs_{k+1}}(\pc))
    \end{align*}
    Therefore, $\abs_i\order{\aconf}\abS_i(\abrho_{\abs_i}(\pc))$ holds for all $i\ge 0$.

    Then by \Cref{thm:gs:acst} and the monotonicity of $\{\abS_i\}$, we have
    $$
    s_i\in \conff{\astate}{\tau}(\abs_i) \subseteq \conff{\astate}{\tau}(\abS_i(\abrho_{\abs_i}(\pc)))\subseteq \conff{\astate}{\tau}(\fixspec(\abS))
    $$
\end{proof}

Similarly, we can define the sequential transition of a abstract configuration $\abS_1$ as $\seqtran{\abS_1}{\abS_2}{}$. 
$\fix^{\text{seq}}(\abS_1)$ denotes the fixpoint of a sequential trace of abstract configurations starting from $\abS_1$.
Finally, we can also establish the soundness of abstract speculative semantics.

\begin{theorem}[Soundness of Abstract Sequential Semantics]\label{thm:gs:seq}
    Let \conctraceseq be a sequential abstract state trace, $\abS$ be a corresponding initial abstract configuration of $s_1$, and $\conff{\astate}{\tau}$ be the concretization function. 
    We have 
    $$
    s_i\in \conff{\astate}{\tau}(\fixseq(\abS))
    $$
\end{theorem}






\section{\LightSLH}\label{app:LightSLH}

LightSLH  operates in three phases. 
For the first phase, LightSLH performs abstract interpretation using abstract sequential semantics.
We use a mapping $\seqfinal(i)=\conf{\abrho,\abmu}$ to denote the maximum abstract configuration the program $p$ can be after executing $p(i)$.

For the second phase, we present rules for ``utilizing the result of the first phase'' in \Cref{fig:switching}.

\begin{figure*}[h]
    \begin{center}
        \typerule{Ld-Switch}{
        n = \abrho(\pc) & p(n) = \pload{x}{e}& \vd = \exprEval{e}{\abrho_{\seqfinal(n)}} & \abt = \exprEval{e}{\abmu_{\seqfinal(n)}}\\
        \abrho_1=\abrho[x\mapsto \abrho_{\seqfinal(n+1)}(x), \pc\mapsto n+1] &
        \abmu_1=\abmu[x\mapsto \abmu_{\seqfinal(n+1)}(x)]
    }{
       \switchtran{(p,\abconf{})}{(p,\conf{\abrho_1,\abmu_1,\vmem,\tmem})}{\Obsload{\vd:\abt}{a}{b}}
    }{app:switch:load}
    \typerule{St-Switch}{
        n = \abrho(\pc) & p(\abrho(\pc)) = \pstore{x}{e} & \vd = \exprEval{e}{\abrho_{\seqfinal(n)}} & \abt = \exprEval{e}{\abmu_{\seqfinal(n)}} \\ 
        \abt_1 = {\begin{cases}
            \abvechigh & \text{if $\abhigh \in \abt$}\\
            \abmu(x)      &\text{if $\abhigh \notin \abt$}
            \end{cases}} & {\begin{aligned}
                \vmem_1 = \abstore{\vmem}(\vd, \abrho(x))\\
                \tmem_1 = \abstore{\tmem}(\vd, \abt_1)\\
                \abrho_1=\rho[\pc\mapsto\abrho(\pc)+1]
            \end{aligned}}
    }{
       \switchtran{(p,\abconf{})}{(p,\conf{\abrho_1,\abmu,\vmem_1,\tmem_1})}{\Obsstore{\vd:\abt}{a}{b}}
    }{app:switch:store}
    \end{center}
    \caption{Rules for Utilizing the Result of Sequential Abstract Interpretation}
    \label{fig:switching}
\end{figure*}

We define a transition operator $\trans: (\astate \times \powerset{\Nat} \times \astate)\to\astate$ to denote the computation of LightSLH's second phase.
$\trans$ takes an abstract configuration, a set of integers representing the program locations to be hardened, and the result of the first phase's analysis as arguments, and returns the next abstract configuration. 
Specifically, for  a program $p$ and $i\in p$, let $\abS'=\trans(\abS,\finalharden,\seqfinal)$, then $\abS'$ is given by
\begin{itemize}
    \item $\abS'(0)=\abS(0)$.
    \item $\abS'(i)=\bigsqcup\limits_{\substack{j\in\pred_p(i)}}\{\abs_2\mid\spectran{\abS_1(j)}{\abs_2}{\abo}, \abrho_{\abs_2}(\pc)=i, j\notin \hac\} \cup \{\abs_2\mid\switchtran{\abS_1(j)}{\abs_2}{\abo}, \abrho_{\abs_2}(\pc)=i, j\in \hac\}$ for $i\neq 0$.
\end{itemize}
where $$\hac = \{k\mid \text{$k\in \finalharden$ and $p(k)$ is a load or store instruction}\}$$.

We define the hardening set of an abstract configuration $\abS$ by $\hardenset(\abS)$, where
$$
\hardenset(\abS)=\{i\mid \text{There exists $\abs\in\aconf$ s.t. \spectran{\abS(i)}{\abs}{\abo} and $\abd{\thigh}\in\abo$}\}
$$
The algorithm of LightSLH's second phase is presented in \Cref{algo:ai-hk}.


        

The convergence of \Cref{algo:ai-hk} is given by (1) for fixed \finalharden and \seqfinal, $\trans(\currentstate, \finalharden, \seqfinal)$ is monotonic, and (2) \finalharden is a monotonically increasing set with an upper bound, thus  stops changing after a finite number of steps.

In the third phase, we use the approach similar to that in \Cref{lst:slh} to harden the instructions identified as requiring hardening in the second place.
Specifically, we utilize a flag to indicate whether the program is in misspeculative execution.
The flag is set to -1 during misspeculative execution and 0 otherwise.
For a program $p$, and a set of locations (denoted by $\finalharden$) where the instructions are marked as requiring hardening, for $i\in \finalharden$, we harden $p(i)$ using the following rules:
\begin{itemize}
    \item If $p(i)=\pload{x}{e}$, then $p(i)$ is hardened to $\pload{x}{e\ \orop\ \texttt{flag}}$.
    \item If $p(i)=\pstore{x}{e}$, then $p(i)$ is hardened to $\pstore{x}{e\ \orop\ \texttt{flag}}$.
    \item If $p(i)=\pbranch{x}{l}$, then $p(i)$ is hardened to $\pbranch{x\ \orop\ \texttt{flag}}{l}$.
\end{itemize}
For brevity, we allow branch instructions to take an expression (i.e., $x\ \orop\ \texttt{flag}$) as the register operand.

Given a policy $P$, we denote the hardened program by $\hardened{p}{P}$.
To facilitate the subsequent discussion, we disregard the instructions in $\hardened{p}{P}$ that compute the speculative flag when numbering the instructions in program $p$.
Consequently, for $i\in p$, the instruction types of $p(i)$ and $\hardenp(i)$ become identical.

To introduce the following theorem, we make  reasonable additions to the semantics in \Cref{fig:tsemantics}: a value loaded from an invalid address (e.g., -1 is considered as an invalid address) is represented by an empty value $\epsilon$ (and a taint vector $\vec{\bot}$), and stores to an invalid address will not change the memory (since such stores never take place).
We further assume any computation with an $\epsilon$ value as the operand results in an $\epsilon$ value.
In particular, we let $\epsilon \in A$ holds for any set $A$ (in other words, we let $\epsilon$ be the symbol representing the concretization of $\emptyset$).

\lightslh*
\begin{proof}
    Let $\tau = (\hardenp,s_1) \xrightarrow[d_1]{o_1} (\hardenp,s_2) \xrightarrow[d_2]{o_2}\cdots$ be a sequential trace of \hardenp.
    Let $\abS_1$ be a corresponding abstract configuration of $s_1$, $\seqfinal$ be the fixpoint of sequential abstract interpretation, and \finalharden be the output set of \Cref{algo:LightSLH} which takes $\abS_1$ and  $\seqfinal$ as inputs.

    Given $\abS_1$,  there is a trace $\abS_1,\abS_2,\cdots$ such that for $i\ge 1$, $\abS_{i+1}=\trans(\abS_i,H,\seqfinal)$.

    We proof that for $i\ge 1$, we have $s_i\in \conff{\aconf}{\tau}(\abS_i(\rho_{s_i}(\pc)))$.\hfill (*)

    (*) holds for $i=1$ naturally. 
    Suppose (*) holds for $i=k$, then we discuss the case for $k+1$. 




    Let $n=\rho_{s_k}(\pc)$ and $n'=\rho_{s_{k+1}}(\pc)$. 
    If $n\notin \hac$, then for $\spectran{s_{k}}{s_{k+1}}{o_k}$, there exists $\abs\in\aconf$ such that $\spectran{\abS_k(n)}{\abs}{\abo}$ and $\abrho_{\abs}(\pc)=n'$.
    By \Cref{lemma:ls:spec}, we have $s_{k+1}\in\confs(\abs)$. Therefore,
    $$
    \begin{aligned}
        s_{k+1}&\in\confs(\abs) \\&\subseteq  \bigcup\limits_{\substack{j\in\pred_p(n')\\j\notin \hac}}  \{s\in\conff{\aconf}{\tau}(\abs)\mid \spectran{\abS_k(j)}{\abs}{\abo},\abrho_{\abs}(\pc)=n'\} \\&\subseteq \conff{\aconf}{\tau}(\abS_{k+1}(n'))
    \end{aligned}
    $$

    If $n(\pc)\in \hac$, let $\abs=\abS_{k}(n)$ and  \switchtran{\abs}{\absprime}{\abd{o'}}.
    \begin{itemize}
        \item If $p(n)=\pload{x}{e}$.
        \begin{itemize}
            \item If $f_{s_k}=\bot$, then $(\hardenp,s_1)\cdots(\hardenp,s_k)$ is a prefix of a sequential trace.   
            By \Cref{thm:gs:seq} and \Cref{tr:app:switch:load} we have
            $
                \rho_{s_{k+1}}(x)\in \confvd(\abrho_{\seqfinal(n)}(x)) = \confvd(\abrho_{\absprime}(x))
            $
            Similarly we have $\mu_{s_{k+1}}(x) \in \cof{\atd}(\abmu_{\absprime}(x))$.
            Thus, combining induction hypothesis, $s_{k+1}\in\confs(\absprime)$ holds.
            \item  If $f_{s_k}=\top$, since $p(n)$ is being hardened and $e$ is evaluated to an invalid address (i.e., $-1$), we have $\rho_{s_{k+1}}(x)= \epsilon \in \confvd(\abrho_{\absprime}(x))$ and $\mu_{s_{k+1}}(x)= \vec{\bot} \in \cof{\atd}(\abrho_{\absprime}(x))$.
            Therefore,  $s_{k+1}\in\confs(\absprime)$ holds.
        \end{itemize}
        \item If $p(n)=\pstore{x}{e}$.
         \begin{itemize}
            \item If $f_{s_k}=\bot$ then $(\hardenp,s_1)\cdots(\hardenp,s_k)$ is a prefix of a sequential trace.
            Let $m=\exprEval{e}{\rho_{s_k}}$.
            Then  
            \begin{align*}
                \rho_{s_{k+1}}(m) & = \rho_{s_k}(x)\tag*{(By \Cref{tr:app:tsemantic:seq:store})}\\
                &\in \confvd({\abrho_{\abs}(x)}) \tag*{(By induction hypothesis)}\\
            &\subseteq \confvd(\abstore{\vmem_{\abs}}(\abaddr{\tau}(\exprEval{e}{\rho_{s_k}}), \abrho_{\abs}(x))(\abaddr{\tau}(m))) \tag*{(By  the definition of \abaddr{\tau} and \abstore{\vmem})}\\
            &\subseteq \confvd(\abstore{\vmem_{\abs}}(\exprEval{e}{\abrho_{\seqfinal(n)}}, \abrho_{\abs}(x))(\abaddr{\tau}(m)))  \tag*{(By \Cref{thm:gs:seq})}\\
            & = \confvd(\vmem_{\absprime}(\abaddr{\tau}(m))) \tag*{(By \Cref{tr:app:switch:store})}
            \end{align*}
            For $l\neq m$ and $l\in \Nat$, we have
            $$
            \begin{aligned}
                \rho_{s_{k+1}}(l) &= \rho_{s_{k+1}}(l) \in\confvd(\vmem_{\abs}(\abaddr{\tau}(l)))\\&\subseteq\confvd(\vmem_{\absprime}(\abaddr{\tau}(l)))
            \end{aligned}
            $$
            Therefore, for any $l\in \Nat$, we have $\rho_{s_{k+1}}(l)\in \confvd(\vmem_{\absprime}(\abaddr{\tau}(l)))$.
            Similarly, $\mu_{s_{k+1}}(l)\in \confvd(\tmem_{\absprime}(\abaddr{\tau}(l)))$ holds for any $l\in\Nat$.
            Consequently, we have $s_{k+1}\in\confs(\absprime)$.
            \item If $f_{s_k}=\top$. Then for $l\in \Nat$, we have
            $$
            \begin{aligned}
                \rho_{s_{k+1}}(l)&=\rho_{s_{k}}(l)\in \confvd(\vmem_{\abs}(\abaddr{\tau}(l))) \\&\subseteq \confvd(\abstore{\vmem_{\abs}}(\exprEval{e}{\abrho_{\seqfinal(n)}}, \abrho_{\abs}(x))(\abaddr{\tau}(l)))\\& = \confvd(\vmem_{\absprime}(\abaddr{\tau}(l)))
            \end{aligned}
            $$
            Similarly, for  $l\in \Nat$, $\rho_{s_{k+1}}(l)\in \confvd(\tmem_{\absprime}(\abaddr{\tau}(l)))$.
            Consequently, we have $s_{k+1}\in\confs(\absprime)$.
        \end{itemize}
    \end{itemize}
    Thus we have $s_{k+1}\in\confs(\absprime)$ in all the cases.
    Considering $\absprime \subseteq \abS_{k+1}(n')$, we have $s_{k+1}\in \confs(\abS_{k+1}(n'))$.
    So (*) holds for all $i\ge 1$.

    Let $\abS'_i$ be a list where $\abS'_1=\abS_1$ and $\abS'_i$ is the value of $\currentstate$ when entering the loop in \Cref{algo:LightSLH} for the i-th time.
    Suppose that the value of \hardenlist remains constant (i.e., equals to \finalharden) starting from the $k$-th iteration of the loop.
    Then for $i\ge k$, we have $\abS'_{i+1}=\trans(\abS'_{i},H,\seqfinal)$.

    Given $\abS_1(0)=\abS'_1(0)=\abS'_k(0)$ and for $j\neq0$, $\abS_1(j)=\botof{\aconf}\order{\aconf}\abS'_k(j)$, we have $\abS_1\order{\astate}\abS'_k$.
    By the monotonicity of $\trans$ (when \seqfinal and \finalharden is fixed), we have 
    $$\abS_i = \trans^i(\abS_1,\finalharden,\seqfinal)  \order{\astate} \trans^i(\abS'_k,\finalharden,\seqfinal) \order{\astate} \fix^\trans(\abS'_k) $$

    Thus, considering the monotonicity of $\hardenset(\abS)$, we have
    $\hardenset(\abS_i)\subseteq \hardenset(\fix^\trans(\abS'_k)) = \finalharden$.\hfill (**)

    Then for $l$ and $ (\hardened{p}{P},s_l)\xrightarrow[d_l]{o_l} (\hardened{p}{P},s_{l+1})$, if $\rho_{s_l}(\pc)\notin\finalharden$, by (*), (**) and the definition of $\hardenset(\abS)$, we have $\thigh\notin t(o_l)$.
    If $\rho_{s_l}(\pc)\in\finalharden$, indicating that $\hardened{p}{P}(\rho_{s_l}(\pc))$ has been hardened, according to our hardening methods, when $f_{s_l}=\top$, we have $\thigh\notin\vec{\tone}=t(o_l)$. 
    Consequently,
    $$\forall i\ge 1, f_{s_i} = \top \Rightarrow\thigh\notin t(o_i) $$

    Proof done.
\end{proof}

\end{document}